%% file: main.tex
\documentclass[onefignum,onetabnum]{siamonline250211}

\input{shared}

\ifpdf
\hypersetup{
  pdftitle={Confidence intervals for functionals in constrained inverse problems via data-adaptive sampling-based calibration},
  pdfauthor={M. Stanley, P. Batlle, P. Patil, H. Owhadi, M. Kuusela}
}
\fi

\begin{document}

\maketitle

\begin{abstract}  
    We address functional uncertainty quantification for ill-posed inverse problems with rank-deficient forward models, known observation noise distributions, and parameter constraints common in scientific applications such as remote sensing and high-energy physics.
    We propose four constraint-aware confidence interval methods built upon a theoretical test inversion framework in \cite{previous_paper} and optimization-based endpoint computations, resulting in intervals that are both computationally feasible and less conservative.
    Our approach first shrinks the potentially unbounded constraint set in a data-adaptive way, obtains samples of the relevant test statistic inside this set to estimate a quantile function, and then uses these computed quantities to produce the constraint-aware confidence intervals.
    Our approach to bounding the constraint set in a data-adaptive way is based on the approach by \cite{berger_boos} and involves defining a subset of the constraint set where the true parameter is guaranteed to exist with high probability.
    The probabilistic guarantee of this subset is then incorporated into the final coverage guarantee in the form of an uncertainty budget. 
    Both theoretical insights and numerical experiments demonstrate that our intervals achieve nominal coverage for specified functionals, consistently outperforming the commonly used one-at-a-time strict bound intervals in terms of coverage accuracy and interval length.
    Numerical results include a realistic, severely ill-posed unfolding problem from high-energy physics involving a rank-deficient forward model.
\end{abstract}

% REQUIRED
\begin{keywords}
uncertainty quantification,
frequentist coverage,
constrained hypothesis testing,
likelihood ratio test
\end{keywords}

% REQUIRED
\begin{MSCcodes}
68Q25, 68R10, 68U05
\end{MSCcodes}

\section{Introduction}

We consider frequentist uncertainty quantification in constrained inverse \allowbreak problems with observation model
\(\by=f(\bx^*)+\bepsilon\), where \(f:\mathbb{R}^p\to\mathbb{R}^n\) is a known parametric forward map, \(\bepsilon\sim\mathcal{N}(\bm{0},\bSigma)\) has a known distribution, and the unknown parameter satisfies \(\bx^*\in\mathcal{X}\).
Throughout the paper, \(\mathcal{X}\) is primarily assumed to be a polyhedral constraint set of the form \(\bA\bx\leq\bb\).
We allow the parameter dimension \(p\) to exceed the observation dimension \(n\), and we do not require the forward map to be injective; both features are common in ill-posed inverse problems.
In this work, the uncertainty quantification target is a confidence interval \(\cC_\alpha(\by)\) for a one-dimensional quantity of interest (QoI) \(\varphi(\bx^*)\), where \(\varphi:\mathcal{X}\to\mathbb{R}\).
For a prescribed miscoverage level \(\alpha\in(0,1)\), we seek finite-sample frequentist coverage uniformly over the constrained parameter space:
\begin{equation} \label{eq:coverage_guarantee}
    \mathbb{P}_{\bx^*}\bigl(\varphi(\bx^*) \in \cC_\alpha(\by)\bigr) \geq 1 - \alpha \quad \text{for all } \bx^* \in \mathcal{X}.
\end{equation}
Here the probability is over the noise in the observation \(\by=f(\bx^*)+\bepsilon\) for fixed $\bx^*$.
Thus, the interval must contain the true functional value \(\varphi(\bx^*)\) with probability at least \(1-\alpha\), no matter which feasible parameter \(\bx^*\) generated the data.

\subsection{Context and related work} 
\label{sec:context_and_related_work}

High-dimensional, ill-posed inverse problems have become more prevalent, especially in remote sensing and data assimilation fields.
This setting includes a wide array of physical science applications, spanning Earth science \cite{rodgers}, atmospheric science and remote sensing \cite{liu_2016,patil}, and high energy physics \cite{kuusela_phd_thesis,stanley_unfolding}, among many others.
Providing guarantees for UQ in parameter inference is essential for assessing the precision of downstream scientific inferences.
However, the inherent ill-posed nature of these problems often leads to inferences that are highly sensitive to noise, posing significant challenges for UQ.
Namely, the ill-posedness leads to an identifiability issue in which statistical inference is impossible without providing some form of regularization, which usually takes either a deterministic (e.g., SVD truncation \cite{hocker_svd} and Tikhonov regularization \cite{schmitt_tunfold}) or probabilistic (e.g., priors and Bayesian inference) form.
Under some assumptions, these two approaches are mathematically equivalent and are therefore subject to the same pitfalls, such as a disruption of the coverage guarantees of downstream intervals as a result of the incurred regularization bias, as thoroughly discussed in \cite{kuusela_phd_thesis}.
Our method's focus on one-dimensional QoI's and incorporation of parameter constraints allow for implicit regularization. 
Therefore, it produces intervals with the promised coverage guarantee while avoiding the problem of regularization bias altogether.

Including parameter constraints shifts complexity to statistical inference under constraints, a non-trivial problem in even elementary settings \cite{gourieroux, wolak87, robertson, shapiro1988towards, wolak89,molenberghs2007likelihood}.
\emph{Optimization-based UQ}, originating in the work of \cite{burrus1965utilization} and \cite{rust_burrus}, offers one solution.
\cite{stark1992inference} extended and generalized their approach, calling this method \emph{strict bounds}, since it produces guaranteed simultaneous coverage interval estimators complying with known physical constraints on the model parameters.
This collection of work defines endpoint optimization problems over the physically constrained parameter space to compute confidence intervals for one-dimensional QoI's directly.
Not only does the optimization form of the confidence interval computation shift statistical inference complexity to numerical optimization, but it also allows known physical constraints to be directly included in the endpoint optimizations.
This practical advantage is dulled by the difficulty of proving the coverage properties of the intervals resulting from the defined endpoint optimizations in the one-at-a-time setting.
Even under the relatively strong assumptions in \cite{rust_burrus} of a linear forward model, non-negativity parameter constraints, and linear QoI, the authors were only able to conjecture the coverage of their interval (known as the Burrus Conjecture).
The interval coverage was unsuccessfully proven in \cite{oleary_rust} (the error pointed out in \cite{tenorio2007confidence}) and finally generally refuted in \cite{previous_paper}.
These optimized-based intervals that have gained recent attention take the form:
\begin{equation} \label{eq:opt_based_intervals}
    \cI(\psi_\alpha^2, \by) := \bigg[\varphi^{l}(\psi_\alpha^2, \by), \varphi^{u}(\psi_\alpha^2, \by)\bigg] = \bigg[\min_{\bx \in D(\psi_\alpha^2, \by)} \varphi(\bx), \max_{\bx \in D(\psi_\alpha^2, \by)} \varphi(\bx)\bigg],
\end{equation}
where
\begin{equation} \label{eq:constraint_def}
    D(\psi_\alpha^2, \by) := \big\{\bx \in \mathcal{X} : \lVert \by - f( \bx) \rVert_2^2 \leq \psi_\alpha^2 \big\}.
\end{equation}
The statistical challenge for these intervals is choosing $\psi_\alpha^2$ such that $\cI(\psi_\alpha^2, \by)$ has the desired coverage guarantee.
Here $D(\psi_\alpha^2, \by)$ assumes standardized noise; i.e., the covariance matrix of the noise has been transformed to the identity matrix.

The literature proposes two settings to guarantee coverage.
One approach is to set $\psi_\alpha^2$ such that $D(\psi_\alpha^2, \by)$ is itself a confidence set for $\bx^*$ in the parameter space.
Since this choice would automatically guarantee coverage for $\cI(\psi_\alpha^2, \by)$ regardless of the chosen QoI, it has been called ``simultaneous'' in \cite{oleary_rust}, and ``Simultaneous Strict Bounds'' (SSB) in \cite{stanley_unfolding, previous_paper} since it aligns with the setting of the strict bounds construction in \cite{stark1992inference}.
Under the Gaussian assumption on the additive noise term $\bepsilon$, a $1 - \alpha$ confidence set for $\bx^*$ can be obtained by setting $\psi_{SSB, \alpha}^2 := \chi^2_{n, \alpha}$, where $\chi^2_{n, \alpha}$ is the upper $\alpha$-quantile for a chi-squared distribution with $n$ degrees of freedom.
This approach is conservative since the guarantee simultaneously holds for all possible QoI choices.
By contrast, the ``one-at-a-time'' setting as described in \cite{rust_burrus, oleary_rust} (``one-at-a-time strict bounds'' (OSB) as called in \cite{stanley_unfolding, previous_paper}) tailors the interval coverage to one particular QoI.
Under the same Gaussian assumption this proposed setting is $\psi_{OSB, \alpha}^2 = \chi^2_{1, \alpha} + s(\by)^2$, where $s(\by)^2 := \min_{\bx \in \mathcal{X}} \lVert \by - f(\bx) \rVert_2^2$.
Unlike the simultaneous setting, $D(\psi_{OSB, \alpha}^2, \by)$ is not a $1 - \alpha$ confidence set for all $\bx^*$, which makes proving the coverage guarantee difficult.
The validity of this claim was proposed by the Burrus conjecture \cite{rust_burrus} and generally disproven by \cite{previous_paper}.
To address the challenge of calibrating strict bounds and optimization-based confidence intervals, \cite{previous_paper} approached these intervals as an inverted likelihood ratio test.
 % as shown in \Cref{eq:original_confidence_set}.
The distribution of the log-likelihood ratio (LLR) statistic is non-standard due to the constraints, complicating Type I error control.

Viewing QoI confidence sets as inverted likelihood ratio tests means that confidence set coverage can be achieved by controlling the Type I error rate of the hypothesis test.
A mathematical challenge following from this formulation is that the parameter space of the null hypothesis is composite (i.e., contains more than one parameter setting) and can be unbounded.
Since controlling the Type I error rate of a test with a composite null hypothesis test requires controlling at the worst-case parameter setting, the unboundedness poses a challenge.

\subsection{Our approach and contributions}
We address this unboundedness challenge by introducing the \emph{\bergerboos} set, inspired by \cite{berger_boos} for handling composite null tests involving nuisance parameters, allowing us to create a bounded version of the composite null hypothesis space.
We then develop maximum quantiles over the bounded \bergerboos set to control the hypothesis test Type I error rates and thus construct calibrated confidence sets.
In the additive Gaussian setting used in the algorithms below, this data-adaptive set has the simple form
\begin{equation} \label{eq:intro_bb_preview}
\begin{aligned}
    \Gamma_\eta(\by) &:= \{\by' \in \mathbb{R}^n : \| \by - \by' \|_2^2 \leq \chi^2_{n,\eta}\}, \\
    \mathcal{B}_\eta(\by) &:= f^{-1}(\Gamma_\eta(\by)) \cap \mathcal{X} \\
    &= \{\bx \in \mathcal{X}: \|\by - f(\bx)\|_2^2 \leq \chi^2_{n,\eta}\}.
\end{aligned}
\end{equation}
The set \(\mathcal{B}_\eta(\by)\) contains the true parameter with probability at least \(1-\eta\); the remaining miscoverage budget is used by calibrating the LLR tests at level \(\gamma\), with \(\gamma \leq \alpha-\eta\).
The formal construction and coverage argument are given in \Cref{sec:conf_sets_using_bb_set}.

Since the aforementioned OSB intervals are the state-of-the-art confidence intervals in this constrained inference setting, we conduct a series of numerical experiments showing that our intervals are competitive with OSB in situations where OSB achieves nominal coverage.
In situations where OSB does not achieve nominal coverage, we show that our intervals do, while also achieving smaller expected interval lengths in a high-dimensional experiment.

Methodologically, we propose new confidence interval constructions for the setting described at the beginning of the section.
Although the use of the \bergerboos set is inspired by \cite{berger_boos}, they originally applied it in a hypothesis testing setting for handling nuisance parameters.
\cite{masserano_berger_boos} applied a similar idea to nuisance parameters in a classification setting, but to the best of our knowledge, our paper is the first to use the idea in the ill-posed inverse problem UQ setting.
Inspired by the methods used in simulation-based inference \cite{dalmasso2020confidence, lf2i, waldo}, we apply the approach of estimating the quantile function of a test statistic using simulated data.
This work differs in the underlying model assumptions and composite nature of the null hypotheses. 
In addition, the prior work assumed bounded and relatively low-dimensional parameter spaces.
The novelty of our method is therefore not the Berger--Boos construction, specific sampling algorithms, or quantile regression in isolation, but their integration into a coverage-guaranteed test-inversion procedure for functionals in constrained, ill-posed inverse problems.
To the best of our knowledge, this work is the first to combine test inversion, sampling, and quantile regression to produce constraint-aware confidence intervals for functionals in high-dimensional, potentially unbounded settings.

Computationally, we adapt sampling-based calibration ideas to draw design points from the data-adaptive \bergerboos set and to estimate the relevant LLR quantiles from simulated data. The sampling problem can be challenging because this set is an observed-data-dependent pre-image intersected with parameter constraints, and its geometry reflects the ill-posedness of the forward map. To efficiently sample from this set, we consider a Voelker--Gosmann--Stewart (VGS) \cite{vgs} sampler in low dimensions and a custom-built Polytope/Vaidya-walk sampler \cite{mcmc_polytope} in higher dimensions and demonstrate their effectiveness in concrete ill-posed uncertainty quantification problems.

The rest of the paper is organized as follows.
\Cref{sec:background} provides the necessary statistical preliminaries for our method, including the hypothesis test we invert, the test statistic we use, and the test statistic quantile functions necessary for confidence set calibration.
\Cref{sec:interval_methodology} defines the \bergerboos set and downstream critical values used to calibrate several different confidence set varieties. It also provides two algorithms for computing the confidence sets.
\Cref{sec:theory} provides a necessary theoretical result to justify that the algorithms described in \Cref{sec:interval_methodology} produce the correct confidence sets.
\Cref{sec:sampling-quantileregress} describes two core computational components necessary for the construction of our confidence sets in practice.
\Cref{sec:numerical_exp} presents simple numerical experiments to illustrate the method's operation and a high-dimensional particle physics-inspired experiment to demonstrate our method's performance relative to the state-of-the-art.
Finally, \Cref{sec:discussion_and_conclusion} summarizes our results and presents next steps.

\section{Statistical preliminaries}
\label{sec:background}

A key part of the approach in \cite{previous_paper} was to view the intervals in \Cref{eq:opt_based_intervals} as inverted hypothesis tests, more specifically, a particular inverted likelihood ratio test.
This change in perspective allowed for interval analysis using the properties of a particular log-likelihood ratio (LLR) test.
We summarize this connection here, as it is critical to our method development carried out in Section~\ref{sec:interval_methodology}.

Suppose $\by \sim  P_{\bx^*}$, where $P_{\bx^*}$ is a distribution parameterized by a fixed but unknown $\bx^* \in \mathbb{R}^p$.
Let $\ell_{\bx}(\by)$ denote the log-likelihood of $\bx$ evaluated at $\by$.
We furthermore suppose we know a set $\cX \subseteq \mathbb{R}^p$ such that $\bx^* \in \mathcal{X}$.
The case $P_{\bx} = \mathcal{N}(\bK \bx, \bI)$, where the mean is linear in $\bx$ and the additive noise is Gaussian with standardized covariance, is a scenario of particular interest; we henceforth refer to it as the \emph{linear-Gaussian} case.
In this case, the log-likelihood takes the following form, $\ell_{\bx}(\by) = -\frac{1}{2} \lVert \by - \bK \bx \rVert_2^2 - \frac{n}{2}\log 2\pi$.
Although our numerical experiments in \Cref{sec:numerical_exp} focus on the linear-Gaussian case, we present the framework using the generic likelihood function to show that it theoretically holds more generally.

The duality between hypothesis tests and confidence sets is well known in statistics (see, e.g., Chapter~7 of \cite{casella_berger} or Chapter~5 of \cite{Panaretos2016}). 
We invert the following hypothesis test:
\begin{equation} \label{eq:hypothesis_test}
    H_0: \bx^* \in \Phi_\mu \cap \mathcal{X} \quad \text{versus} \quad H_1: \bx^* \in \mathcal{X} - \Phi_\mu,
\end{equation}
where $\Phi_\mu := \{\bx : \varphi(\bx) = \mu \}$.
Since we seek a confidence interval on the real line, each hypothesis test is defined by $\mu \in \mathbb{R}$.
This null hypothesis is \emph{composite}: fixing the value of the QoI, $\varphi(\bx)=\mu$, generally does not identify a single parameter vector, but instead leaves all parameter settings in the level set $\Phi_\mu \cap \mathcal{X}$.

We define the following test statistic to evaluate Test~\eqref{eq:hypothesis_test}:
\begin{equation} \label{eq:llr}
    \lambda(\mu, \by) := 
    \begin{cases}
        \displaystyle
        \inf_{\bx \in \Phi_\mu \cap \mathcal{X}} \bigl(-2 \ell_\bx(\by)\bigr)
        \;-\; 
        \inf_{\bx \in \mathcal{X}} \bigl(-2 \ell_\bx(\by)\bigr),
        & \text{if } \Phi_\mu \cap \mathcal{X} \neq \emptyset, \\[8pt]
        \infty, 
        & \text{otherwise}.
    \end{cases}
\end{equation}
where $\mu$ denotes the level set of the null hypothesis, $\by$ is the observed data, and $\mathcal{X}$ is the constraint set known to contain the parameter. 
The test is rejected if the test statistic is large, and hence we automatically reject the null if $\Phi_\mu \cap \mathcal{X} = \emptyset$ independently of the observed data. 
Alternatively, we can consider the test \eqref{eq:hypothesis_test} only for $\mu \in \varphi(\mathcal{X})$ to perform test inversion.
Under the linear-Gaussian case, the LLR test statistic is,
\begin{equation} \label{eq:llr_gaussian}
    \lambda(\mu, \by) = \min_{\bx \in \Phi_\mu \cap \mathcal{X}} \lVert \by - \bK \bx \rVert_2^2 - \min_{\bx \in \mathcal{X}} \lVert \by - \bK \bx \rVert_2^2.
\end{equation}

We control the behavior of this test statistic, and therefore the test, by bounding from above the probability of erroneously rejecting the null hypothesis (i.e., Type I error).
As such, we consider the distribution of $\lambda(\mu, \by)$ under the null.
For each $\bx \in \mathcal{X}$, let $\mu = \varphi(\bx)$.
Define $Q_{\bx}: (0, 1) \to \mathbb{R}$ such that, for all $\alpha \in (0, 1)$,
\begin{equation} \label{eq:quantile_function}
    \mathbb{P} \left(\lambda(\mu, \by) \leq Q_{\bx}(1 - \alpha) \right) = 1 - \alpha,
\end{equation}
where the probability is over $\by \sim P_{\bx}$.
We refer to $Q_{\bx}$ as the \emph{quantile function} of the LLR under the null hypothesis at $(\bx, \varphi(\bx) = \mu)$.
Since the null hypothesis is composite, using this quantile function to define a cutoff is not enough to control Type I error.
Thus, we use the \emph{sliced} maximum quantile function over the level set under consideration:
\begin{equation} \label{eq:max_q_mu}
    \localQ := \sup_{\bx \in \Phi_\mu \cap \mathcal{X}} Q_{\bx}(1 - \alpha).
\end{equation}
Note that we refer to this version as ``sliced'' since it considers the maximum of the quantile function defined by the level set of the functional, defining a slice through the constrained parameter space.
$\localQ$ as defined above controls the Type I error for a specific value of $\mu$.
We can define a more conservative confidence set using the following \emph{global} maximum quantile:
\begin{equation}  \label{eq:max_q}
    \globalQ := \sup_{\mu \in \varphi(\mathcal{X})} \localQ = \sup_{\bx \in \mathcal X}  Q_{\bx}(1 - \alpha).
\end{equation}
By Lemma 2.1 in \cite{previous_paper}, the set
\begin{equation}\label{eq:confidence_set}
\mathcal{C}^\mu_\alpha(\by; \mathcal{X}) := \{\mu: \lambda(\mu, \by) \leq \localQ \} \subset \mathbb{R}
\end{equation}
defines a $1 - \alpha$ confidence set for the true functional, $\mu^* = \varphi(\bx^*)$.
The global quantile \eqref{eq:max_q} is more conservative than the sliced \eqref{eq:max_q_mu} since it holds for all null hypotheses, and thus, 
\begin{equation}\label{eq:confidence_set_global}
\mathcal{C}_\alpha(\by; \mathcal{X}) := \{\mu: \lambda(\mu, \by) \leq \globalQ \} \subset \mathbb{R}
\end{equation}
is also a $1 - \alpha$ confidence set.

To connect this framework back to Interval \eqref{eq:opt_based_intervals}, under the linear-Gaussian assumption producing the LLR in \Cref{eq:llr_gaussian}, \cite{previous_paper} (Theorem 2.4) proved that
\begin{equation} \label{eq:interval_equivalence}
    \Bigg[\inf_{\mu: \lambda(\mu, \by) \leq \globalQ} \mu, \sup_{\mu: \lambda(\mu, \by) \leq \globalQ} \mu \Bigg] = \Bigg[\inf_{\bx \in D(\globalQ + s(\by)^2, \by)} \varphi(\bx), \sup_{\bx \in D(\globalQ + s(\by)^2, \by)} \varphi(\bx) \Bigg],
\end{equation}
where $D(\cdot, \by)$ is defined by \Cref{eq:constraint_def}.
This equivalence asserts that by setting $\psi_{\alpha}^2 := \globalQ + s(\by)^2$, where $s(\by)^2 := \min_{\bx \in \cX} \lVert \by - \bK \bx \rVert_2^2$, we guarantee coverage for OSB Interval~\eqref{eq:opt_based_intervals}.
It asserts that the original OSB interval formulation is valid if and only if $\chi^2_{1, \alpha} \geq \globalQ$ (see Lemma 2.2 and Corollary 2.3).
\cite{previous_paper} showed that this inequality generally does not hold.

As explored in \cite{previous_paper} and reiterated above, Interval \eqref{eq:opt_based_intervals} can be calibrated by computing $\localQ$ or $\globalQ$.
However, pursuing calibration in this way is computationally challenging and statistically conservative.
Both of these values require the ability to evaluate $Q_{\bx}$.
Without parameter constraints, this quantile function can be constant under linear-Gaussian assumptions \cite{previous_paper}, i.e., the test statistic is pivotal.
But even with a relatively simple two-dimensional Gaussian noise model with non-negativity constraints, this quantile function becomes non-trivial (e.g., see Figure 5.3 in \cite{previous_paper}).
Beyond the practical difficulty of dealing with the underlying quantile function, both \eqref{eq:max_q_mu} and \eqref{eq:max_q} can be expressed as chance-constrained optimization problems, which are known to be NP-hard in general and would need to be solved over an unbounded constraint set \cite{geng2019data, pena2020solving, previous_paper}.
Statistically, since both \eqref{eq:max_q_mu} and \eqref{eq:max_q} are the largest quantile function values subject to their respective constraints, they are, by definition, conservative, especially in scenarios where the true $\bx^*$ is far from these most conservative points.
Since we do not know $\bx^*$, this conservatism is necessary under the framework of \cite{previous_paper}.

As we see in Section~\ref{sec:interval_methodology}, our method addresses both challenges.
In scenarios where evaluating $Q_{\bx}$ is difficult, we sample a collection of design points in a bounded subset of the constraint set, sample the LLR under the null at each design point, and use quantile regression to estimate the quantile surface.
The most conservative points can potentially be removed from consideration by only considering parameter values that are \emph{not unlikely} given the observed data.

\section{Interval construction methodology}
\label{sec:interval_methodology}

This section presents four related interval constructions building on the theory of \cite{previous_paper} by addressing the aforementioned key challenges.
The first implementation challenge is to handle the potentially unbounded constraint set, for which we define and apply the \bergerboos set to create a data-dependent subset of the original constraints.
This set leads to our four interval definitions, which follow a two-stage taxonomy: Global versus Sliced and Inverted versus Optimized.
The second implementation challenge is computing these interval constructions, which we achieve using a combination of novel sampling algorithms and quantile regression.
We will cover this in a later section (\Cref{sec:sampling-quantileregress}).

We rewrite the data-generating process articulated in the introduction:
\begin{equation} \label{eq:sampling_data_gen}
    \by = f( \bx^*) + \bepsilon, \quad \bepsilon \sim \mathcal{N}(\bm{0}, \bI), \quad \bx^* \in \mathcal{X}
\end{equation}
where $f: \mathbb{R}^p \to \mathbb{R}^n$ is a known forward model and $\text{Cov}(\bepsilon) = \bI$, without loss of generality.
Note, in the linear-Gaussian case, we have $f(\bx^*) = \bK \bx^*$.
Let $f^{-1}(A) := \{\bx: f(\bx) \in A \}$ be the pre-image of a set $A \subset \mathbb{R}^n$ under the forward map $f$.

Although the assumed data-generating process in \Cref{eq:sampling_data_gen} is less general than $\by \sim P_{\bx^*}$ assumed in \Cref{sec:background}, it is sufficiently general to contain the application areas mentioned in \Cref{sec:context_and_related_work}.
Slightly generalizing the form of the additive noise distribution would make the construction of the \bergerboos set more complicated, but it is possible in principle.

\subsection{Global and sliced confidence sets using the \bergerboos set} \label{sec:conf_sets_using_bb_set}

Both $\localQ$ and $\globalQ$ in Section~\ref{sec:background} suffer from the same theoretical and practical concerns.
Theoretically, they are \emph{conservative} in the sense that they control Type I error under the worst case truth (i.e., the parameter setting with the largest quantile).
Practically, not only can $Q_{\bx}$ be difficult to compute, but if the constraint set $\mathcal{X}$ is unbounded, computing quantiles becomes even more difficult.
Both challenges stem from the composite nature of the null hypothesis ($H_0$ in \Cref{eq:hypothesis_test}).
Namely, for a given $\mu$, $\Phi_\mu \cap \mathcal{X}$ usually contains more than one feasible parameter value, so Type I error control must hold uniformly over all of these possible truths.
\cite{berger_boos} introduce a compelling solution in the context of hypothesis testing with nuisance parameters (a special case of hypothesis testing on $\varphi(\bx)$), which is to control Type I error only over a data-informed region of the parameter space. 
Following their construction, we build a confidence interval that, instead of using the maximum $1-\alpha$ quantile over $\mathcal{X}$ in \eqref{eq:max_q} (or $\mathcal{X}\cap \Phi_\mu$ in \eqref{eq:max_q_mu}), uses a larger quantile ($1-\gamma$ with $\gamma < \alpha$) but that is maximized over a smaller set. The construction follows a three-step process:

(1) Choose $\eta \in (0, \alpha)$ and build a $1-\eta$ confidence set for $\bx^*$, $\mathcal{B}_\eta$. 
Under the additive Gaussian noise assumption in \Cref{eq:sampling_data_gen}, letting $\Gamma_\eta(\by) := \{\by' \in \mathbb{R}^n: \lVert \by - \by' \rVert_2^2 \leq \chi^2_{n, \eta} \}$ we have $\mathbb{P} \left( f(\bx^*) \in \Gamma_\eta(\by) \right) = 1 - \eta$, hence 
\begin{equation}
    \label{eq:bbsetdef}
    \mathcal{B}_\eta := \bbset = \{ \bx \in \mathcal{X}: \|\by - f(\bx)\|^2_2 \leq \chi^2_{n, \eta}\}
\end{equation}
is a $1-\eta$ confidence set for $\bx^*$\footnote{Intersecting the pre-image $f^{-1}(\Gamma_\eta(\by))$ with the constraint set $\mathcal{X}$ does not change the coverage probability since we know $\bx^* \in \mathcal{X}$.}. 
This pre-image confidence set, $\mathcal{B}_\eta$, is the \bergerboos set.

(2) Optimize the $1-\gamma$ quantile of the test statistic over $\mathcal{B}_\eta$ (or $\mathcal{B}_\eta \cap \Phi_\mu$), instead of $\mathcal{X}$ (or $\mathcal X \cap \Phi_\mu)$, where $\gamma < \alpha$ is chosen to ensure calibration. 
\Cref{lem:setting_level_parameters_and_coverage} proves that $\gamma = \alpha - \eta$ is valid.

(3) Use the obtained quantiles, which we define as 
\begin{align}
    \aosbQlocal &:= \sup_{\bx \in \mathcal{B}_\eta \cap \Phi_\mu} Q_{\bx} ( 1 - \gamma), \label{eq:mq_local} \\
    \aosbQ &:= \sup_{\bx \in \mathcal{B}_\eta} Q_{\bx} ( 1 - \gamma), \label{eq:mq_global}
\end{align}
and construct the following sliced (sl) and global (gl) confidence sets:
\begin{align}
    \localcs &:= \{\mu : \lambda(\mu, \by) \leq \aosbQlocal \}, \label{eq:bb_local_cs} \\
    \globalcs &:= \{\mu : \lambda(\mu, \by) \leq \aosbQ \}. \label{eq:bb_global_cs}
\end{align}

Figure~\ref{fig:bb_interval_schematic} provides a visual roadmap of the calibration procedure after the relevant quantities have been defined. The left panel illustrates how the LLR quantile function \(Q_{\bx}\) can vary over the constrained parameter space. The right panels show how the data-adaptive \bergerboos set restricts the region over which the quantiles are maximized and how the resulting global or sliced cutoffs are used to invert the LLR test.
\begin{figure}[H]
  \centering
  \includegraphics[width=0.48\textwidth]{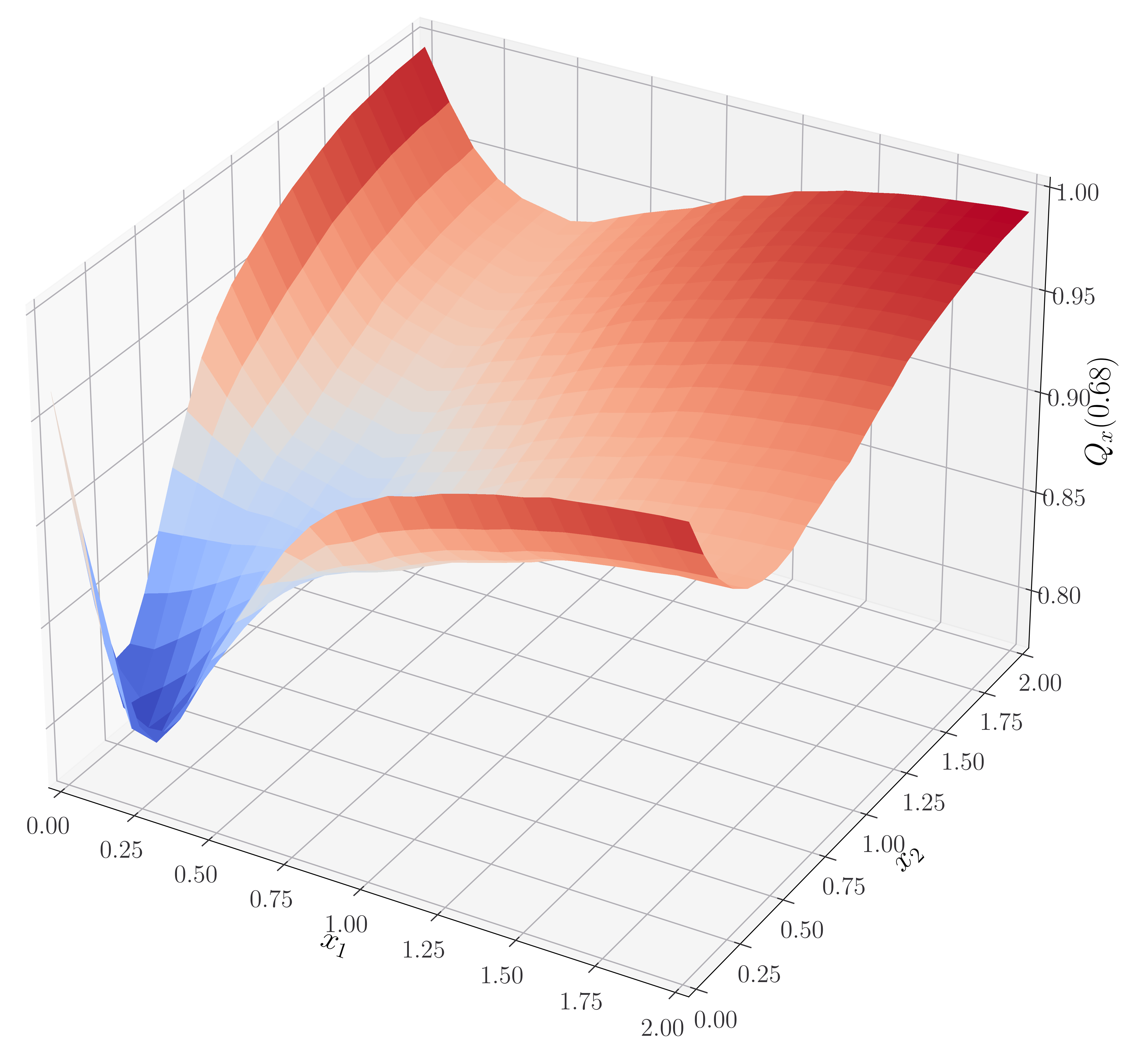}
  \includegraphics[width=0.48\textwidth]{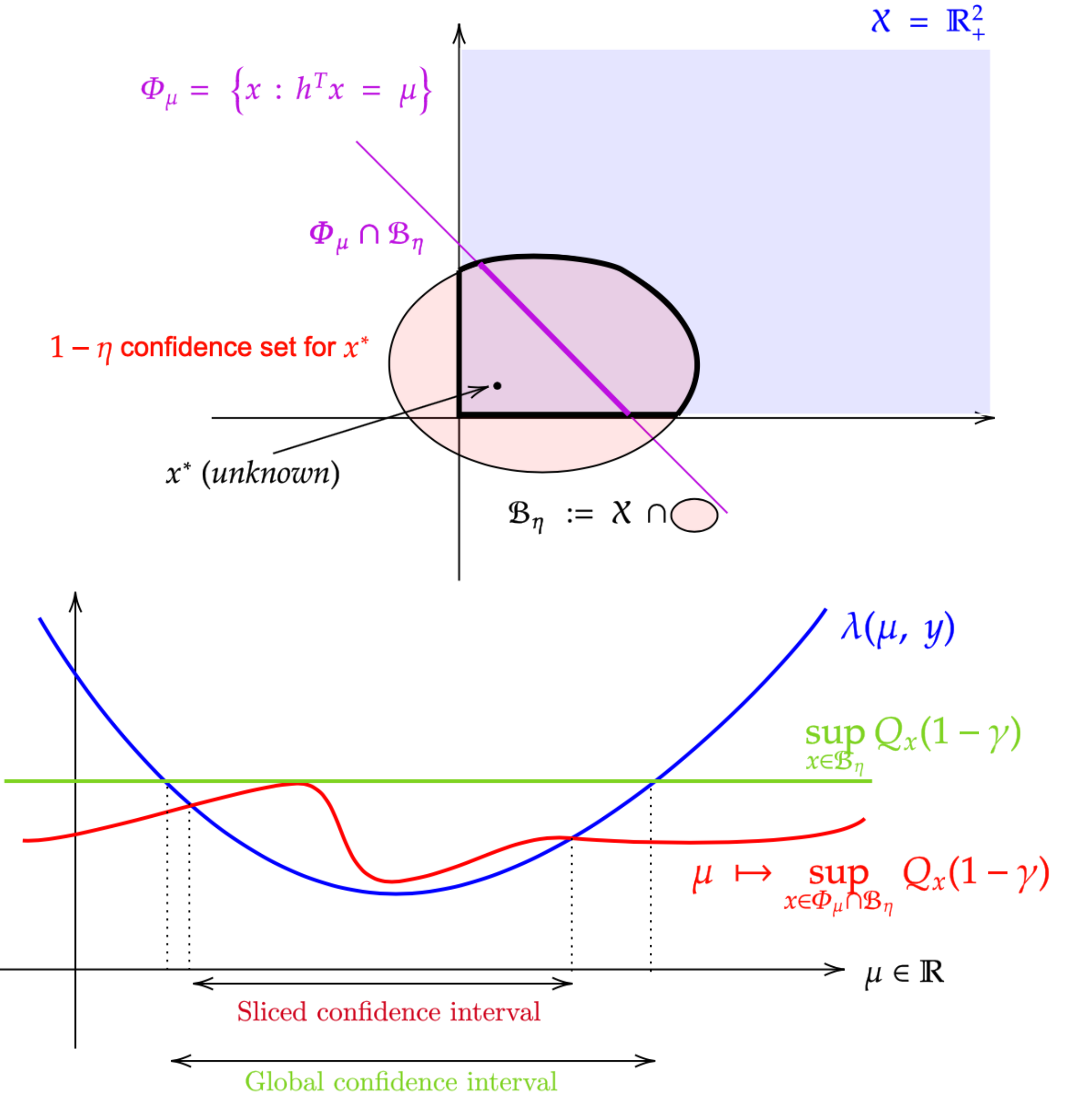}
  \caption[Visual roadmap of data-adaptive calibration]{Visual roadmap of the data-adaptive calibration procedure. Left: an illustrative LLR quantile surface, \(Q_{\bx}(1-\alpha)\), in a two-dimensional constrained-Gaussian example, showing how the relevant quantile can vary over the constrained parameter space. Right-top: the \bergerboos set \(\mathcal{B}_\eta\) restricts the constrained parameter space to a data-adaptive region. Right-bottom: global or sliced LLR cutoffs, based on \(\aosbQ\) or \(\aosbQlocal\), yield different inverted intervals for the QoI.}
  \label{fig:bb_interval_schematic}
\end{figure}

Analogous to the presentation in \cite{previous_paper}, we define \emph{sliced} and \emph{global} max-quantiles that include the \bergerboos set to control the Type I error of Test \eqref{eq:hypothesis_test}. 
Both max-quantiles maximize a larger $\gamma$-quantile over the set of interest instead of the $\alpha$-quantile, as shown in \eqref{eq:max_q_mu} and \eqref{eq:max_q}. 

The following Lemma (analogous to Lemma 2.1 in \cite{previous_paper}) gives a sufficient condition to select the values $\gamma$ and $\eta$ to ensure a $1-\alpha$ coverage. 

\begin{lemma}[Setting $\eta$ and $\gamma$ to guarantee $1 - \alpha$ coverage] 
    \label{lem:setting_level_parameters_and_coverage}
    Let $\alpha \in (0, 1)$ and define $\aosbQlocal$ according \eqref{eq:mq_local} and $\localcs$ according to \eqref{eq:bb_local_cs}. For $\eta \in (0, \alpha)$, $\localcs$ is a $1 - \alpha$ confidence set for $\mu^* = \varphi(\bx^*)$ if $\gamma \leq \alpha - \eta$. 
    The length of the obtained interval is a non-increasing function of $\gamma$, so the tightest interval will be obtained when equality is satisfied.
\end{lemma}

The coverage guarantee implied by \Cref{lem:setting_level_parameters_and_coverage} also implies coverage for $\globalcs$, as shown in the following Corollary. 
Proofs of both statements are provided in \Cref{app:sec:interval_methodology_proofs} (in the Supplementary Material).

\begin{corollary} \label{cor:eta-gamma-alpha-adjustment-coverage}
    In the setting of \Cref{lem:setting_level_parameters_and_coverage}. $\globalcs$ is a $1 - \alpha$ confidence set for $\mu^* = \varphi(\bx^*)$ if $\gamma \leq \alpha -\eta$.
\end{corollary}

\begin{remark}
Following the generality of the original construction, \Cref{lem:setting_level_parameters_and_coverage} generalizes to any $1-\eta$ confidence set of $x^*$ and test statistic calibrated with the $1-\gamma$ quantile. 
In our particular case, since both the confidence set $\mathcal{B}_\eta$ and the test statistic $\lambda(\mu, y)$ depend on the Gaussian log-likelihood $\|\by - f(\bx)\|^2_2$, there is a stronger interpretation of the Berger-Boos construction. 
Whenever $\mathcal{B}_\eta \cap \Phi_\mu \neq \emptyset$, the Berger-Boos construction is equivalent to a data-dependent reduction of the constraint set, replacing $\mathcal{X}$ for $\mathcal{B}_\eta$ both in the test statistic and the quantile optimization problems. Since this is done after seeing the data, the optimized quantile needs to increase to $1-\alpha + \eta$ to maintain $1-\alpha$ coverage.
\end{remark}

\myparagraph{Navigating the \bergerboos parameter choices and trade-offs.}
Setting $\eta = 0$ in $\globalcs$ we maximize $Q_{\bx}(1 - \gamma) = Q_{\bx}(1 - \alpha)$ over $\mathcal{X}$, which returns $Q_{1 - \alpha}^{\text{max}}$ from Equation~\eqref{eq:max_q}. 
For $\eta > 0$, the \bergerboos construction restricts the quantile maximization to a smaller subset of the parameter space while maximizing a larger $1 - \gamma$ quantile to maintain the desired $1 - \alpha$ confidence level of the final confidence set. 
Fundamentally, the $\eta$ choice is a trade-off between these two opposing effects: 
As we increase $\eta$ from $0$, the set over which the quantile function is optimized shrinks, but the quantile level is increased. 
In \Cref{sec:3d_constrained_gaussian}, we perform a numerical experiment comparing average interval length for different values of $\eta$, showing that a small $\eta > 0$ can be beneficial both numerically (since in many instances it allows sampling over a bounded set) and in terms of the average interval length.
In practice, we treat $\eta$ as a small pre-specified calibration parameter rather than choosing it after inspecting coverage outcomes. The original Berger--Boos recommendation and our experiments both support using a small positive value and, when computationally affordable, checking the sensitivity of interval length to a small grid of candidate values. A fully automated rule for choosing an expected-length-optimal $\eta$ while preserving the finite-sample coverage guarantee is an important direction for future work. Related calibration-parameter tuning ideas have been explored in other UQ settings, including objective frequentist uncertainty quantification for atmospheric CO$_2$ retrievals in \cite{patil}; adapting such ideas to the present \bergerboos/test-inversion setting while preserving finite-sample coverage is left for future work.

\subsection{Algorithmic implementations}
\label{sec:interval_constructions}

There are possibly many ways to compute $\globalcs$ and $\localcs$ in practice.
Obtaining either set reduces to the computation of $\aosbQlocal$ and $\aosbQ$.
If the quantile function $Q_{\bx}$ and its gradient $\nabla_{\bx} Q_{\bx}$ could be evaluated, computation could potentially be achieved using a first-order numerical optimizer.
However, since the sliced max-quantile, $\aosbQlocal$, is a function of the level-set parameter, $\mu$, such an optimization would have to be done for each possible functional value.
Since easy function and gradient evaluations rarely exist (see \cite{previous_paper} for some examples where such evaluations are possible), this paper develops a \emph{sampling}-based approach to estimate these quantities to construct Confidence Sets~\eqref{eq:bb_local_cs} and \eqref{eq:bb_global_cs}.
Once we estimate $\aosbQlocal$ and $\aosbQ$, we can either use the output from the sampling algorithm to compute the Global or Sliced interval via classical test inversion, or we can use the estimated max-quantiles in optimizations similar to those in Interval~\eqref{eq:interval_equivalence}.
As such, we introduce four interval constructions: Global Inverted, Global Optimized, Sliced Inverted, and Sliced Optimized.
These four options are summarized in \Cref{tbl:interval_table}.

We use our ability to sample $\lambda(\mu, \by)$ via $\by \sim \cN(f(\bx), \bI)$ to estimate the desired max-quantiles.
We present two algorithms, both of which first generate a random set of design points from the \bergerboos set.
If we can directly \emph{efficiently} compute the desired quantile at each design point, we use \Cref{alg:max_q_rs}.
Absent such an ability, \Cref{alg:max_q_qr} presents an alternative, which first samples a realization of the LLR test statistic at each design point and then performs quantile regression with the generated pairs of design points and LLR values to estimate the underlying quantile surface.
The sampled design points and either the exact or estimated quantiles at the design points are then used to produce the final intervals.
% Both \Cref{alg:max_q_rs} and \eqref{alg:max_q_qr} are detailed in \Cref{supp:algos_for_quantile_function}.

Both algorithms start by sampling $\mathcal{B}_\eta$ uniformly at random to generate a collection of $M$ design points, i.e., $\bar{\bx}_1, \bar{\bx}_2, \dots, \bar{\bx}_M \sim \mathcal{U}\left(\mathcal{B}_\eta \right)$.
Section~\ref{subsec:preimage_sampling} provides the practical sampler choices used for this step: a VGS-based sampler in low-dimensional full-rank settings and a Polytope sampler based on the Vaidya walk in higher-dimensional or rank-deficient settings.
These geometric difficulties are well known in convex-body and polytope sampling. Our contribution is not a new generic sampler for such sets, but rather the use of appropriate samplers to generate design points over the data-adaptive \bergerboos set and across the range of QoI values needed for calibrated test inversion.
If computing $\lambda(\mu, \by)$ for a given $(\mu, \by)$ is inexpensive, \Cref{alg:max_q_rs} directly estimates the quantile function at each design point.
This computation is most likely inexpensive when there are closed-form solutions to the LLR's subordinate optimizations.
As shown in \Cref{alg:max_q_rs}, an easy way to estimate each design point's quantile is to sample its test statistic $N$ times and take the appropriate percentile.
More often, computing $\lambda(\mu, \by)$ is expensive since it involves two constrained optimizations, which are possibly non-convex due to either the constraints or the forward model, or just numerically challenging due to the ill-posedness of the problem.
In this scenario, \Cref{alg:max_q_qr} samples a realization of the test statistic at each design point and estimates the quantile function over the \bergerboos set using quantile regression.
While \Cref{alg:max_q_rs} relies upon computational strength to compute the LLR $N \times M$ times, \Cref{alg:max_q_qr} shifts complexity to the quantile regression and relieves the computational burden by assuming a smooth quantile surface.
% assuming that there is information to be shared about the quantile surface between design points (namely that the quantile surface is smooth).
% The emphasis on the quantile regression further necessitates that the quantile regression be performed well.
We discuss some considerations to this end in Section~\ref{subsec:quantile_regression}.
We note that although \Cref{alg:max_q_rs} and \Cref{alg:max_q_qr} make use of the additive Gaussian noise assumption, they are not necessarily limited to this assumption given the proper adjustments to the \bergerboos set.
% assuming one still has the ability to sample from the data-generating process.
% Since these approaches are sampling-based, it is necessary to show convergence as the number of samples gets large.

The approach of \Cref{alg:max_q_qr} is inspired by recent UQ approaches in likelihood-free inference, where one can sample from a likelihood but cannot easily compute it. 
Specifically, we draw inspiration from the use of quantile regression in \cite{dalmasso2020confidence, lf2i, waldo, masserano_berger_boos}.
Although these approaches differ in detail and implementation from our approach (e.g., they typically focus on settings with low-dimensional $\bx$), they overlap in the sampling and quantile regression perspectives, which effectively allow machine learning to supplement computationally intensive or intractable quantities.
Here, rather than assuming a purely stochastic forward model and therefore only being able to sample from the likelihood, we assume $f$ is a deterministic function involved in $\lambda(\mu, \by)$, which is a random quantity due to the additive noise.
Since \Cref{alg:max_q_qr} involves training a quantile regressor, it includes separate training and testing sets of design points over the \bergerboos set and samples from their respective test statistics.

\begin{algorithm}[!ht]
\caption{Direct estimation of quantiles}
\label{alg:max_q_rs}
\begin{algorithmic}[1]
\REQUIRE
$\alpha, \gamma, \eta \in (0, 1)$ such that $\gamma = \alpha-\eta$, $M, N \in \mathbb{N}$.
\vspace{0.35em}
\STATE \textbf{Construct the \bergerboos set}:
Define $\Gamma_\eta(\by) := \{\by' \in \mathbb{R}^n : \|\by-\by'\|_2^2 \leq \chi^2_{n,\eta}\}$, so that $\mathbb{P}(f(\bx^*) \in \Gamma_\eta(\by)) \geq 1-\eta$, and set $\mathcal{B}_\eta$ as in \Cref{eq:bbsetdef}.
% $\mathcal{B}_\eta := f^{-1}(\Gamma_\eta(\by))\cap\mathcal{X} = \{\bx\in\mathcal{X}:\|\by-f(\bx)\|_2^2\leq \chi^2_{n,\eta}\}$.}
\STATE \textbf{Sample from the \bergerboos confidence set, $\mathcal{B}_\eta$}: 
Sample $\bar{\bx}_1, \dots, \bar{\bx}_M \sim \mathcal{U} \left(\mathcal{B}_\eta \right)$.
\FOR{$k = 1, 2, \dots, M$}
    \STATE \textbf{Sample noise realizations}: Sample $N$ noise realizations: $\bepsilon_1, \dots, \bepsilon_N \sim \mathcal{N}\left(\bm{0}, \bI \right)$.
    \STATE \textbf{Generate LLR draws}: For $i=1,\ldots,N$, set $\by_i := f(\bar{\bx}_k)+\bepsilon_i$ and compute $\lambda_i := \lambda\left(\varphi(\bar{\bx}_k), \by_i; \cX \right)$.
    \STATE \textbf{Compute percentile estimate of the $\gamma$-quantile}: Compute the $(1 - \gamma) \times 100$ percentile of the LLR samples for the data-generating process under $\bar{\bx}_k$, i.e., $\hat{q}_{\gamma}^{k} := \lambda_{(\{(1 - \gamma) N \})}$, where $\{ \cdot \}$ denotes the nearest whole number and $\lambda_{(i)}$ denotes the $i$-th order statistic.
\ENDFOR
\vspace{0.15em}
\ENSURE Pairs of sampled design points and their respective $\gamma$-quantiles, i.e., $\{(\bar{\bx}_k, \hat{q}^k_{\gamma}) \}_{k = 1}^M$.
\end{algorithmic}
\end{algorithm}

\begin{algorithm}[!ht]
\caption{Quantile regression estimate of quantile surface}
\label{alg:max_q_qr}
\begin{algorithmic}[1]
\REQUIRE $\alpha, \gamma, \eta \in (0, 1)$ such that $\gamma = \alpha-\eta$; $M_{\tr}, M \in \mathbb{N}$.
\vspace{0.35em}
\STATE \textbf{Construct the \bergerboos set}: 
Define $\Gamma_\eta(\by) := \{\by' \in \mathbb{R}^n : \|\by-\by'\|_2^2 \leq \chi^2_{n,\eta}\}$, so that $\mathbb{P}(f(\bx^*) \in \Gamma_\eta(\by)) \geq 1-\eta$, and set $\mathcal{B}_\eta$ as in \Cref{eq:bbsetdef}.
% := f^{-1}(\Gamma_\eta(\by))\cap\mathcal{X} = \{\bx\in\mathcal{X}:\|\by-f(\bx)\|_2^2\leq \chi^2_{n,\eta}\}$.}
\STATE \textbf{Sample from the \bergerboos confidence set $\mathcal{B}_\eta$}: 
Sample $\bar{\bx}_{1}, \dots, \bar{\bx}_{M_{\tr}} \sim \mathcal{U} \left(\mathcal{B}_\eta \right)$ design points to train the quantile regressor and $\tilde{\bx}_{1}, \dots, \tilde{\bx}_{M} \sim \mathcal{U} \left(\mathcal{B}_\eta \right)$ test points to invert the interval, generating $M_{\tr} + M$ total samples.
Since the test points are used for the interval inversion, they are used as out-of-sample points for the quantile regressor.
\FOR{$k = 1, 2, \dots, M_{\tr}$}
    \STATE \textbf{Sample a noise realization:} Sample a noise realization: $\bepsilon_k \sim \mathcal{N}\left(\bm{0}, \bI \right)$.
    \STATE \textbf{Generate one LLR draw:} Set $\by_k := f(\bar{\bx}_{k}) + \bepsilon_k$ and compute $\lambda_k := \lambda\left(\varphi(\bar{\bx}_{k}), \by_k; \cX \right)$.
\ENDFOR
\STATE \textbf{Estimate the quantile function using quantile regression:} Using the generated pairs $\left\{\left(\bar{\bx}_{k}, \lambda_k \right)\right\}_{k = 1}^{M_{\tr}}$, estimate the upper $\gamma$-conditional quantile function, $\hat{q}_{\gamma}(\bx)$, using quantile regression. 
\vspace{0.15em}
\ENSURE Generate $\gamma$-quantile predictions at out-of-sample test points, $\left\{\left(\tilde{\bx}_{k}, \hat{q}_\gamma(\tilde{\bx}_{k}) \right) \right\}_{k = 1}^{M}$.
\end{algorithmic}
\end{algorithm}

With the generated pairs from Algorithms~\eqref{alg:max_q_rs} and \eqref{alg:max_q_qr}, we present two strategies to estimate each of $\globalcs$ and $\localcs$.
To streamline notation, let $\left(\bx_k, q_\gamma^k \right)$ denote the $k$-th pair from either algorithm.
This is notationally helpful since \Cref{alg:max_q_rs} only generates one set of parameter samples, whereas \Cref{alg:max_q_qr} generates two.
That is, using \Cref{alg:max_q_rs}, $\bx_k := \bar{\bx}_k$ and $q_\gamma^k := \hat{q}_\gamma^k$ and using \Cref{alg:max_q_qr}, $\bx_k := \tilde{\bx}_{k}$ and $q_\gamma^k := \hat{q}_\gamma(\tilde{\bx}_{k})$.

To estimate the global confidence set, we estimate $\aosbQ$ using the empirical maximum, $\hat{q} := \max_k q_\gamma^k$.
This results in the following two interval constructions:
\begin{align}
    C^{\gl}_{\opt}(\by) &:= \min/\max \big\{\varphi(\bx): \bx \in D(\hat{q} + s(\by)^2, \by) \big\} \label{eq:global_optimized} \\
    C^{\gl}_{\inv}(\by) &:= \min / \max \big\{\varphi(\bx_k): k = 1, \dots, M \; \text{and} \; \lambda(\varphi(\bx_k), \by; \cX) \leq \hat{q} \big\}, \label{eq:global_inverted}
\end{align}
We refer to Interval~\eqref{eq:global_inverted} as the Global Inverted interval construction since it is endpoints defined by only those sampled parameter values that comport with the maximum estimated quantile LLR cutoff.
We refer to Interval~\eqref{eq:global_optimized} as the Global Optimized interval since its endpoints are defined by the extreme functional values of a feasible region defined by $\hat{q}$.
Although, $C^{\gl}_{\inv}(\by) \neq C^{\gl}_{\opt}(\by)$ due to finite sample, they are asymptotically equal and in practice show similar performance in terms of coverage and expected length, as shown in \Cref{sec:numerical_exp}.
We prove the consistency of the interval constructed via inversion and via optimization in \Cref{thm:maintheoremalgs}.
Since the endpoints of both the inverted and optimized intervals converge in probability to the endpoints of $\globalcs$, the two interval constructions are asymptotically equivalent.

To estimate the sliced set, we estimate $\aosbQlocal$ using a rolling maximum of the sampled $q_\gamma^k$ values, as ordered by the sampled functional values, and directly accept or reject functional values based on their estimated quantile.
These approaches result in two interval constructions:
\begin{align}
    C^{\sl}_{\opt}(\by) &:= \min/\max \big\{ \mu \in \mathbb{R}: \lambda(\mu, \by; \cX) \leq \hat{m}_\gamma(\mu) \big\}, \label{eq:sliced_optimized} \\
    C^{\sl}_{\inv}(\by) &:= \min/\max \big\{ \varphi(\bx_k): k = 1, \dots, M \; \text{and} \; \lambda(\varphi(\bx_k), \by; \cX) \leq q_\gamma^k \big\}, \label{eq:sliced_inverted}
\end{align}
% where $I^{\sl}(\by) = \{k: \lambda(\varphi(\bx_k), \by; \mathcal{X}) \leq q_\gamma^k \} \subseteq [M]$, $L(\by) := \{\mu: \lambda(\mu, \by; \mathcal{X}) \leq \hat{m}_\gamma(\mu) \}$ 
where $\hat{m}_\gamma(\mu)$ denotes rolling estimate of $\aosbQlocal$ defined as follows.
We refer to $I^{\sl}(\by)$ as the ``sliced'' index set, i.e., those indices for which the LLR at a particular functional value is less than the estimated quantile at a design point generating that functional value.

The rolling maximum quantile is defined using estimated quantiles ordered by the sampled functional values.
Choose ``rolling'' parameter, $T \in \mathbb{N}$, and let $\sigma(1), \sigma(2), \dots, \sigma(M)$ define an ordering such that $\mu_{\sigma(k)} \leq \mu_{\sigma(k + 1)}$ for all $k = 1, \dots, M - 1$.
Define $$Q_k := \{q^{\sigma(k)}_\gamma, q^{\sigma(k - 1)}_\gamma, \dots, q^{\sigma(k - T)}_\gamma \}.$$
Then, for a given $\mu \in \left[\min_{\bx \in \mathcal{B}_\eta} \varphi(\bx), \max_{\bx \in \mathcal{B}_\eta} \varphi(\bx) \right]$, define $k^*(\mu) := \text{argmin}_k \lvert \mu - \mu_{\sigma(k)} \rvert$ and define $\hat{m}_\gamma(\mu) := \max \{q \in Q_{k^*(\mu)} \}$.
Using the rolling maximum quantile is one way to estimate $\aosbQlocal$.
One could also bin the quantiles by functional value, compute the maximum predicted quantile in each bin, and then fit a nonparametric regression to fit the maximum binned quantiles to the functional values.

This estimator choice for $\hat{m}_\gamma(\mu)$ affects how one computes $C^{\sl}_{\opt}(\by)$.
To see why, note that we can re-express $C^{\sl}_{\opt}(\by)$ as follows:
\begin{equation}\label{eq:opt_endpoints}
    C^{\sl}_{\opt}(\by) := \min/\max \big\{\varphi(\bx): \bx \in D(\hat{m}_\gamma(\varphi(\bx)) + s(\by)^2, \by) \big\}.
\end{equation}
As such, the estimated $\hat{m}_\gamma(\cdot)$ could be substituted in each endpoint optimization.
However, this estimated curve is likely not convex in $\bx$ and therefore complicates optimizations.
One could also pursue a root-finding approach to find the set of $\mu$ such that $\lambda(\mu, \by) = \hat{m}_\gamma(\mu)$ since these intersection points define a set of accepted functional values.
This approach can be as complex as the complexity of the estimated curve.
In our view, the most pragmatic approach is to determine all $\mu_k$ such that $\lambda(\mu_k, \by; \mathcal{X}) \leq \hat{m}_\gamma(\mu_k)$, and then define the endpoints of $C^{\sl}_{\opt}(\by)$ as the minimum and maximum of the accepted sample values.
The consistency of the computed endpoints for both inverted and optimized sliced intervals is proven in \Cref{thm:maintheoremalgs}.
Similarly to the global interval constructions, the consistency of both constructions to $\localcs$ establishes the asymptotic equivalence of the inverted and optimized approaches.

\begin{table}[!t]
\centering
\caption{Summary of methods based on global and sliced max quantile, and whether they are optimization-based or inversion-based.} \label{tbl:interval_table}
\begin{tabularx}{\textwidth}{X X X}
Category & Global max quantile & Sliced max quantile \\
\midrule
Inversion-based methods & Interval \eqref{eq:global_inverted} & Interval \eqref{eq:sliced_inverted} \\
Optimization-based methods & Interval \eqref{eq:global_optimized} & Interval \eqref{eq:sliced_optimized} \\
\bottomrule
\end{tabularx}
\end{table}

Although both Algorithms~\ref{alg:max_q_rs} and \ref{alg:max_q_qr} can produce a quantile estimate, it is worth noting a few key differences between the two. 
First, \Cref{alg:max_q_rs} involves a nested sampling loop. 
Thus, any statistical guarantee regarding the validity of its output has the two moving parts of the accuracy of $\hat{q}_\gamma^k$ as $N$ gets large and the proximity of the approximated quantile to the true maximum quantile function value over $\mathcal{B}_\eta$ as $M$ gets large. 
By contrast, \Cref{alg:max_q_qr} only has one sampling loop and estimates the full quantile function surface over $\mathcal{B}_\eta$, producing an estimate of the maximum as a consequence, but is reliant on smoothness to do an accurate regression with finitely many data points. 
In \Cref{sec:theory}, we provide a rigorous theoretical justification for the consistency of each algorithm estimate and prove that the consistency propagates to all defined interval constructions.

\section{Theoretical justification}
\label{sec:theory}

In this section, we provide convergence results for the interval constructions of the previous section. 
We prove that, under certain assumptions, whenever \Cref{alg:max_q_rs} or \Cref{alg:max_q_qr} are used to estimate the maximum quantiles, the global interval constructions \eqref{eq:global_optimized} and \eqref{eq:global_inverted} converge in probability to the true global interval $\globalcs$ in \eqref{eq:bb_global_cs}, and the sliced interval constructions \eqref{eq:sliced_optimized} and \eqref{eq:sliced_inverted} converge in probability to the true sliced interval $\localcs$ in \eqref{eq:bb_local_cs}.
Although our approach involving quantile regression is inspired by the approaches in \cite{dalmasso2020confidence, lf2i, waldo} and therefore requires similar theoretical results to connect the consistency of the quantile estimation with the interval validity, our approach is sufficiently different so that it requires novel theoretical insights.
First, the inversion of a hypothesis test with a composite null adds a layer of complexity to the proofs.
\cite{lf2i} do address composite null hypotheses, but not in the functional setting, as we do here.
% Second and more fundamentally, \cite{lf2i, waldo} construct confidence sets which are shown to achieve the desired coverage level asymptotically in the number of samples used to train the quantile regressor.
% By contrast, we show that our computed intervals converge in probability to the theoretical Intervals~\eqref{eq:bb_local_cs} and \eqref{eq:bb_global_cs}, which achieve the correct coverage by definition.
Second, since we use the sampled points to invert the interval, it is insufficient to show that the cutoff is consistent as the results in \cite{dalmasso2020confidence, lf2i, waldo} show.
Instead, our inverted intervals require proof that the samplers can get arbitrarily close to the true endpoint boundaries.

Throughout the following results, we assume that for fixed $\gamma \in (0, 1)$, $Q_{\bx}(1 - \gamma)$ is a continuous function of $\bx$.
We refer the reader to \cite{Kibzun1997} for a complete analysis of the properties of parametrized quantile functions.
We furthermore assume that $\mathcal B_\eta$ is a compact set without isolated points so that the sampling methods eventually sample close to every point.

\begin{theorem}\label{thm:maintheoremalgs}
    Assume that $Q_{\bx}(1 - \gamma)$ is a continuous function of $\bx$, and let $\mathcal B_\eta$ be compact and without isolated points. 
    Let the quantile regression in \Cref{alg:max_q_qr} be consistent for all $\bx$, i.e, such that  
    $\mathbb{P}(\lvert \hat{q}_\gamma(\bx) - Q_{\bx}(1 - \gamma) \rvert > \varepsilon) \to 0 \text{ as } M_{tr} \to \infty$ is satisfied $\forall \varepsilon > 0$. 
    We will write $\convprob{}{}$ for convergence in probability, understood as $N,M \to \infty$ if \Cref{alg:max_q_rs} is used and as $M_{tr}, M \to \infty$ if \Cref{alg:max_q_qr} is used. For either algorithm, we have, for a given observation $\by$:
% \begin{enumerate}
    % \item 
    $\convprob{C^{\gl}_{\inv}(\by)}{C^{\gl}_\alpha(\by)}$,
    % \item 
    $\convprob{C^{\sl}_{\inv}(\by)}{C^{\sl}_\alpha(\by)}$.
% \end{enumerate}
Further assume that there exists a point $\bar{\mu} \in \varphi(\mathcal{B}_\eta)$ satisfying $\lambda(\bar{\mu}, \by; \mathcal{X}) < \aosbQ$, then
% \begin{enumerate}
    % [resume]
    % \item 
    $\convprob{C^{\gl}_{\opt}(\by)}{C^{\gl}_\alpha(\by)}$.
% \end{enumerate}
Finally, further assuming technical conditions discussed in \Cref{proofs_appendix}, then
% \begin{enumerate}
    % [resume]
    % \item 
    $\convprob{C^{\sl}_{\opt}(\by)}{C^{\sl}_\alpha(\by)}$.
% \end{enumerate}
\end{theorem}

% \begin{proof}
%     See \Cref{proofs_appendix}.
% \end{proof}

Note that if $\lambda(\mu, \by)$ is convex in $\mu$, as it is for linear forward models and quantities of interest \cite[Proposition 2.5]{previous_paper}, the condition of the existence of $\bar \mu$ satisfying $\lambda(\bar{\mu}, \by; \mathcal{X}) < \aosbQ$ is equivalent to the interval not being empty. For nonlinear forward models, it is a slightly stronger condition.
For the convergence of the sliced optimized version, one needs to show that $\inf_{\mu : \lambda(\mu, \by) \leq \hat{m}_\gamma(\mu)} \mu$ converges to $\inf_{\mu : \lambda(\mu, \by) \leq m_\gamma(\mu)} \mu$ as $\hat{m}_\gamma(\mu)$ converges to $m_\gamma(\mu)$. 
In order to do so, we study the convergence of optimization problems of the form $\inf_{\mu : \hat{f}(\mu) \geq 0} \mu$ to $\inf_{\mu: f(\mu) \geq 0} \mu$ as $\hat{f}$ converges to $f$. 
Although the result is not true in general, we provide sufficient technical conditions on $f$ and the uniformity of the convergence of $\hat{f}$ to $f$ for the result to hold. 
% We discuss the details in \Cref{proofs_appendix}.
The technical details and proofs of these results are provided in \Cref{proofs_appendix} (in the Supplement).

\section{Implementation methodology}
\label{sec:sampling-quantileregress}

\subsection{Sampling the pre-image \bergerboos set}
\label{subsec:preimage_sampling}

The viability of this method directly relies upon our ability to sample from the \bergerboos set, $\mathcal{B}_\eta = \bbset$.
Under the assumed data-generating process in Equation~\eqref{eq:sampling_data_gen}, the \bergerboos set is defined as follows:
\begin{equation} \label{eq:sim_conf_set_easy}
    \mathcal{B}_\eta = \bbset = \big\{\bx \in \mathcal{X} : (\by - f(\bx))^\top \bSigma^{-1} (\by - f(\bx)) \leq \chi^2_{n, \eta}\big\},
\end{equation}
where we have generalized to a non-identity covariance matrix, $\bSigma$, to make the following exposition more general.
This set is equivalent to the set over which the strict bounds intervals are optimized in \cite{stark_strict_bounds}, also called ``SSB'' intervals in \cite{stanley_unfolding, previous_paper}.
Note, however, that one would use $\chi^2_{n, \alpha}$ instead of $\chi^2_{n, \eta}$ for $1 - \alpha$ interval computation in that scenario.

We focus on sampling strategies under the linear-Gaussian setting ($f(\bx)=\bK\bx$) corresponding to the numerical experiments described in Section~\ref{sec:numerical_exp}. 
Sampling from constrained ellipsoids, polytopes, and convex bodies is a well-developed topic; the purpose of this subsection is therefore to explain which existing sampling ideas are useful for the particular \bergerboos sets arising in our calibrated interval construction, and what trade-offs they introduce.
In low-dimensional cases with full column rank, $\mathcal{B}_\eta$ becomes an ellipsoid intersected with the parameter constraints $\mathcal{X}$.
In such settings, we apply the efficient Voelker--Gosmann--Stewart (VGS) algorithm from \cite{vgs}, which samples uniformly from the ellipsoid by using a suitable linear transformation of uniform samples from a ball, followed by a constraint-based accept-reject step.

However, the VGS algorithm becomes either inefficient or infeasible as dimensionality and ill-conditioning of the forward model increase, due to rapidly declining accept-reject probabilities.
To address these challenging scenarios, we introduce an MCMC-based sampling approach called the Polytope sampler, inspired by the Vaidya walk algorithm from \cite{mcmc_polytope}. 
This method constructs a bounding polytope around $\mathcal{B}_\eta$ defined by linear constraints (e.g., derived from the parameter constraints and eigenvectors of the forward model) and performs parallel random walks starting from strategically chosen points. 
These starting points span the space between the extremes of the functional of interest and the polytope's Chebyshev center. 
Such parallelization helps to ensure exploration of the \bergerboos set without relying solely on a single chain reaching stationarity, which may be challenging due to the geometry induced by parameter constraints.
Additional technical details, including exact algorithms and further illustrations of both samplers, are provided in \Cref{app:sec:sampling_bergerboos} (in the Supplementary Material).

\subsection{Quantile regression}
\label{subsec:quantile_regression}

\Cref{alg:max_q_qr} explained in \Cref{sec:interval_constructions} uses quantile regression to learn a quantile surface from a collection of pairs of design points and LLR test statistic samples.
As mentioned previously, similar approaches have been taken in \cite{dalmasso2020confidence, lf2i, waldo, masserano_berger_boos}, and since quantile regression is a technique that facilitates our interval constructions, we give a brief overview of quantile regression and different ways to implement it in \Cref{app:sec:quantile_regression_overview}.

Using the design points in the \bergerboos set as sampled via the VGS or Polytope samplers, \Cref{alg:max_q_qr} shows how we sample from the test statistic distribution defined at each design point to define a data set to fit a quantile regressor.
We use the Gradient Boosting Regressor implemented in scikit-learn \cite{scikit-learn} with the ``quantile'' loss function (i.e., pinball loss defined in \cref{eq:pinball_loss}) for our numerical examples in \Cref{sec:numerical_exp}.
This algorithm involves a collection of hyperparameters (i.e., the minimum number of samples required to split an internal node, the minimum number of samples needed to be a leaf node, the maximum depth of any individual estimator, the learning rate, and the number of estimators) which we determine using $10$-fold cross validation in a pilot study ahead of our simulation experiments in \Cref{sec:numerical_exp}.
Although one may use any quantile regressor to estimate the quantile surface, we emphasize the importance of choosing a strategy that can accommodate a nonlinear surface in the parameters (such as the Gradient Boosting Regressor) as the parameter constraints are known to produce nonlinear quantile surfaces in even simple examples as seen in \cite{previous_paper}.

Practically, the quantile-regression step trades repeated LLR simulation at every design point for a modeling assumption about the smoothness of the conditional quantile surface.
This trade-off is beneficial when LLR evaluations require expensive constrained optimizations, but it also makes diagnostic checks important: in our experiments we use a pilot cross-validation study to choose hyperparameters and reserve out-of-sample design points for interval construction.
For substantially higher-dimensional problems, the same pipeline can be used in principle, but its performance will depend on obtaining enough design points to cover the relevant functional range and on using a regressor with sufficient inductive bias for the quantile surface.

\section{Numerical experiments}
\label{sec:numerical_exp}

For scenarios within the case of the linear-Gaussian data-generating process of \eqref{eq:sampling_data_gen}, we regard the OSB interval as the previous state-of-the-art option for computing constraint-aware confidence intervals.
Although these intervals have empirically achieved nominal coverage in applications \cite{patil, stanley_unfolding}, they generally do not guarantee coverage \cite{previous_paper}.
As such, we use the OSB interval in the following numerical experiments as a main comparison point to the intervals defined in this paper.
In scenarios where the OSB interval achieves at least nominal coverage, we show that our intervals are either competitive or better in terms of expected interval length.
In scenarios where the OSB interval does not achieve nominal coverage, our intervals do achieve nominal coverage and can have shorter expected length.
The first two experiments use a constrained Gaussian noise model in two or three dimensions to make the above points while demonstrating parts of the interval construction.

Given the straightforward mathematics of these problems, both the mechanics of our intervals and their performance improvements over the OSB interval can be clearly seen.
The second two experiments (inspired by particle unfolding in high-energy physics \cite{stanley_unfolding}) consider a wide-bin deconvolution setup featuring a substantially more challenging 80-dimensional parameter space with a rank-deficient forward model.
These examples demonstrate the superior performance of our intervals over OSB in terms of both coverage and expected length.
The experiments are intentionally ordered from pedagogical to application-oriented.
The two-dimensional example is a visualization and sanity-check case in which the geometry of the \bergerboos set, sampled design points, and LLR cutoffs can be inspected directly.
The three-dimensional example is a calibration stress test in which the OSB interval fails to achieve nominal coverage, while the proposed intervals retain the desired coverage.
The wide-bin deconvolution experiments are the application-oriented inverse-problem UQ examples: they use a rank-deficient, high-dimensional forward model motivated by particle unfolding and test whether the full sampling and quantile-regression pipeline remains effective in the setting for which the method is designed.

In the following experiments, we form $68$\% confidence intervals, set $\eta = 0.01$, and compute $\gamma$ according to \Cref{lem:setting_level_parameters_and_coverage}.
We draw $10^3$ observations from each data-generating process to estimate both coverage and expected interval length for our four interval constructions and the OSB interval.
We also provide $95$\% confidence intervals in the form of orange line segments to characterize the statistical error for both coverage and expected length estimates.
The coverage confidence intervals are Clopper-Pearson intervals for the success probability parameter of a binomial distribution, while the expected length confidence intervals are the average length plus/minus the appropriately scaled standard error of the mean.
The source code to reproduce our experiments is available at \texttt{https://github.com/mcstanle/adaptive-functional-ci}.

\subsection{Constrained Gaussian in two dimensions} \label{sec:2d_constrained_gaussian}

The two-dimensional Gaussian noise model is defined as follows:
\begin{equation}
    \by = \bx^* + \bepsilon, \quad \bepsilon \sim \mathcal{N}\left(\bm{0}, \bI_2 \right), \quad \bx^* \in \mathbb{R}^2_+,
\end{equation}
where $\varphi(\bx) = x_1 - x_2$ and $\bx^* = \begin{pmatrix} 0.5 & 0.5 \end{pmatrix}^\top$.
The LLR is then given as follows:
\begin{equation}
    \lambda(\mu, \by) = \min_{\substack{x_1 - x_2 = \mu \\ \bx \in \mathbb{R}^2_+}} \lVert \by - \bx \rVert_2^2 - \min_{\bx \in \mathbb{R}^2_+} \lVert \by - \bx \rVert_2^2.
\end{equation}
This example first appeared in \cite{tenorio2007confidence} as a case where the OSB interval allegedly fails to achieve nominal coverage when the true parameter $\bx^*$ is such that $\varphi(\bx^*) = 0$.
Since \cite{previous_paper} overturned this result by proving OSB validity in this case, this example is important to include for its historical context and OSB validity.
The proof of OSB interval coverage in this case relies upon showing that $\globalQ = \chi^2_{1, \alpha}$ for all $\alpha \in (0, 1)$, where $\chi^2_{1, \alpha}$ is the upper $\alpha$-quantile of a chi-squared distribution with one degree of freedom.
% Alternatively stated, it holds that $\lambda(\varphi(\bx^*), \by; \mathbb{R}^2_+)$ is stochastically dominated by $\chi^2_1$.
This result is shown in Lemma 4.4 of \cite{previous_paper}.

Estimated coverage and length results are shown in \Cref{fig:2d_coverage_length}.
We note that all four of our interval constructions are competitive with OSB in terms of coverage, while all of our interval constructions have higher estimated expected length, apart from the Sliced constructions, which are within statistical error of OSB.
Since the OSB interval is defined using $\globalQ = \chi^2_{1, \alpha}$ and the $\alpha$-quantile surface of the LLR rapidly approaches this global max-quantile as one moves away from the origin (see Figure 5.3 in \cite{previous_paper}), the OSB interval lengths are difficult to beat in practice with intervals based on the \bergerboos sets since these sets likely contain parameter settings with quantiles near $\chi^2_{1, \alpha}$.
The left panel of \Cref{fig:2d_bb_sets} shows four realizations of the data-generating process with the observations shown as red points.
For each observation, the blue points show uniformly distributed draws within its \bergerboos set, sampled using the VGS sampler.
Cross-referencing the spread of the \bergerboos set samples in \Cref{fig:2d_bb_sets} with Figure 5.3 in \cite{previous_paper}, it is clear that there are always samples in the parameter space where the quantile surface is nearly the same as the $\chi^2_{1, \alpha}$ quantile.
Further, when including the \bergerboos set in the interval construction, we instead construct our intervals using the $\gamma$-quantile, where $\gamma < \alpha$, as the LLR cutoff, resulting in a more relaxed constraint.
This fact can be clearly observed in the central panel of \Cref{fig:2d_bb_sets}, showing the sampled $\gamma$-quantiles within the \bergerboos set of one observation from the data-generating process.
Since a non-trivial portion of this distribution is above $\chi^2_{1, \alpha}$, the longer average length of the Global intervals is explained.
In the right panel of \Cref{fig:2d_bb_sets}, for the same observation, we show the estimated sliced max-quantile function, $\hat{m}_\gamma(\mu)$, in orange alongside $\chi^2_{1, \alpha}$.
Since this estimated function is above $\chi^2_{1, \alpha}$ at their intersection points with the underlying LLR function shown in the solid blue line, it further makes sense that the Sliced interval constructions provide no additional length improvement compared to the OSB interval in this particular setting.

\begin{figure*}[!ht]
    \includegraphics[width=0.49\textwidth]{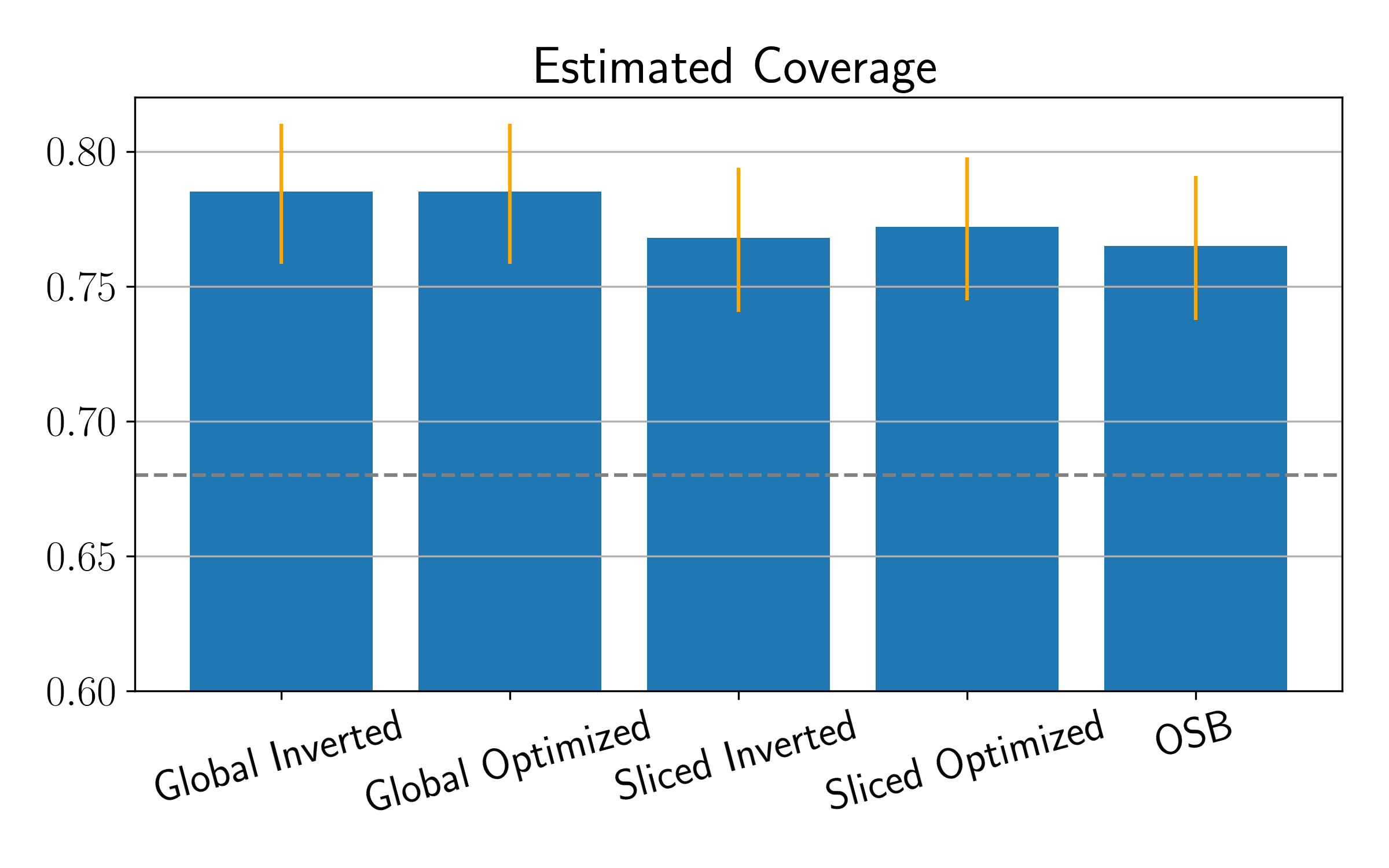}
    \includegraphics[width=0.49\textwidth]{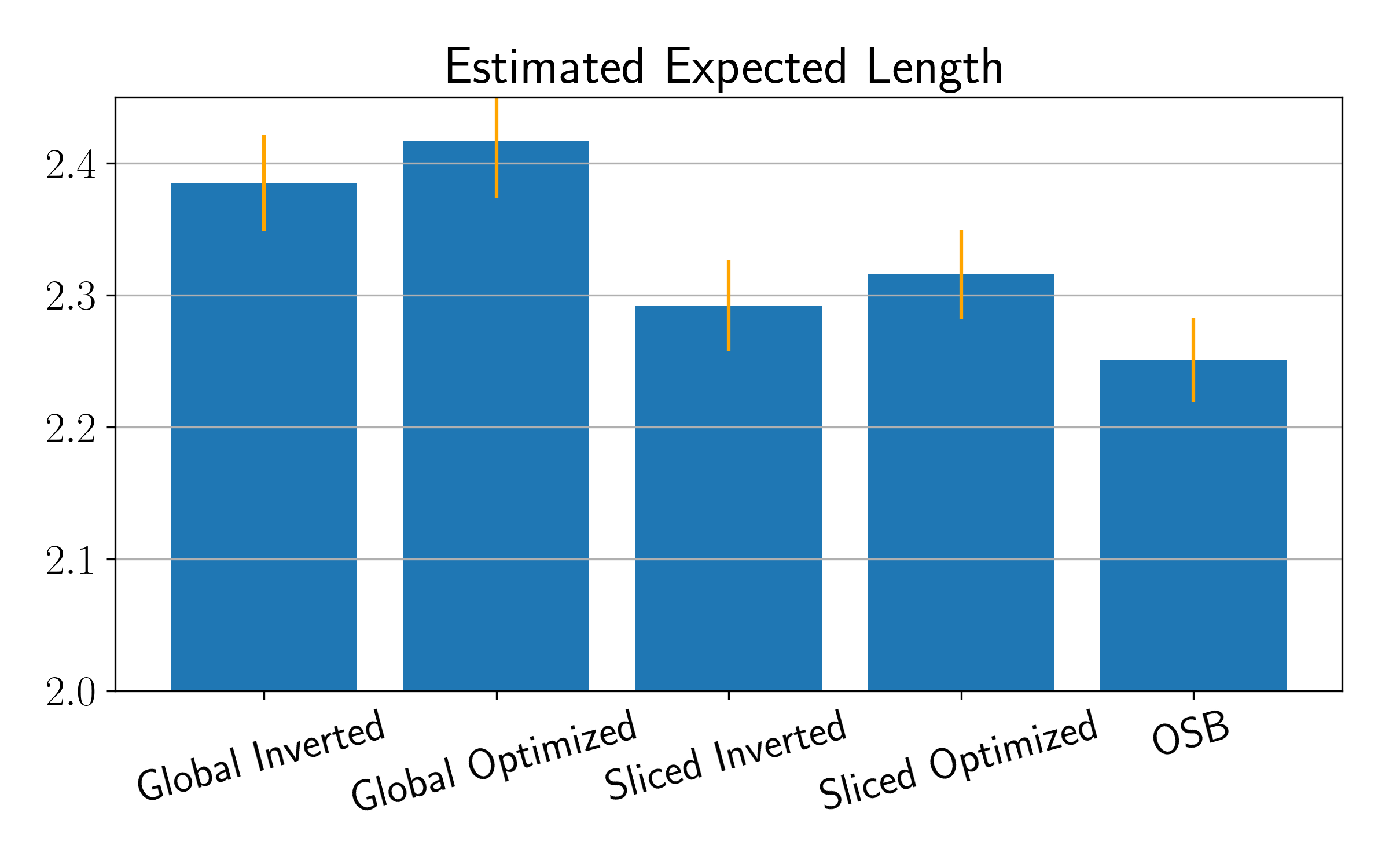}
    \caption[Estimated coverage and expected length for the two-dimensional constrained Gaussian example]{
    Estimated coverages and expected lengths across all four interval constructions and OSB for comparison at the $68$\% level for the two-dimensional constrained Gaussian setting. 
    All four of our interval constructions are comparable to OSB with respect to coverage, but OSB shows better expected length performance, aside from our two sliced interval constructions. 
    Although the OSB intervals are defined using the global max-quantile ($\globalQ$) and therefore can potentially be improved upon by limiting the considered parameter space via the \bergerboos set, due to the rapidity with which the $\alpha$-quantile surface meets the $\chi^2_{1, \alpha}$ quantile (see Figure 5.3 in \cite{previous_paper}), the OSB interval lengths are difficult to beat in practice.
}
    \label{fig:2d_coverage_length}
\end{figure*}

\begin{figure}
    \centering
    \includegraphics[width=0.32\textwidth]{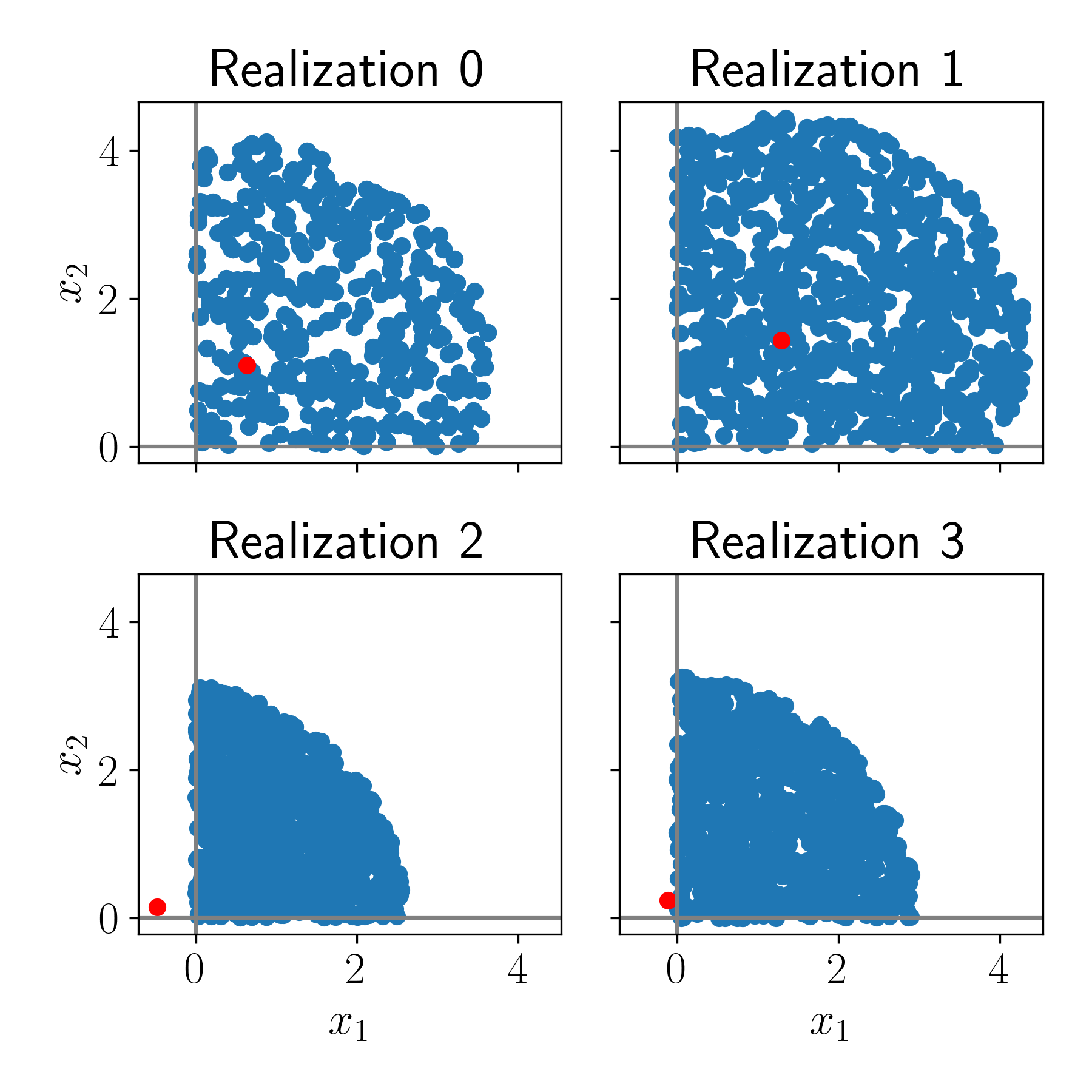}
    \includegraphics[width=0.32\textwidth]{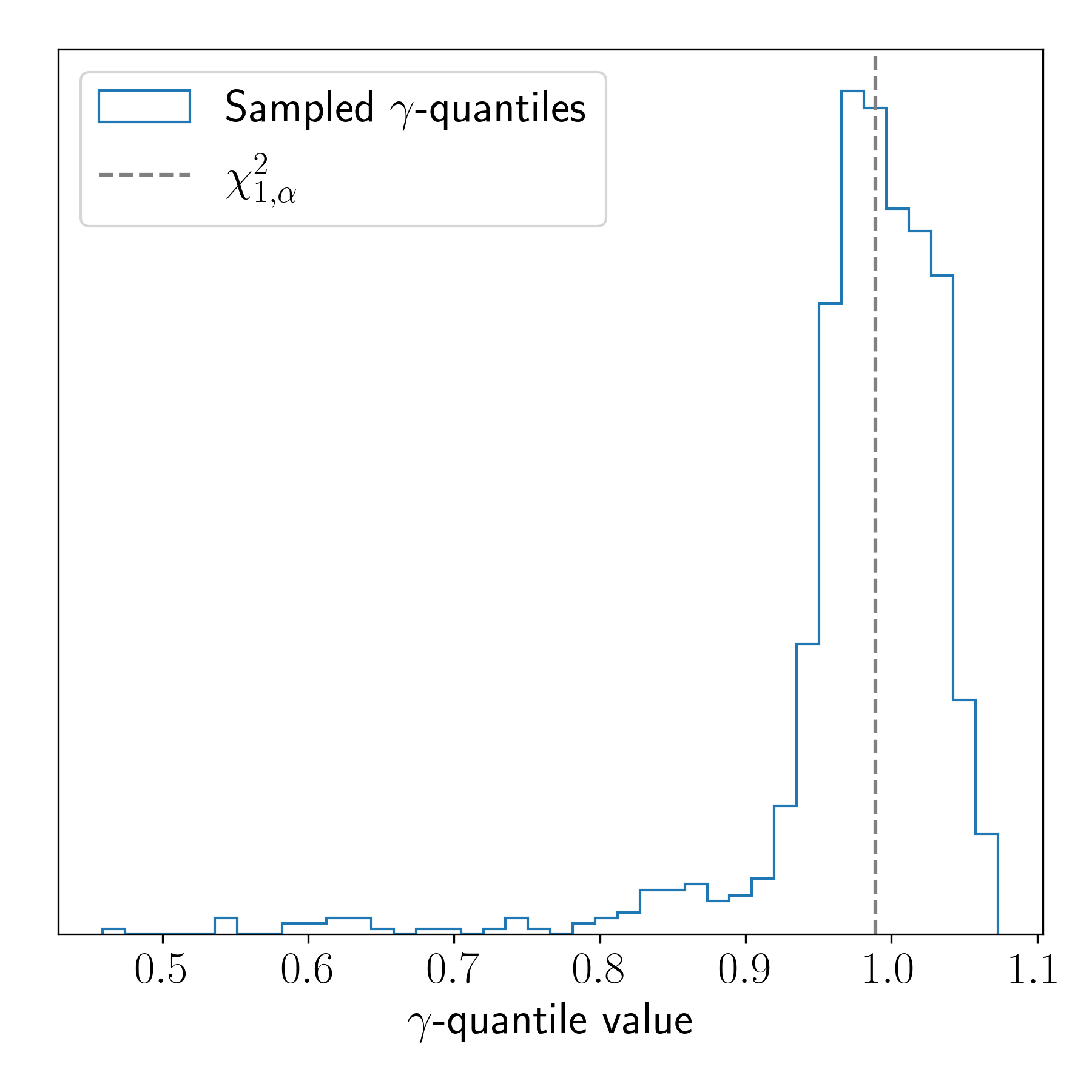}
    \includegraphics[width=0.32\textwidth]{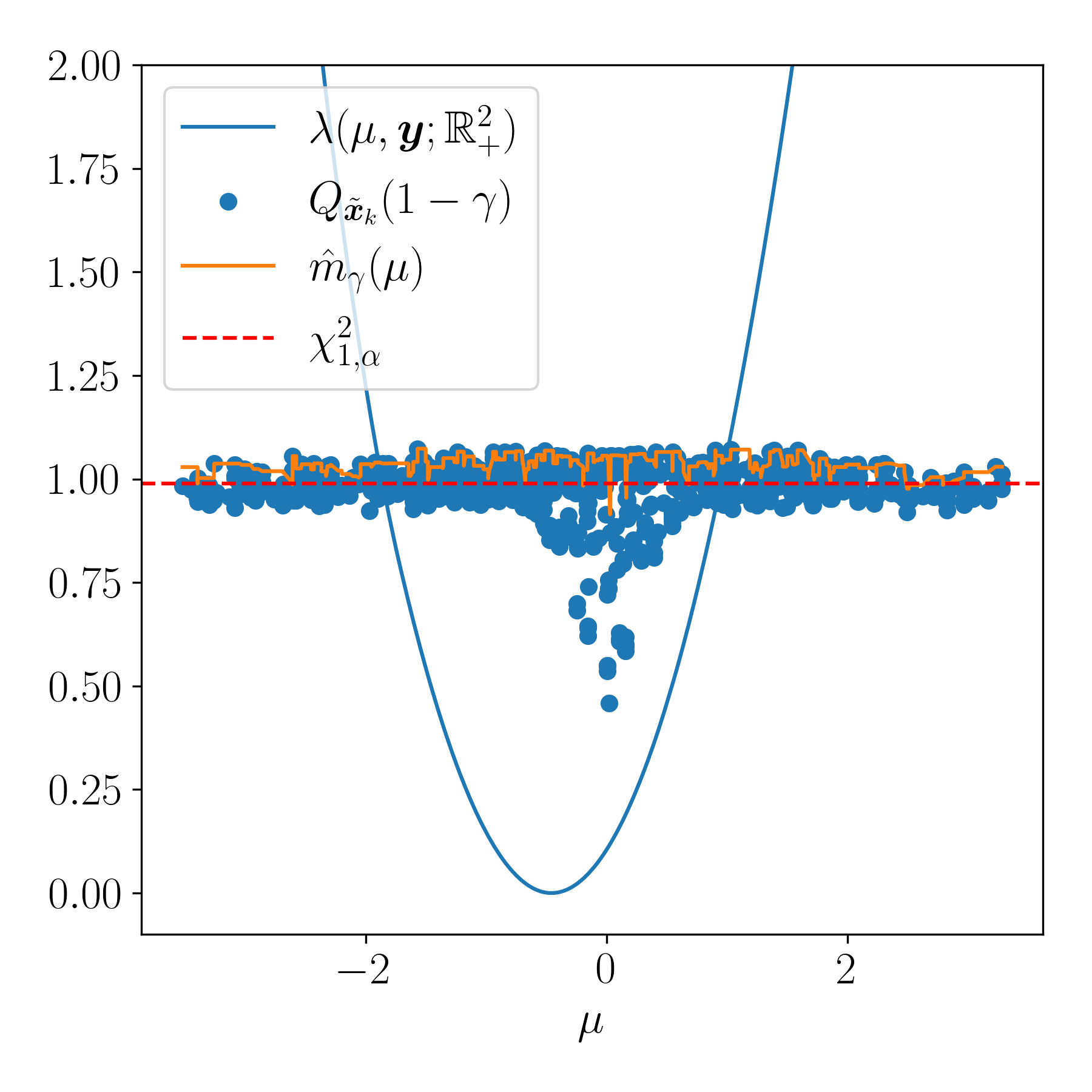}
    \caption[Berger--Boos samples, sampled quantiles, and interval cutoffs in the two-dimensional example]{
    (\textbf{Left}) Four realizations of the data-generating process where the observations are shown in red. 
    For each realization, the blue points are uniformly distributed samples from its \bergerboos set, sampled using the VGS sampler. 
    (\textbf{Center}) For a realization of the data-generating process, we plot the distribution of $\gamma$-quantiles for the points sampled by the VGS sampler. 
    Notably, a non-trivial percent of these are above $\chi^2_{1, \alpha}$ defining the OSB interval. 
    (\textbf{Right}) For the same realization, we plot the estimated sliced max-quantile function, $\hat{m}_\gamma(\mu)$ in orange alongside $\chi^2_{1, \alpha}$ in red. 
    The blue points correspond to sampled parameter values, each of which has a functional and quantile value, while the solid blue line shows the LLR over the functional varies. 
    All intervals can be read immediately from this image by inspecting where the blue LLR curve intersects the sampled points.
    }
    \label{fig:2d_bb_sets}
\end{figure}

\subsection{Constrained Gaussian in three dimensions}
\label{sec:3d_constrained_gaussian}

% As seen in the previous example, in a case where the OSB interval is known to achieve nominal coverage, its expected length can be difficult to beat.
% However, OSB coverage guarantee can be difficult to prove or disprove, since it amounts to proving stochastic dominance on the non-trivial LLR statistic.
In situations where the OSB coverage guarantee is difficult to prove, our intervals provide a clear theoretical advantage.
We apply our four intervals to one such case involving a three-dimensional constrained Gaussian case as explored in \cite{previous_paper}.
The three-dimensional Gaussian noise model is defined as follows:
\begin{equation}
    \by = \bx^* + \bepsilon, \quad \bepsilon \sim \mathcal{N}\left(\bm{0}, \bI_3 \right), \quad \bx^* \in \mathbb{R}^3_+,
\end{equation}
where $\varphi(\bx) = x_1 + x_2 - x_3$ and $\bx^* = \begin{pmatrix} 0.03 & 0.03 & 1 \end{pmatrix}^\top$.
The LLR is then defined as follows:
\begin{equation}
    \lambda(\mu, \by) = \min_{\substack{x_1 + x_2 - x_3 = \mu \\ \bx \in \mathbb{R}^3_+}} \lVert \by - \bx \rVert_2^2 - \min_{\bx \in \mathbb{R}^3_+} \lVert \by - \bx \rVert_2^2.
\end{equation}

\Cref{fig:3d_coverage_length} shows estimated coverage and expected length across all four interval constructions and the OSB interval.
While the OSB interval fails to attain nominal coverage, all four of our interval constructions do, with the Sliced constructions providing the best calibration.
% While the Global constructions pay a fairly steep price for coverage in expected interval length, the Sliced constructions navigate the trade-off well, paying for coverage with only slightly longer intervals compared to the OSB interval.
The Sliced constructions best navigate the trade-off between coverage and expected length, paying for coverage with only slightly longer intervals compared to the OSB interval.

\begin{figure*}[!ht]
    \includegraphics[width=0.49\textwidth]{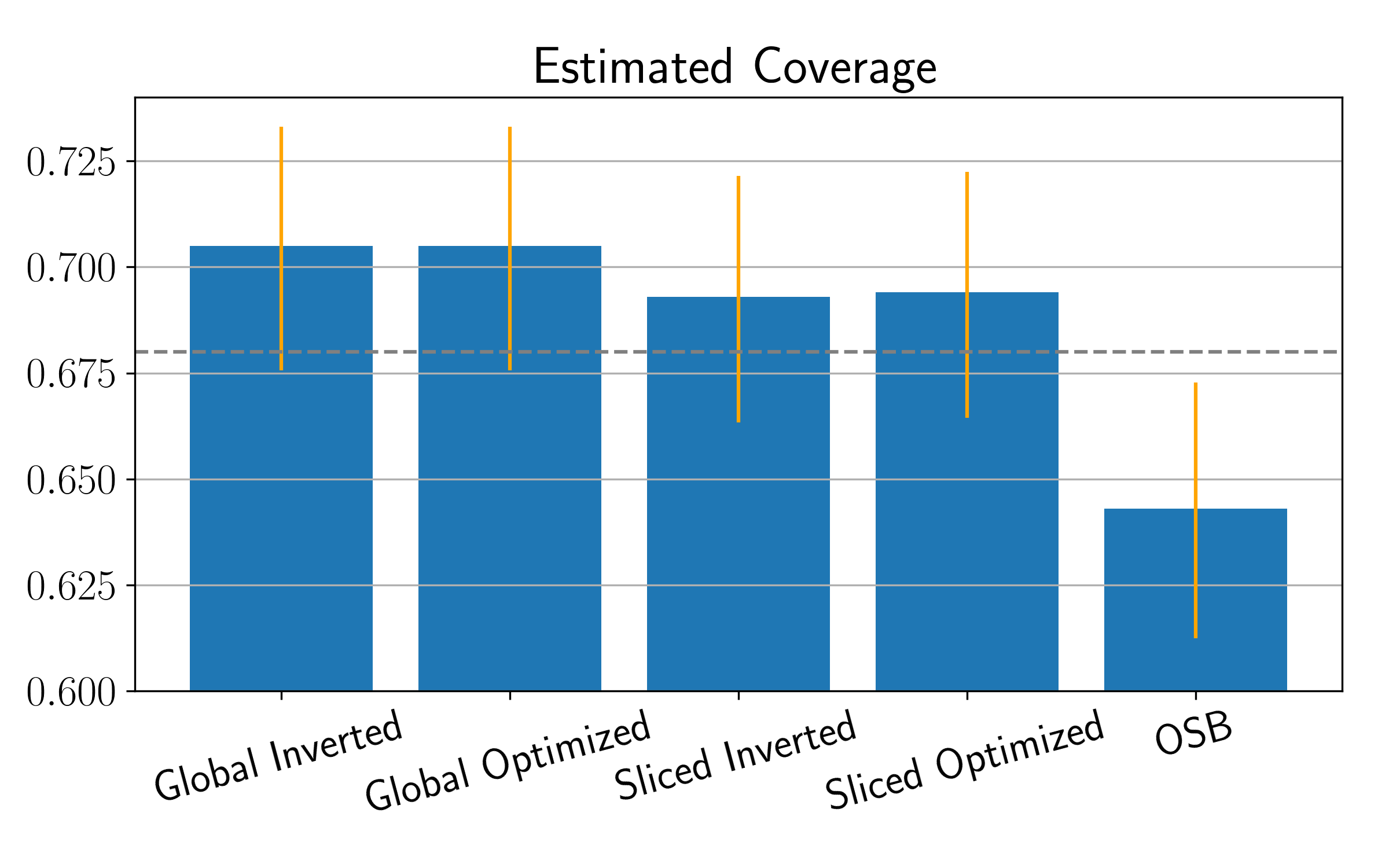}
    \includegraphics[width=0.49\textwidth]{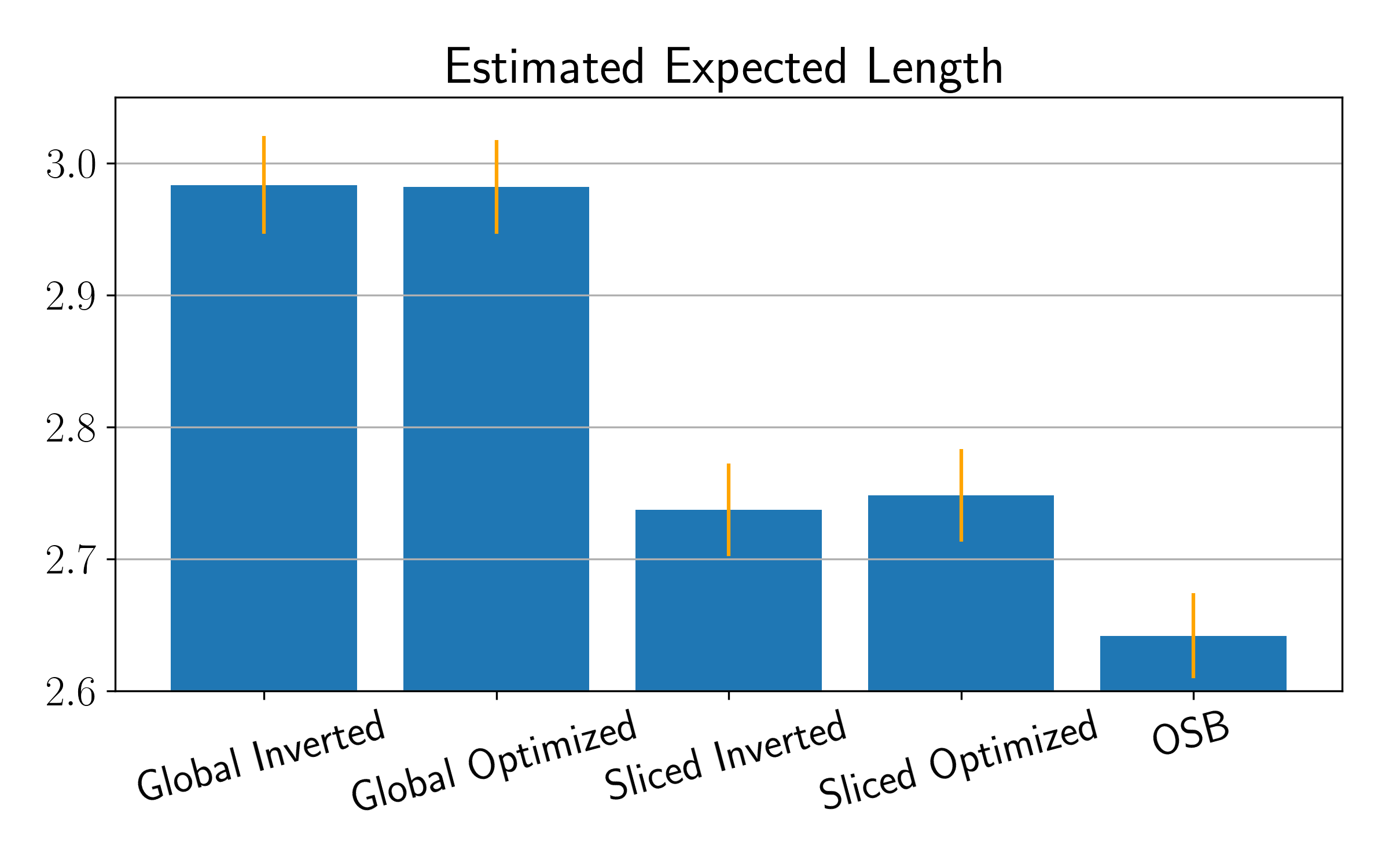}
    \caption[Estimated coverage and expected length for the three-dimensional constrained Gaussian example]{
    Estimated coverage and expected length across all four interval constructions and OSB for comparison at the $68$\% level for the three-dimensional constrained Gaussian example. 
    All four of our interval constructions achieve nominal coverage, while the OSB interval does not. 
    While the Global interval constructions pay a steep price in expected length compared to OSB, the Sliced constructions are only slightly longer than OSB.
    }
    \label{fig:3d_coverage_length}
\end{figure*}

The setting of $\bx^*$ used here is slightly different than that of \cite{previous_paper}, where $\bx^* = \begin{pmatrix} 0 & 0 & 1 \end{pmatrix}^\top$ was used.
In \cite{previous_paper}, this setting was used as a counter-example for OSB coverage, since $Q_{\bx^*}(1 - \alpha) > \chi^2_{1, \alpha}$ for at least some $\alpha \in (0, 1)$.
However, as shown in Figure 5.5 of \cite{previous_paper}, when $\alpha = 0.05$, $Q_{\bx^*}(1 - \alpha) \leq \chi^2_{1, \alpha}$ for $\bx^* = \begin{pmatrix} t & t & 1 \end{pmatrix}^\top$ when $t$ is approximately greater than $e^{-2} \approx 0.135$, which indicates that parameter settings violating stochastic dominance by $\chi^2_1$ exist close to the parameter constraint boundary.
The location of these key parameter settings presents a challenge for the Polytope sampler (\Cref{alg:polytope_sampler}).
In \Cref{app:sec:importance_like_sampler}, \Cref{alg:3d_sampler} presents a modified version of \Cref{alg:polytope_sampler} that better handles this situation.

\myparagraph{\bergerboos set experiment.}
For this model, we investigate the effect of changing the parameter $\eta$ that controls the \bergerboos construction in the global interval.
We compute intervals for different $\by$ and fixed $\bx^* = (2,2,0), \bx^* = (3,3,0)$ and $\bx^* = (5,5,0)$ as $\eta$ ranges between $0$ and $\alpha = 0.32$.
The lengths of such intervals, averaged over the data $\by$, are shown in \Cref{fig:BBSet_combined}.
In this example the maximum quantile is achieved at $\bx = (0,0,t)$ for large $t$, and, as expected, the benefit of using a small $\eta > 0$ becomes more pronounced as the true $\bx^*$ becomes farther from the point that achieves the maximum quantile. 
However, as $\eta$ grows too close to $\alpha$ the downside of optimizing the $1-\alpha+\eta$ quantile instead of $1-\alpha$ outweighs the benefit of optimizing it over a smaller set, resulting in larger intervals.
This suggests that a small $\eta > 0$ is a reasonable default, as proposed originally by \cite{berger_boos}.

\begin{figure}[ht]
    \centering
    \includegraphics[width=0.9\columnwidth]{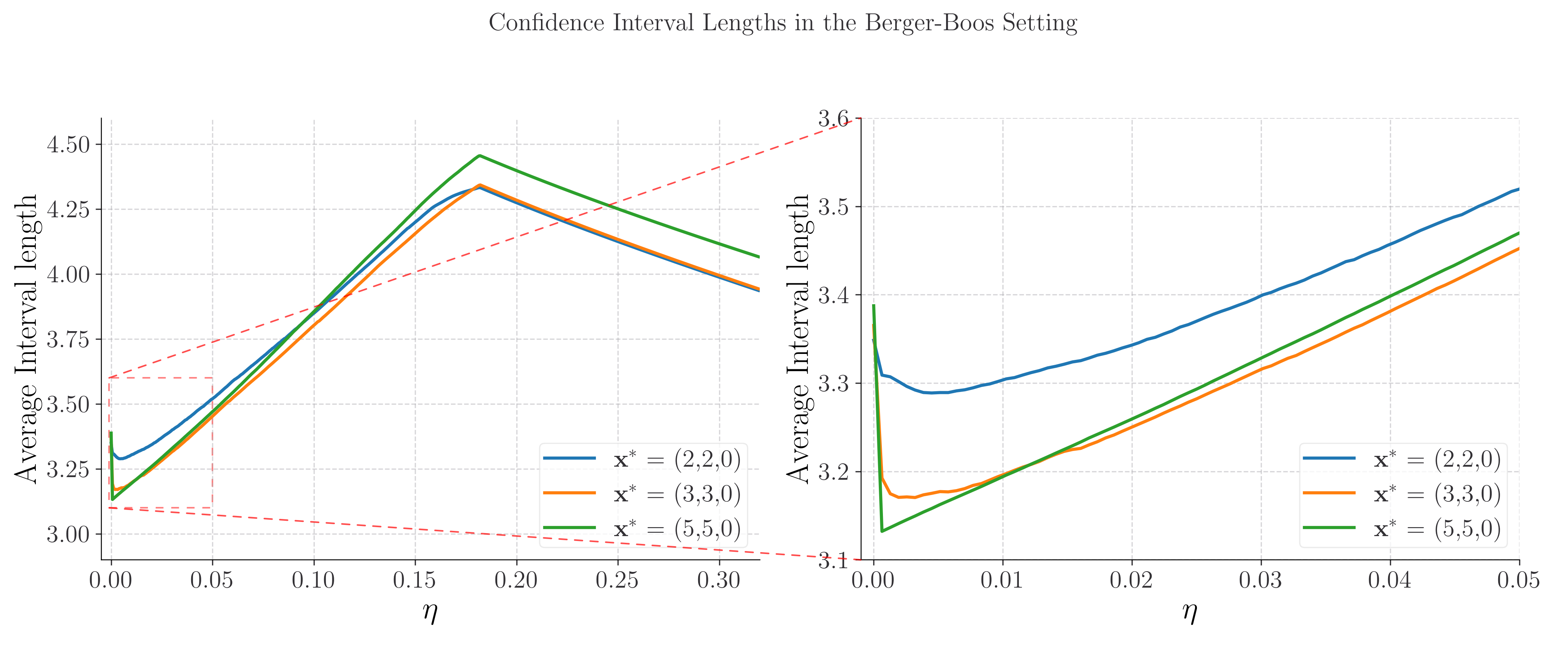}
    \caption[Effect of the Berger--Boos parameter on average interval length]{
    Confidence interval lengths in the \bergerboos setting, averaged over values of $\by$, for varying $\eta$ and $\bx^*$. 
    The minimum average length occurs at a small $\eta > 0$, showing that the construction is beneficial if $\eta$ is tuned correctly. 
    This occurs because even for moderately small $\eta$, the \bergerboos set, which is a three-dimensional sphere intersected with the non-negative orthant, avoids the point with the highest $1-\alpha$ quantile.
    }
    \label{fig:BBSet_combined}
\end{figure}

\subsection{Wide-bin deconvolution}
\label{sec:wide_bin_deconvolution}
 
This section shows the superior performance of our interval constructions relative to OSB in a more complex high-dimensional setting by considering the problem of computing a confidence interval for the sum of adjacent bins of a deconvolved histogram as described in \cite{stanley_unfolding}.
We deliberately use this previously studied model as a benchmark: the forward model is not the novelty of the present paper, but keeping the application model fixed allows the effect of the new data-adaptive calibration and sampling-based interval construction to be isolated and compared directly with OSB and related optimization-based intervals.
This problem is a core statistical problem of particle unfolding in high-energy physics.
We refer the readers to the following references for more detailed information \cite{kuusela_panaretos, kuusela_phd_thesis, kuusela_stark, cms_2016, cms_2019}.

The data-generating process is linear with Gaussian noise,
\begin{equation} \label{eq:unfolding}
    \by = \bK \bx^* + \bepsilon, \quad \bepsilon \sim \mathcal{N}(\bm{0}, \bI), \quad \bx^* \geq \bm{0},
\end{equation}
where $\bK \in \mathbb{R}^{40 \times 80}$ and $\varphi(\bx) = \bh^\top \bx$.
The vector $\bh$ defines the bin-adjacent aggregation.
In particle unfolding, the vectors $\bx^*$ and $\by$ represent particle counts within discretized bins.

We note that the dimension of these discretized bins is directly related to the resolution of the scientific inferences that can be made from the system, and it is advantageous to set the discretization so that the dimension of $\bx^*$ is significantly larger than the dimension of $\by$.
Other approaches typically cannot use this high-dimensional discretization and still achieve calibrated UQ on $\varphi(\bx^*)$, since they require regularization and thus incur estimator bias.

The LLR is then defined as follows:
\begin{equation}
    \lambda(\mu, \by) = \min_{\substack{\bh^\top \bx = \mu \\ \bx \in \mathbb{R}^{80}_+}} \lVert \by - \bK \bx \rVert_2^2 - \min_{\bx \in \mathbb{R}^{80}_+} \lVert \by - \bK \bx \rVert_2^2.
\end{equation}
% We emphasize two features of this setup that complicate the task of computing confidence intervals for $\varphi(\bx)$.
The forward model, $\bK$, has a non-trivial null space and a large condition number, making inverse problem point estimation and inference challenging (see \cite{kuusela_phd_thesis} for more details).
% Typically, this sort of ill-posedness is handled with regularization of some kind, but as is well-known in the inverse problem literature and specifically shown in \cite{kuusela_phd_thesis}, including such regularization induces a bias, which can undercut desired statistical guarantees (e.g., coverage) of the inference object of interest.
% Including constraints and focusing on a particular functional of the parameter vector \emph{implicitly} regularizes the problem \citep{patil, stanley_unfolding}, but shifts the problem difficulty to inference with constraints.
Although \cite{previous_paper} and this paper propose a theoretical framework to perform inference with constraints, implementation is challenging in practice due to the high-dimensional parameter space of this scenario.
As we show in the following sections, the Polytope sampler (\Cref{alg:polytope_sampler}) and quantile regression do an adequate job producing samples and fitting quantile surfaces in this high-dimensional space to ensure the desired coverage of the interval constructions.

As extensively discussed in \cite{stanley_unfolding}, the OSB interval (i.e., using $\chi^2_{1, \alpha}$ in the optimization-based interval construction) produces empirically valid (though typically conservative) confidence intervals in all tested scenarios.
% In the scenarios considered in \cite{stanley_unfolding}, the underlying function generating the true histogram means ($\bx^* \in \mathbb{R}^{80}_+$) was relatively smooth, likely contributing to the over-coverage.
Since the over-coverage in \cite{stanley_unfolding} is likely the result of relatively smooth mean vector ($\bx^* \in \mathbb{R}^{80}_+$), we present two true parameter settings for $\bx^*$ in \eqref{eq:unfolding} to highlight two advantages of our interval constructions over the OSB interval.
First, we use the original smooth parameter setting from \cite{stanley_unfolding} to show improved over-coverage relative to the OSB interval via a reduction in the expected interval length.
Second, we present an ``adversarial'' setting where our interval constructions achieve nominal coverage while the OSB interval does not.
% \Cref{fig:unfolding_true_settings} shows the smooth and adversarial settings for $\bx^*$.
We constructed the adversarial setting (see \Cref{fig:unfolding_true_settings}) by first computing our interval constructions on the smooth setting and then looking at the maximum out-of-sample predicted quantile for a generated observation with a large predicted quantile.
For each observation drawn within both settings, we draw $2.1 \times 10^{4}$ samples using the Polytope sampler as described by \Cref{alg:polytope_sampler}.

\begin{SCfigure}[50][!ht]
    \centering
    \includegraphics[width=0.5\columnwidth]{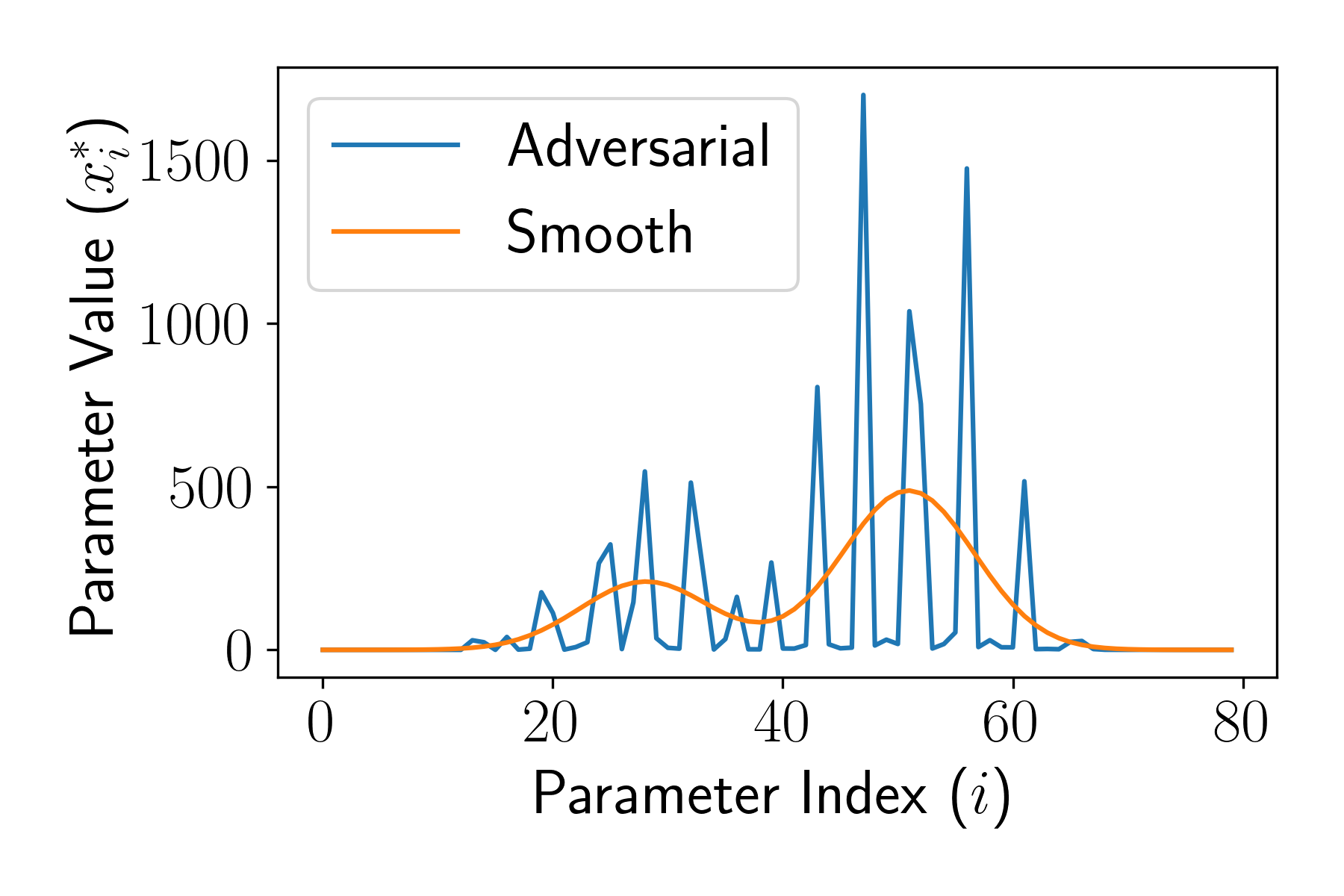}
    \caption[Smooth and adversarial parameter settings for the wide-bin deconvolution experiments]{
    Parameter values for the smooth and adversarial settings for $\bx^*$ used to illustrate our interval construction versus the OSB interval. 
    The adversarial setting is made more difficult by the sharp jumps in parameter values.
    \vspace{5em}
    }
    \label{fig:unfolding_true_settings}
\end{SCfigure}

\subsubsection{Smooth setting}
\label{sec:unfolding_realistic}
% Using the smooth $\bx^*$ shown in \Cref{fig:unfolding_true_settings}, \cite{stanley_unfolding} showed that the OSB interval over-covers at the $95$\% level.
% Furthermore, it was shown that the OSB interval was the shortest across a range of other interval options, including SSB, prior optimized, and minimax.
We show that our interval constructions achieve nominal coverage and the sliced constructions reduce over-coverage compared to OSB by producing shorter intervals on average.
Estimated coverage and expected interval lengths are shown in \Cref{fig:unfolding_coverage_length}.

Both Global interval constructions and OSB dramatically over-cover, which highlights the conservatism of the Global constructions.
These estimated coverage values indicate that within each observation's \bergerboos set, there is a parameter setting against which the method has to protect that is substantially more difficult to cover than the true realistic parameter setting.
% Interestingly, both Global constructions produce markedly longer intervals on average compared with the OSB interval.
% Aligning with the intuition from the Global and Sliced construction definitions, t
The Sliced intervals show lower over-coverage and significantly smaller average lengths, and both Sliced constructions are shorter on average compared to the OSB interval, with the Sliced Inverted interval showing an $18.7$\% reduction in average length and the Sliced Optimized showing a $11.1$\% reduction in average length.

\begin{figure*}[!ht]
    \includegraphics[width=0.49\textwidth]{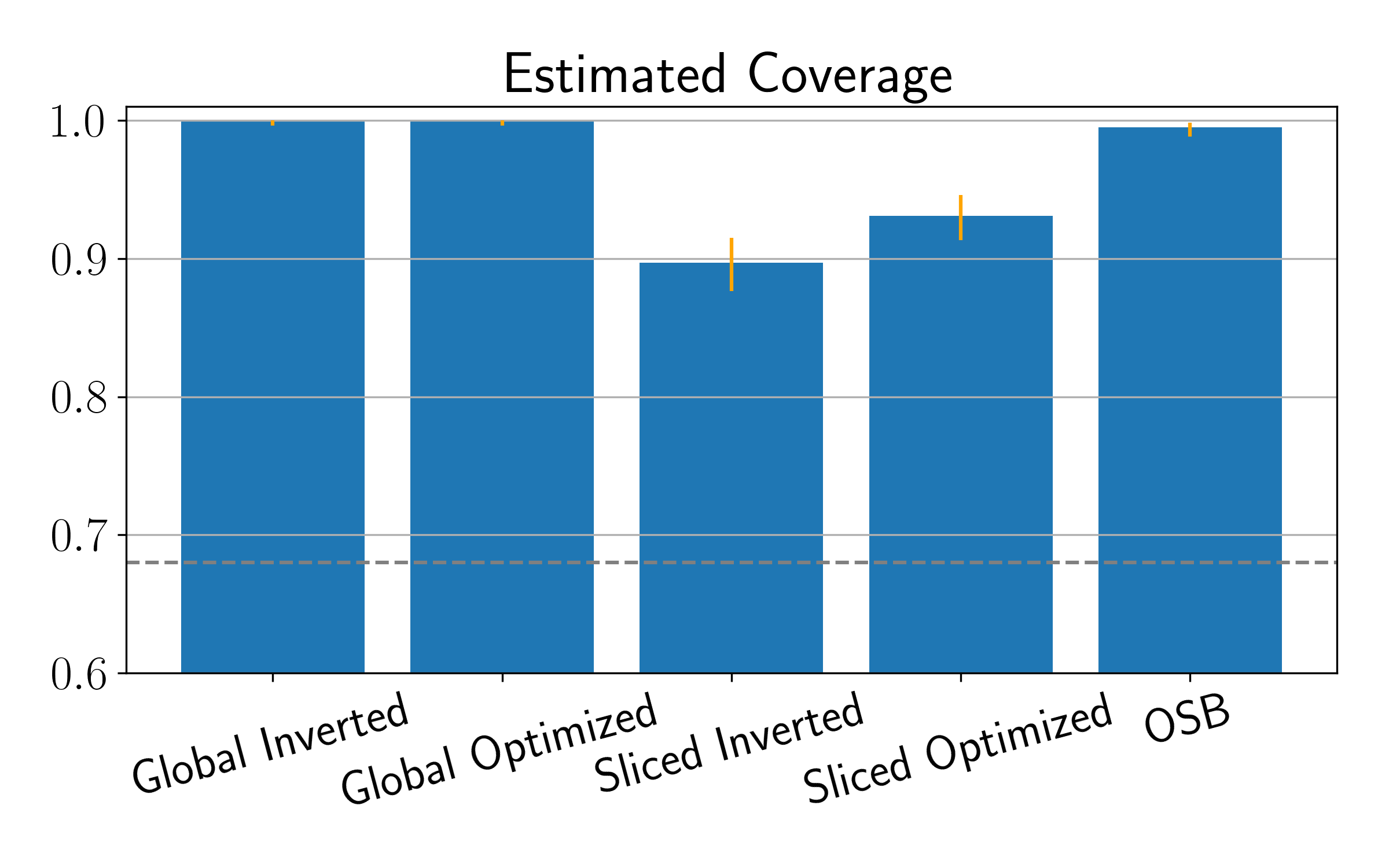}
    \includegraphics[width=0.49\textwidth]{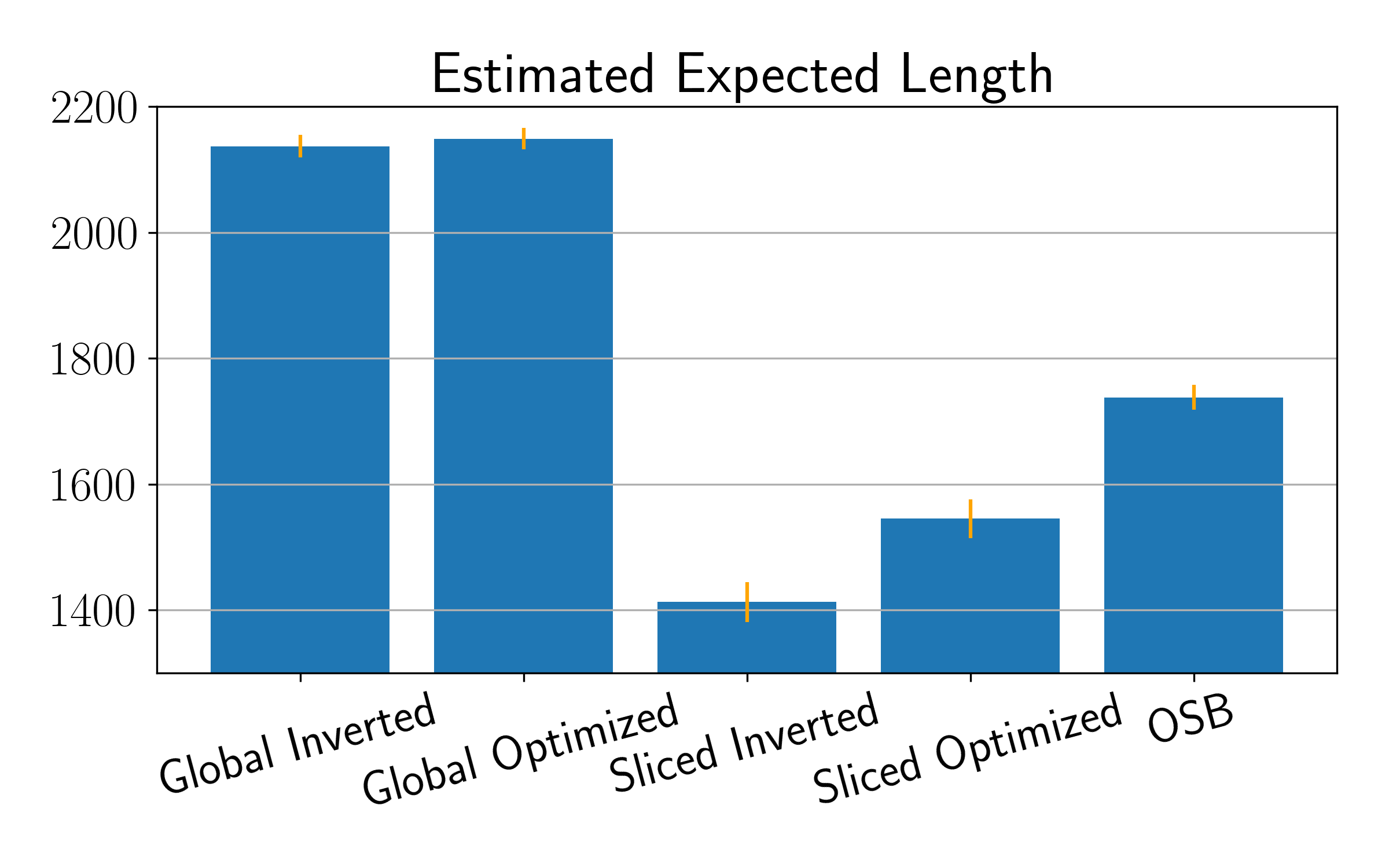}
    \caption[Estimated coverage and expected length for the smooth wide-bin deconvolution experiment]{
    Estimated coverage and expected length across all four interval constructions and OSB at the $68$\% level for the \textbf{smooth} wide-bin deconvolution experiment. 
    While the Global interval constructions over-cover like the OSB interval, the Sliced interval constructions reduce both over-coverage and expected interval length.
    }
    \label{fig:unfolding_coverage_length}
\end{figure*}

\subsubsection{Adversarial setting}
\label{sec:unfolding_adversarial}
% The key result in \Cref{sec:unfolding_realistic} is that the Sliced interval constructions both reduce over-coverage and expected interval length compared with the OSB interval.
We show that for the adversarial parameter setting, the OSB interval does not achieve nominal coverage, whereas all four of our interval constructions do while still reducing the expected interval length in the case of the Sliced interval constructions compared to the OSB interval (see \Cref{fig:unfolding_coverage_length_adversarial} for primary results).
% The corresponding estimated coverage and expected length results are shown in \Cref{fig:unfolding_coverage_length_adversarial}.

Both Global interval constructions and the Sliced Optimized interval over-cover, while the Sliced Inverted intervals achieve nominal coverage within statistical uncertainty.
% with the Sliced Optimized over-covering to a lesser extent than the Global intervals.
% The Sliced Inverted interval achieves nominal coverage within the statistical uncertainty.
The estimated expected lengths are similar to those of the smooth example, with the Global intervals showing the longest average interval lengths, the Sliced intervals showing the shortest, and the OSB interval being between the two.
Importantly, the Sliced intervals are again significantly shorter than the OSB interval, even though the OSB interval does not achieve nominal coverage.
The Sliced Inverted interval shows a $18.9$\% average interval length reduction over OSB while the Sliced Optimized interval shows an $11.4$\% average interval length reduction.

\begin{figure*}[!ht]
    \includegraphics[width=0.49\textwidth]{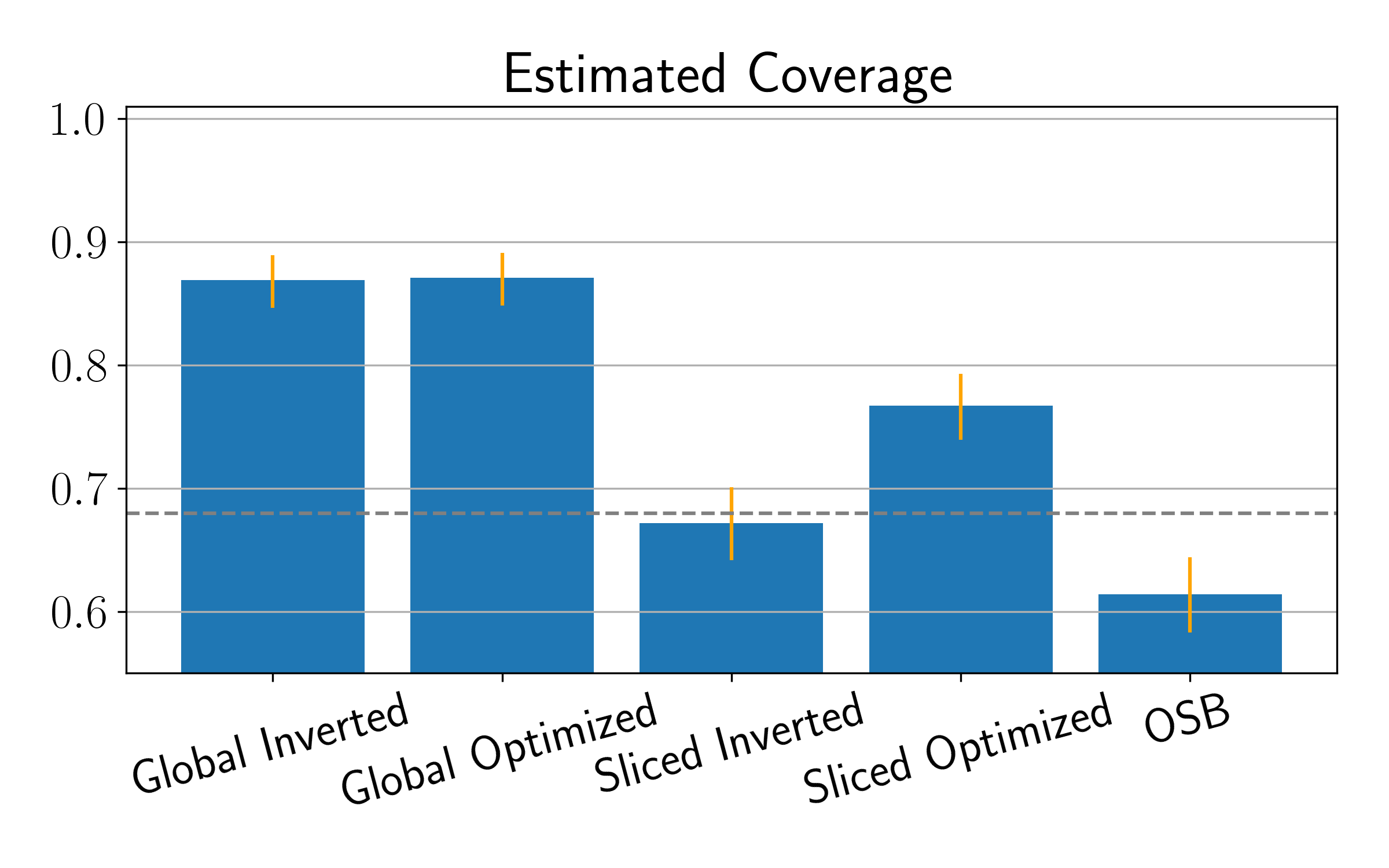}
    \includegraphics[width=0.49\textwidth]{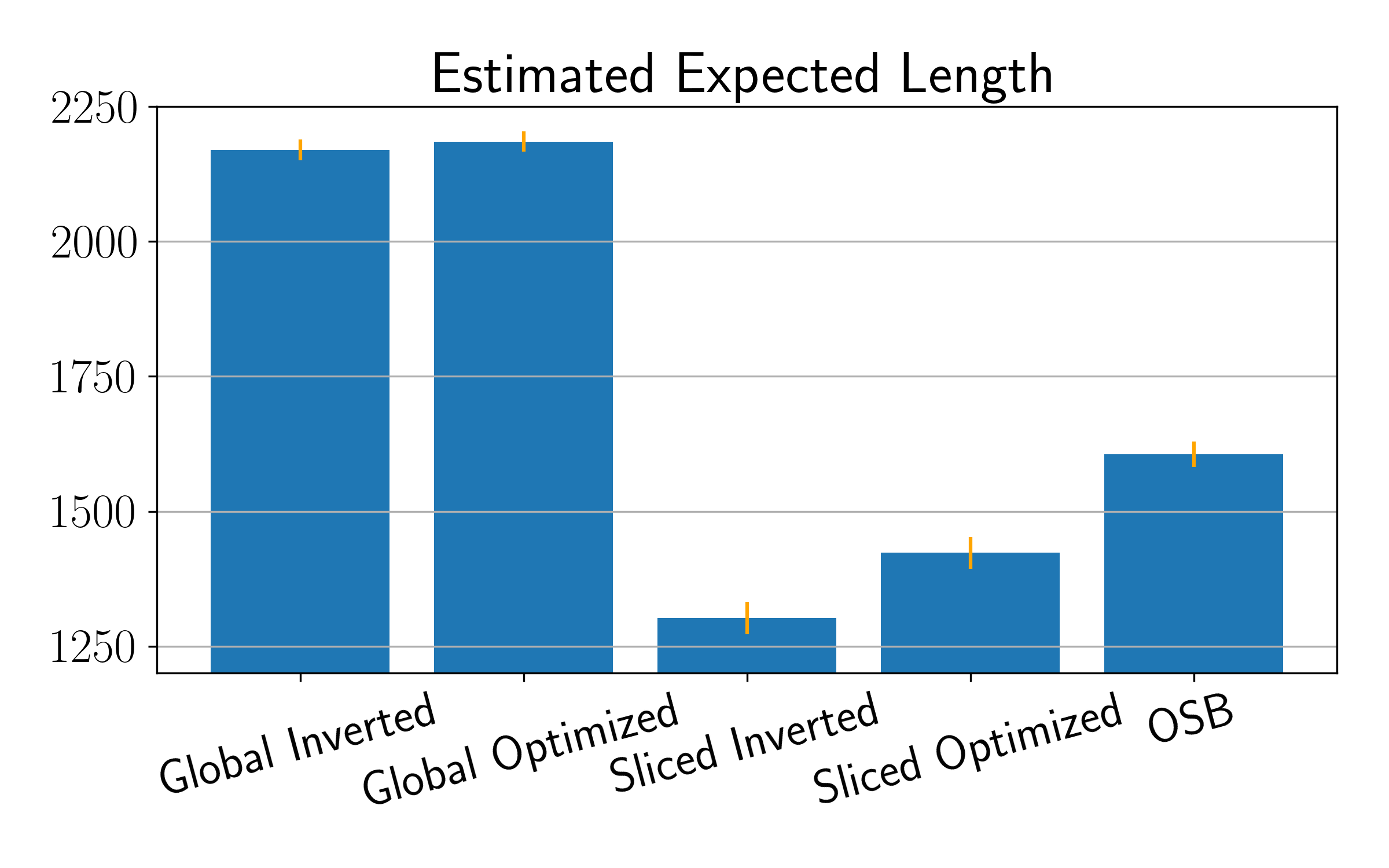}
    \caption[Estimated coverage and expected length for the adversarial wide-bin deconvolution experiment]{
    Estimated coverage and expected length across all four interval constructions and OSB at the $68$\% level for the \textbf{adversarial} wide-bin deconvolution experiment. 
    While the OSB interval fails to achieve nominal coverage, all four of our interval constructions do. 
    Interestingly, the Sliced interval constructions are meaningfully shorter than the OSB interval while also providing coverage.
    }
    \label{fig:unfolding_coverage_length_adversarial}
\end{figure*}

For both experiments, we note the differences in coverage and expected length between the two Sliced intervals.
Although both approaches theoretically compute the same interval, significant differences can arise from their construction.
The Sliced Inverted intervals are constructed by accepting individual functional values, which makes each point's acceptance dependent upon the quality of the quantile regressor at that point.
By contrast, since the Sliced Optimized intervals essentially smooth over the max quantiles as a function of the functional space, the intervals are less sensitive to quantile regressor performance at any individual point.
In the above experiments, there were realizations of the data for which the Sliced Inverted construction only accepted a single functional value sample, thereby making the confidence interval a single point.
We found the Sliced Optimized construction to be more robust in these settings as the max quantile values were shared in a sliding window over the functional space. 

\section{Conclusion and future work}
\label{sec:discussion_and_conclusion}

This paper proposes several confidence interval constructions for functionals in constrained ill-posed inverse problems. 
Our approach is based on two key ideas: data-adaptive constraints using a \bergerboos construction and sampling-based inversion.
Two independent decisions when constructing intervals provide four different valid intervals: Global versus Sliced, using the quantile function either over the entire or along level-set slices of the \bergerboos set and Inverted versus Optimized, constructing the interval either by individually accepted functional values via the estimated quantile function or using the estimated quantile function in endpoint optimizations.
All of the constructions are built upon the preliminary constraint by the data-informed \bergerboos set, followed by a sampling procedure to estimate a quantile function that can be used to invert or optimize the interval endpoints.
We have validated the method (including all four aforementioned interval constructions) through several numerical examples, demonstrating its ability to provide correct coverage, better calibration, and comparable or shorter interval length compared to the OSB interval baseline. 
Overall, our approach offers a flexible framework that can incorporate constraints directly and can be tailored to various types of inverse problems. 
The main takeaway is that data-adaptive constraining helps improve the length of the resulting confidence intervals and enables sampling, which makes it feasible to carry out the test inversion needed to construct confidence intervals with a desired nominal coverage.

There are several promising directions for future work.
One immediate direction is principled automated guidance for choosing the Berger--Boos parameter $\eta$. The finite-sample guarantee holds for any pre-specified $\eta \in (0,\alpha)$ with $\gamma \leq \alpha-\eta$, but the expected-length-optimal value is problem-dependent because it depends on both the geometry of $\mathcal{B}_\eta$ and the shape of the LLR quantile surface. Our experiments support using a small positive value and, when feasible, checking sensitivity over a small grid of pre-specified candidate values. In this paper, for example, the $68\%$ experiments use $\alpha=0.32$ and $\eta=0.01$, so $\eta$ uses only a small fraction of the total miscoverage budget while still making the data-adaptive set useful computationally. We do not claim that this value is universally optimal. Developing a fully automated, data-adaptive choice of $\eta$ that optimizes interval length while preserving the coverage guarantee is left for future work.
Another direction is to find ways to extend our method to even higher-dimensional problems, which are more challenging.
This would involve developing improved techniques to handle the curse of dimensionality and exploring the trade-off between accuracy and computational complexity.
For the approaches in this paper in particular, this extension would require a more tailored sampling approach.
Another direction is to leverage more sophisticated machine learning algorithms (deep learning models or ensemble methods) to improve the estimate of the quantile function and, thus, improve the accuracy and efficiency of our confidence intervals. 
Additionally, applying our approach to other applications involving ill-posed inverse problems, such as medical imaging or geophysics, would provide further validation of the effectiveness of our strategy. 
Finally, it is of interest to conduct a further theoretical analysis of our approach under different constraints and noise conditions to better understand its limitations and strengths. 
This would involve studying the statistical properties of our confidence intervals and investigating the impact of various assumptions on their performance.
Of particular interest are relaxations of the linear forward model and Gaussian noise assumptions to extend our method's application domain.
Although the original theoretical foundation developed in \cite{previous_paper} does not make these assumptions, our implementation relies upon them for the tractability of the sampling algorithms.
Overall, these future research directions have the potential to demonstrate the applicability and robustness of our approach in a wide range of domains.

\section*{Acknowledgment}
We thank members of the STAMPS (Statistical Methods for the Physical Sciences) Research Center at Carnegie Mellon University for fruitful discussions on this work.     
MS and MK were partially supported by NSF grants DMS-2053804 and PHY-2020295, JPL RSAs No.~1670375, 1689177 \& 1704914 and a grant from the C3.AI Digital Transformation Institute.
PB and HO acknowledge support from the Air Force Office of Scientific Research under MURI awards number FA9550-20-1-0358, FOA-AFRL-AFOSR-2023-0004 and by the Department of Energy under award number DE-SC0023163. Additionally HO acknowledge support from the DoD Vannevar Bush Faculty Fellowship Program.

% \section*{Supplementary material}
% \label{SM}

% \edit{To keep the main article focused while retaining reproducibility, the Supplementary Material collects the extended proofs, sampler algorithms, quantile-regression details, and numerical diagnostics supporting the main construction.}
% The Supplementary Material includes: 
% proofs from \Cref{sec:interval_methodology,sec:theory} in \Cref{app:sec:interval_methodology_proofs,proofs_appendix}, 
% additional algorithmic details and illustrations for sampling the Berger–Boos set (VGS and Polytope samplers) in \Cref{app:sec:sampling_bergerboos}, 
% quantile regression overview in \Cref{app:sec:quantile_regression_overview}, 
% and extended numerical details and illustrations, including an importance-like sampler used in numerical experiments in \Cref{sec:additional-details-sec:numerical_exp}.

\bibliographystyle{siamplain}
\bibliography{references}

\clearpage
\appendix

\newcommand{\removelinebreaks}[1]{\def\\{\relax}#1}
\def\titleRLB{\removelinebreaks{\titletext}}

\setcounter{equation}{0}
\setcounter{figure}{0}
\renewcommand{\theequation}{S.\arabic{equation}}
\renewcommand{\thefigure}{S.\arabic{figure}}

\newgeometry{left=0.5in,top=0.5in,right=0.5in,bottom=0.5in,head=0.1in,foot=0.1in}

\begin{center}
\Large
{
{\bf
\framebox{Appendix}}
}
\end{center}

\bigskip

This document serves as an appendix to the paper ``\titleRLB''.
Below we provide an organization for the appendix, followed by a summary of the main notation used in both the paper and the appendix.
The equation and figure numbers in this appendix begin with the letter ``S'' to differentiate them from those appearing in the main paper.

\section*{Organization}

% The content of this supplement is organized as follows.
A roadmap for the supplement is provided in \Cref{tab:supplement-organization}.

\begin{table}[!ht]
\centering
\caption{Roadmap of the supplement.}
\begin{tabularx}{\textwidth}{L{3cm}L{15cm}}
\toprule
\textbf{Appendix} & \textbf{Description} \\
\midrule
\Cref{app:sec:interval_methodology_proofs} & Proofs in \Cref{sec:interval_methodology} (Proofs of \Cref{lem:setting_level_parameters_and_coverage} and \Cref{cor:eta-gamma-alpha-adjustment-coverage}) \\
\addlinespace[0.5ex] \arrayrulecolor{black!25} \midrule
\Cref{proofs_appendix} & Proofs in \Cref{sec:theory} (Proof of \Cref{thm:maintheoremalgs}) \\
\addlinespace[0.5ex] \arrayrulecolor{black!25} \midrule
\Cref{app:sec:sampling_bergerboos} & Additional details and illustrations in \Cref{subsec:preimage_sampling} (VGS and Polytope samplers) \\
\addlinespace[0.5ex] \arrayrulecolor{black!25} \midrule
\Cref{app:sec:quantile_regression_overview} & Additional details in \Cref{subsec:quantile_regression} (quantile regression) \\
\addlinespace[0.5ex] \arrayrulecolor{black!25} \midrule
\Cref{sec:additional-details-sec:numerical_exp} & Additional details and illustrations in \Cref{sec:numerical_exp} (importance-like sampler) \\
\addlinespace[0.5ex] \arrayrulecolor{black} \bottomrule
\end{tabularx}
\label{tab:supplement-organization}
\end{table}

\section*{Notation}

% An overview of the main notation used in this paper is as follows.
An overview of the main notation used in this paper is provided in \Cref{tab:notations}.

\begin{table}[ht]
\centering
\caption{Summary of main notation used in the paper and the appendix.}
\begin{tabularx}{\textwidth}{L{3cm}L{15cm}}
\toprule
\textbf{Notation} & \textbf{Description} \\
\midrule
Non-bold lower or upper case & Denotes scalars (e.g., $\alpha$, $\mu$, $Q$). \\
Bold lower case & Denotes vectors (e.g., $\bx$, $\by$, $\bh$). \\
Bold upper case & Denotes matrices (e.g., $\bK$, $\bI$). \\
Calligraphic font & Denotes sets (e.g., $\mathcal{X}$, $\mathcal{C}$, $\mathcal{D}$). \\
\arrayrulecolor{black!25}\midrule
$\mathbb{R}$ & Set of real numbers. \\
$\mathbb{R}_+$ & Set of non-negative real numbers. \\
$[n]$ & Set $\{1, \dots, n\}$ for a positive integer $n$. \\
\midrule
$\| \bm{u} \|_{2}$ & The $\ell_2$ norm of vector $\bm{u}$. \\
$\| f \|_{L_2}$ & The $L_2$ norm of function $f$. \\
\midrule
$\bm{K}^\top$ & The transpose of a matrix $\bK \in \RR^{m \times p}$. \\
$\bI_m$ or $\bI$ & The $m \times m$ identity matrix. \\
$\bm{v} \le \bm{u}$ & Componentwise ordering for vectors $\bm{v}$ and $\bm{u}$. \\
$\bm{A} \preceq \bm{B}$ & The Loewner ordering for symmetric matrices $\bm{A}$ and $\bm{B}$.\\
\midrule
$\mathbf{1}\{A\}$ & Indicator random variable associated with event $A$. \\
$Y = \cO_\alpha(X)$ & Deterministic big-O notation, indicating that $Y$ is bounded by $| Y | \le C_\alpha X$. \\
$C_\alpha$ & A numerical constant that may depend on the parameter $\alpha$ in context. \\
$\cO_p$ & Probabilistic big-O notation. \\
$\pto$ & Convergence in probability. \\
$X \succeq Y$ & Stochastic dominance of $X$ by $Y$, indicating that $\mathbb{P}(X\geq z)\geq\mathbb{P}( Y\geq z)$ for all $z \in \mathbb{R}$. \\
\arrayrulecolor{black}\bottomrule
\end{tabularx}
\label{tab:notations}
\end{table}

\section{Proofs in \Cref{sec:interval_methodology}}
\label{app:sec:interval_methodology_proofs}

\subsection{Proof of \Cref{lem:setting_level_parameters_and_coverage}}
\label{app:lem:coverage_proof}

We reproduce the original argument in \cite{berger_boos_supp}, originally stated in terms of p-values, translated into our quantile setting. 
Fix any $\bx^* \in \mathcal{X}$ and consider the sets: 
\begin{itemize}
    \item $A_1 = \{\by : B_\eta(\by) \ni \bx^*\}$
    \item $A_2 = \{\by : \lambda(\mu^*, \by) \leq  Q_{\bx} ( 1 - \gamma)\}$
    \item $A_3 = \{\by : \lambda(\mu^*, \by) \leq \bar{q}_{\gamma, \eta}(\mu^*)\} = \{\by : \lambda(\mu^*, \by) \leq \sup_{\bx \in \mathcal{B}_\eta \cap \Phi_{\mu^*}} Q_{\bx} ( 1 - \gamma)\}$
\end{itemize}
We now have that $\mathbb{P}(\by \in A_1) = 1-\eta, \mathbb{P}(\by \in A_2) = 1-\gamma$, and that $(A_1 \cap A_2) \subset (A_1 \cap A_3)$. 
Therefore,
\begin{align}  
\mathbb{P}(\by \notin A_3) &=  \mathbb{P}(\by \notin A_3, \by \in A_1) +  \mathbb{P}(\by \notin A_3, \by \notin A_1) \\
&\leq \mathbb{P}(\by \notin A_2, \by \in A_1) +  \mathbb{P}(\by \notin A_1) \\ &\leq \mathbb{P}(\by \notin A_2) +  \mathbb{P}(\by \notin A_1)
\\ &= \gamma + \eta
\end{align}
so that $\mathbb{P}(\by \in A_3) \geq 1-\gamma -\eta$. 
Imposing $1-\gamma -\eta \geq 1-\alpha$ gives the desired result.

\subsection{Proof of \Cref{cor:eta-gamma-alpha-adjustment-coverage}}
\label{proof:cor:eta-gamma-alpha-adjustment-coverage}

By definition, $\aosbQ^\mu \leq \aosbQ$ for all $\mu \in \mathbb{R}$.
Therefore, $\localcs \subseteq \globalcs$, and thus
\begin{equation}
    \mathbb{P}_{\bx^*} \big(\mu^* \in \globalcs \big) \geq \mathbb{P}_{\bx^*} \big(\mu^* \in \localcs \big) \geq 1 - \alpha.
\end{equation}

\section{Proof of \Cref{thm:maintheoremalgs}}
\label{proofs_appendix}

Throughout the proof, we make use of the following lemma: 

\begin{lemma}\label{lemma:conv_minimum}
    Let $f: \mathbb{R}^n \to \mathbb{R}$ and $\mathcal{X} \subset \mathbb{R}^n$ such that $\underset{\bx \in \mathcal{X}}{\min} f(\bx)$ is achieved, and at least one of the minimizers $\bx^*$ satisfies: 
    \begin{enumerate}
        \item $f$ is continuous at $\bx^*$,
        \item $\bx^*$ is not an isolated point of $\mathcal{X}$ (i.e., $\forall \delta > 0, B_\delta(\bx^*) \cap \mathcal{X} \neq \emptyset$).
    \end{enumerate}
    Let $\mu$ be a measure on $\mathcal{X}$ such that $\mu(B) > 0$ for all $B \subseteq \mathcal{X}$ such that $\lambda_{\text{Leb}
    }(B) > 0$ (where $\lambda_{\text{Leb}
    }(B)$ refers here to the Lebesgue measure of the set $B$). Let $Y_m = \underset{i = 1, \dots, m}{\min} f(\bx_i)$, where $\bx_i$ are i.i.d. samples from $\mu$. Then $Y_m \xrightarrow{p} f(\bx^*)$.
\end{lemma} 
\begin{proof}
Fix $\varepsilon > 0$ and let us show that $\mathbb{P}(|Y_m - f(\bx^*)| > \varepsilon) \rightarrow 0$. Since $f$ is continuous at $\bx^*$ there exists a $\delta > 0$ such that $f(B_\delta(\bx^*)) \subset B_\varepsilon(f(\bx^*))$ so that, 
\begin{equation}
    \mathbb{P}(|Y_m-f(\bx^*)| \geq \varepsilon) \leq \mathbb{P}(\bx_i \notin B_\delta(\bx^*), \; \forall i = 1, \dots, m) = (\mathbb{P}_{\bx \sim \mu}(\bx \notin B_\delta(\bx^*)))^m. 
\end{equation}
Since $\bx^*$ is not an isolated point, we have $\lambda_{\text{Leb}}(B_\delta(\bx^*) \cap \mathcal{X}) > 0$ and therefore $\mathbb{P}_{\bx \sim \mu}(\bx \notin B_\delta(\bx^*)) < 1$ and  $\mathbb{P}(|Y_m-f(\bx^*)| \geq \varepsilon) \rightarrow 0$.
\end{proof}
We begin by proving that the empirical maximums of the quantiles obtained both by \Cref{alg:max_q_rs,alg:max_q_qr} (assuming the quantile regressor is consistent) converge to the true max quantile. 
\myparagraph{\Cref{alg:max_q_rs}}
\label{app:lem:rs_consistency_proof}

Let $\maxqRS := \underset{i = 1, \dots, M}{\max} \hat{q}^i_\gamma(N)$, where we explicitly write the dependence with the number of samples and the index $i$ refers to the quantile estimated at the $i$-th sampled point $\bx_i$.
We aim to show that $\convprob{\maxqRS}{\aosbQ}$.
We know that $\hat{q}^i_\gamma(N)$ converges in probability to $Q_{P_{\bx_i}}(1-\gamma)$ as $N \rightarrow \infty$. We have 

\begin{align}
    \left\lvert \maxqRS - \aosbQ \right\rvert \leq \left\lvert \maxqRS -\max_{i = 1, \dots, M} Q_{P_{\bx_i}}(1-\gamma) \right\rvert + \left\lvert \max_{i = 1, \dots, M} Q_{P_{\bx_i}}(1-\gamma) - \aosbQ \right \rvert.
\end{align}

The first term can be made smaller than $\varepsilon/2$ as $N \rightarrow \infty$ by convergence of the estimator, and the second term can be made smaller than $\varepsilon/2$ as $M \rightarrow \infty$ by application of \Cref{lemma:conv_minimum} to the quantile function (maximizing instead of minimizing).

\myparagraph{\Cref{alg:max_q_qr}}
\label{app:lem_qr_max_consistency_proof}

The proof is identical to that of \Cref{alg:max_q_rs}, with the only difference of replacing the quantiles estimated via Monte Carlo sampling with those estimated by the quantile regression, and $N$ to $M_{\text{tr}}$, the number of samples needed to train the quantile regression. 
Since the quantile regression is assumed to be consistent, the first term can be made arbitrarily small as $M_{tr}$ grows, and the result follows.

Since identical convergence results apply for both algorithms, we will not explicitly distinguish. 
The proof is written in terms of $N$, which can be replaced by $M_{tr}$.

\subsection{Proof of Statement 1 (Global Inverted)}

Recall,
\begin{align}
    C^{\gl}_{\inv}(\by) &= \bigg[\min_{k \in \{1, \ldots, M\}: \lambda(\varphi(\bx_k), \by) \leq \hat{q}(N)} \, \varphi(\bx_k), \max_{k: \lambda(\varphi(\bx_k), \by) \leq \hat{q}(N)}  \, \varphi(\bx_k) \bigg] \\ 
    &= \bigg[\min_{k \in \{1, \ldots, M\}: \lambda(\mu_k, \by) \leq \hat{q}(N)} \, \mu_k, \max_{k: \lambda(\mu_k, \by) \leq \hat{q}(N)}  \, \mu_k \bigg]
\end{align}
where $\hat{q}(N) :=  \underset{i = 1, \dots, M}{\max} \hat{q}^i_\gamma(N)$, the estimated quantiles of the sampled $\bx_i \in \mathcal{B}_\eta$ and $\mu_i := \varphi(\bx_i)$.
Note that $\mu_i$ are samples in $\varphi(\mathcal{B}_\eta) \subset \mathbb{R}$.
Also note that we have more explicitly written out the interval definition (i.e., \Cref{eq:global_inverted}) to emphasize clarity rather than presentation.

Consider the left extreme of the interval; a similar argument follows from the right extreme. Consider three quantities: 
\begin{align}
    \mu_1(N, M)&:= \min \mu_k \quad \text{s.t.} \quad k = 1, \dots, M \text{ and } \; \lambda(\mu_k, \by) \leq \hat{q}(N) \\ 
    \mu_2(M) &:= \min \mu_k \quad \text{s.t.} \quad k = 1, \dots, M \text{ and } \; \lambda(\mu_k, \by) \leq \aosbQ \\
    \mu_3 &:= \min \mu \quad \text{s.t.} \quad \mu \in \varphi\left(\mathcal{B}_\eta \right) \text{ and } \; \lambda(\mu, \by) \leq \aosbQ
\end{align}
Our goal is to show that as $N, M \rightarrow \infty$, $\convprob{\mu_1}{\mu_3}$.
\Cref{lemma:conv_minimum} shows that $\convprob{\mu_2}{\mu_3}$.
Indeed, the minimization over the indices $k$ such that the condition is satisfied can be seen as a rejection sampling strategy in which all accepted samples are samples of the feasible region of the optimization in $\mu_3$.
As $M$ grows, since the sampler eventually samples all areas of $\varphi(\mathcal B_\eta)$, some samples are guaranteed to be close to the optimum with high probability. Finally, for fixed $M$ and $N$ going to infinity, $\convprob{\mu_1}{\mu_2}$.
This follows from the continuity of the optimization problem with respect to the right-hand side of the constraint, and the fact that $\convprob{\underset{i = 1, \dots, M}{\max} q(\bx_i)}{\aosbQ}$. 
It follows that as $N, M \rightarrow \infty$, $\convprob{\mu_1}{\mu_3}$.

\subsection{Proof of Statement 2 (Sliced Inverted)}

The proof technique is similar to the one of Statement 1, replacing $\hat{q}(N) :=  \underset{i = 1, \dots, M}{\max} \hat{q}^i_\gamma(N)$ for $\hat{q}^k_\gamma(N)$. 
Defined then, similarly to the previous proof:
\begin{align}
    \mu_1(N, M)&:= \min \mu_k \quad \text{s.t.} \quad k = 1, \dots, M \text{ and } \; \lambda(\mu_k, \by) \leq \hat{q}^k_\gamma(N) \\ 
    \mu_2(M) &:= \min \mu_k \quad \text{s.t.} \quad k = 1, \dots, M \text{ and } \; \lambda(\mu_k, \by) \leq \aosbQ(\mu_k) \\
    \mu_3 &:= \min \mu \quad \text{s.t.} \quad \mu \in \varphi\left(\mathcal{B}_\eta \right) \text{ and } \; \lambda(\mu, \by) \leq \aosbQlocal
\end{align}

where $\aosbQlocal = \max_{\bx \in \Phi_\mu \cap \mathcal{B}_\eta} Q_\bx(1 - \gamma)$. As $N \to \infty$, $\hat{q}^k_\gamma(N) \to Q_{\bx_k}(1-\gamma)$. \Cref{lemma:conv_minimum} can be used to show $\convprob{\mu_2}{\mu_3}$, because the sampler strategy used that first samples $\bx_k$ and then accepts $\mu_k = \varphi(\bx_k)$ as a sample if $\lambda(\varphi(\bx_k), \by) < q_\gamma^k$, it can be shown that for every feasible point $\mu$, there is eventually a sample close to it.

Finally, $\mu_1$ will become arbitrarily close to $\mu_2$ as $M$ grows large, since both are taking the minimum over samples that are sampled densely from the feasible set, meaning accepted $\mu_k$ will eventually be close to $\mu_3$  for both the case of $\mu_1$ and the case of $\mu_2$

\subsection{Proof of Statement 3 (Global Optimized)}

We prove the continuity of the optimization problem,
\begin{equation}   
\max_{\mu \in \varphi(\mathcal{B}_\eta)} \mu \text{ s.t. } \lambda(\mu, \by) \leq q,
\end{equation}
as a function of $q$ in the positive measure interval $(\lambda(\bar{\mu}, \by), \aosbQ]$.
Therefore, convergence in probability follows as we have convergence in probability to $\aosbQ$ as $N, M \to \infty$, which in particular implies that the maximum quantile estimate is in the interval $(\lambda(\bar{\mu}, \by), \aosbQ]$ almost surely.
We do so by appealing to the maximum theorem \cite[\S~E.3]{ok2007real}, which, in general, guarantees continuity of functions of the form $f^*(\theta) = \sup \{ f(x, \theta) : x \in C(\theta) \}$ as long as $f$ is continuous, $C$ is a continuous compact-valued correspondence and $C(\theta)$ is non-empty for all $\theta \in \Theta$.
The continuity of $C$ comes from the continuity of the LLR, $f$ is equal to the identity and the strict feasibility condition ensures $C$ is non-empty in $\Theta := (\lambda(\bar{\mu}, \by), \aosbQ]$.

\subsection{Proof of Statement 3 (Sliced Optimized)}

We will prove sufficient conditions for convergence of   $\underset{\mu : \hat{f}_k(\mu) \geq 0}{\inf} \mu$ to $\underset{\mu: f(\mu) \geq 0}{\inf} \mu$ as $\hat{f}_k$ converges to $f$, and the result will follow by taking $\hat{f}_k(\mu) = \hat{m}_\gamma(\mu) - \lambda(\mu, \by)$ and $f(\mu) = m_\gamma(\mu) - \lambda(\mu, \by)$.
A similar argument can be repeated for the supremum.
Use the notation $\hat{f}_k$ to indicate that $k$ sampled points are used to estimate this function via the definition of $\hat{m}_\gamma(\mu)$ (see \Cref{sec:interval_constructions}).

\begin{lemma}\label{lemma:conv_minimum-v2}
    Let $f: \mathbb{R} \to \mathbb{R}$ be a function, and let $f_k$ be a sequence of functions $f_k: \mathbb{R} \to \mathbb{R}$. Let $\displaystyle \mu^* = \inf_{f(\mu) \geq 0} \mu$ and $\displaystyle \mu_k = \inf_{f_k(\mu) \geq 0} \mu$.
    Let the sequence of functions $\{ f_k \}$ be such that for all $\delta > 0$,
    \begin{equation}
        \mathbb{P}\left(\underset{\mu}{\sup} |f_k(\mu) - f(\mu)| > \delta\right) \rightarrow 0 \text{ as } k \rightarrow \infty,
    \end{equation}
    namely, $f_k$ converges in probability uniformly to $f$.
    % Furthermore, let $f$ be such that for all $\varepsilon > 0$, there exists $\delta' > 0$ such that $\lvert f(\mu) \rvert > \delta'$ if and only if $\lvert \mu - \mu^* \rvert < \varepsilon$.
    Furthermore, let $f$ be such that for all $\varepsilon > 0$, there exists $\delta' > 0$ such that if $\lvert \mu - \mu^* \rvert \geq \varepsilon$, then $\lvert f(\mu) \rvert > \delta'$.
    Then, we have $\convprob{\mu_k}{\mu^*}$.

\end{lemma}
% We then have $\mu_k \xrightarrow{p} \mu^*$. 

\begin{proof}
    Assume for the sake of contradiction that there exists $\varepsilon > 0$ such that $\mathbb{P}\left(\lvert \mu_k - \mu^* \rvert \geq \varepsilon\right)$ does not go to $0$.
    Then, by the condition on $f$, it must follow that $\mathbb{P}\left(\lvert f(\mu_k) \rvert > \delta \right)$ does not go to $0$.
    But, 
    \begin{align}
        \mathbb{P}\left(\lvert f(\mu_k)\rvert > \delta\right) \leq \mathbb{P} \left( \lvert f(\mu_k) - f_k(\mu_k)\rvert > \delta \right) + \mathbb{P} \left( \lvert f_k(\mu_k) \rvert > \delta \right),
    \end{align}
    and the right-hand side goes to 0 by uniform convergence in probability (first term) and by feasibility of $\mu_k$ (second term).
\end{proof}

\section{Sampling the Berger-Boos set}
\label{app:sec:sampling_bergerboos}

\subsection{VGS Sampler for low-dimensional and full column rank settings}
\label{subsec:vgs_ellipsoid}

Under the linear-Gaussian assumptions, when $\bK$ has full column rank, it can be shown that
\begin{equation} \label{eq:sim_conf_set_around_ls}
    \mathcal{B}_\eta = \big\{ \bx \in \mathcal{X}: (\bx - \hat{\bx})^\top \bK^\top \bSigma^{-1} \bK (\bx - \hat{\bx}) \leq \chi^2_{n, \eta} \big\},
\end{equation}
where $\hat{\bx}$ is the Generalized Least-Squares (GLS) estimator.
Equation~\eqref{eq:sim_conf_set_around_ls} describes an ellipsoid in $\mathbb{R}^p$ intersected with $\mathcal{X}$ with axis directions and lengths determined by $\chi^2_{n, \eta}$ and the eigenvectors and eigenvalues of $\bK^\top \bSigma^{-1} \bK$, respectively (see \cite{rust_burrus_supp} for more details).
An ellipsoid is nothing but a deformed ball.
As such, we can sample uniformly at random from the ellipsoid in $\mathbb{R}^p$ using an algorithm sampling uniformly at random from a $p$-ball, followed by an appropriate linear transformation and translation.
We can then include an additional accept-reject step to account for the constraint set $\mathcal{X}$.
Note that all subsequent discussions of spheres and balls assume unit radii and centering at the origin.
Also note that because the ellipsoid in \eqref{eq:sim_conf_set_around_ls} is centered around the GLS estimator, we can sample points first from within an ellipsoid of the same shape centered at the origin, and then translate those points by $\hat{\bx}$.

\cite{vgs_supp} propose a particularly efficient and clever algorithm to sample uniformly at random from the $p$-ball by proving a connection between uniform sampling on the $(p + 1)$-sphere (i.e., the surface of the ball in $\mathbb{R}^{p + 2}$) and uniform sampling within the $p$-ball (i.e., \emph{within} the unit ball in $\mathbb{R}^p$).
Namely, one can sample a Gaussian $\tilde{\bz} \sim \cN( \bm{0}, \bI_{p + 2})$ and normalize $\bz := \tilde{\bz} / \Vert \tilde{\bz} \Vert_2$ to sample a point uniformly at random from the $(p + 1)$-sphere.
Next, one simply drops the last two elements of $\bz$ to obtain a point that is uniformly sampled from within a $p$-ball.
The validity of this approach is substantiated by Lemma 1 and Theorem 1 in \cite{vgs_supp} and relies upon a distribution result about the ratio of chi-squared distributions and preservation of distribution under an orthogonal transformation.
To denote sampling a point following this procedure, we use the notation $\bx \sim \VGS(p)$.

To sample from our desired ellipsoid in \eqref{eq:sim_conf_set_around_ls}, first consider the eigendecomposition of $\bK^\top \bSigma^{-1} \bK = \bP \bOmega \bP^\top$, where $\bOmega = \text{diag} \left(\omega_1^2, \omega_2^2, \dots, \omega_p^2 \right)$, $\omega_i$ is the $i$-th eigenvalue of $\bK^\top \bSigma^{-1} \bK$ and $\bP$ is an orthonormal matrix where the columns vectors are the eigenvectors of $\bK^\top \bSigma^{-1} \bK$.
Denote $\bOmega^{1/2} = \text{diag} \left(\omega_1, \omega_2, \dots, \omega_p \right)$ and by extension, $\bOmega^{-1/2} = \text{diag} \left(\omega_1^{-1}, \omega_2^{-1}, \dots, \omega_p^{-1} \right)$ when $\omega_i > 0$ for all $i$.
These decompositions imply the correct transformation to apply to points sampled from the $p$-ball via $\bx \sim \VGS(p)$.
Namely, define $\bw := \sqrt{\chi^2_{n, \eta}} \bP \bOmega^{-1/2} \bx$.
To know that $\bw$ is sampled from the correct ellipsoid, it should be the case that $\bw^\top \bK^\top \bSigma^{-1} \bK \bw \leq \chi^2_{n, \eta}$, as then we can simply make the update $\bw \leftarrow \bw + \hat{\bx}$ to ensure that we have a point sampled from \eqref{eq:sim_conf_set_around_ls}.
This guarantee is verified as follows:
\begin{align}
    \bw^\top \bK^\top \bSigma^{-1} \bK \bw &= \bw^\top \bP \bOmega^{1/2} \bOmega^{1/2} \bP^\top \bw \nonumber \\
    &= \chi^2_{n, \eta} \cdot \bx^\top \bOmega^{-1/2} \bP^\top \bP \bOmega^{1/2} \bOmega^{1/2} \bP^\top \bP \bOmega^{-1/2} \bx  \label{eq:change_of_basis_justification} \\
    &= \chi^2_{n, \eta} \cdot \bx^\top \bx \leq \chi^2_{n, \eta}, \nonumber
\end{align}
where $\bP^\top \bP = \bI$ by definition and the last line follows since we know $\bx$ is sampled from the $p$-ball and therefore $\bx^\top \bx \leq 1$.
\Cref{eq:change_of_basis_justification} justifies that the $\bw$ samples lie within the desired ellipsoid, and their uniform distribution follows from the uniform distribution of $\bx$ in the $p$-ball since the distribution is preserved under a linear transformation.
Finally, we accept $\bw$ if $\bw \in \mathcal{X}$ and reject if $\bw \notin \mathcal{X}$.
This procedure for sampling at random from $\mathcal{B}_\eta$ is summarized in \Cref{alg:ellipsoid_sample_full_rank}.

\begin{algorithm}[!ht]
\caption{VGS Sampler}
\label{alg:ellipsoid_sample_full_rank}
\begin{algorithmic}[1]
\REQUIRE $p \in \mathbb{N}$, $M \in \mathbb{N}$, $\bP$, $\bOmega^{-1/2}$, $\hat{\bx}$, $\chi^2_{n, \eta}$.
\vspace{0.35em}
\STATE Define $\mathcal{S} := \{ \}$ to be the initialized set in which the sampled points are to be placed.
\FOR{$k = 1, 2, \dots, M$}
    \STATE \textbf{Sample from the $p$-ball using VGS:} $\bx_k := \bz_{1:p} / \Vert \bz \Vert_2$, where $\bz \sim \cN(\bm{0}, \bI_{p + 2})$ and $\bz_{1:p}$ denotes taking the first through $p$-th indices (inclusive).
    \STATE \textbf{Transform VGS output:} $\bw_k := \sqrt{\chi^2_{n, \eta}} \bP \bOmega^{-1/2} \bx_k$.
    \STATE \textbf{Translate $\bw_k$ by the GLS estimator:} $\bw_k \leftarrow \bw_k + \hat{\bx}$.
    \STATE \textbf{Accept-Reject to incorporate constraints}: If $\bw_k \in \mathcal{X}$, then add $\mathcal{S} \leftarrow \mathcal{S} \cup \{\bw_k\}$, else start loop iteration $k$ again.
\ENDFOR
\vspace{0.15em}
\ENSURE $\mathcal{S}$ containing uniformly sampled points over $\mathcal{B}_\eta$ (as defined in \eqref{eq:sim_conf_set_around_ls}).
\end{algorithmic}
\end{algorithm}
Generating samples using the VGS Sampler is efficient, but its feasibility diminishes in high dimensions because of the accept-reject step.
To illustrate this point, consider the simple scenario where $\mathcal{X} = \mathbb{R}^p_+$, i.e., the non-negative orthant of $\mathbb{R}^p$ and suppose we sample from the $p$-ball intersected with $\mathbb{R}^p_+$ by sampling $\bx \sim \VGS(p)$ where $\bP = \bOmega = \bI$, $\hat{\bx} = \bm{0}$ and we use $1$ instead of $\chi^2_{n, \eta}$.
Then, $\mathbb{P}(\bx \in \mathbb{R}^p_+) = 2^{-p}$ and the acceptance probability of each sample goes to zero exponentially in $p$.
To make this point slightly more general, we generate data $\by \sim \mathcal{N}(\bx^*_p, \bI_p)$ where $\bx^*_p \in \mathbb{R}^p_+$ and is a vector of ones.
For a collection of dimensions $p \in [2, 30]$, we sample $\bx \sim \VGS(p)$ with $\bP = \bOmega = \bI$ and $\hat{\bx} = \by$ to estimate the probability that $\bx$ is in $\mathbb{R}^p_+$ and plot the results in the left panel of \Cref{fig:vgs_sampler_failure}.
Since the acceptance probability decays exponentially (under $10^{-4}$ for $30$ dimensions), it becomes clear that this algorithm is inefficient in dimensions larger than $10$, since the acceptance probability quickly becomes prohibitively small in regimes where a larger sample size is even more important.

Dimension is only one of two primary complicating factors.
Not only is less of a (potentially) shifted $p$-ball's volume in the non-negative orthant as $p$ grows but if our data are generated via $\by \sim \mathcal{N}(\bK \bx^*, \bI_p)$ with $\bx^* \in \mathbb{R}^p_+$ where the condition number of $\bK$ is large, it is possible that even less of the ellipsoid from which the VGS Sampler draws points intersects with the parameter constraint.
To illustrate this point, we generate a single observation from the aforementioned model using a $\bK \in \mathbb{R}^{40 \times 40}$ as described in \Cref{sec:wide_bin_deconvolution}.
The $\bx^* \in \mathbb{R}^{40}_+$ is created in the same way as that of \Cref{sec:unfolding_realistic}.
The right panel of \Cref{fig:vgs_sampler_failure} shows a computed probability mass function for the number of coordinates in a VGS Sampler draw (with $\bP$ and $\bOmega$ defined such that $\bK^\top \bK = \bP \bOmega \bP^\top$) complying with the non-negativity parameter constraints.
Equivalently, the right panel of \Cref{fig:vgs_sampler_failure} shows the computed probability mass function for the number of coordinates lying within the \bergerboos set.
For this particular setup, we critically note that out of $5 \times 10^4$ draws from the VGS Sampler, none of the draws had all coordinates comply with the non-negativity constraint.
By contrast, for the aforementioned noise model where we more simply sample from the $p$-ball intersected with the non-negative orthant, the computed acceptance probability is approximately $3.35 \times 10^{-6}$, in both cases emphasizing the VGS Sampler's poor performance in high-dimensional and ill-conditioned forward model regimes.

\begin{figure*}[!ht]
  \includegraphics[width=0.49\textwidth]{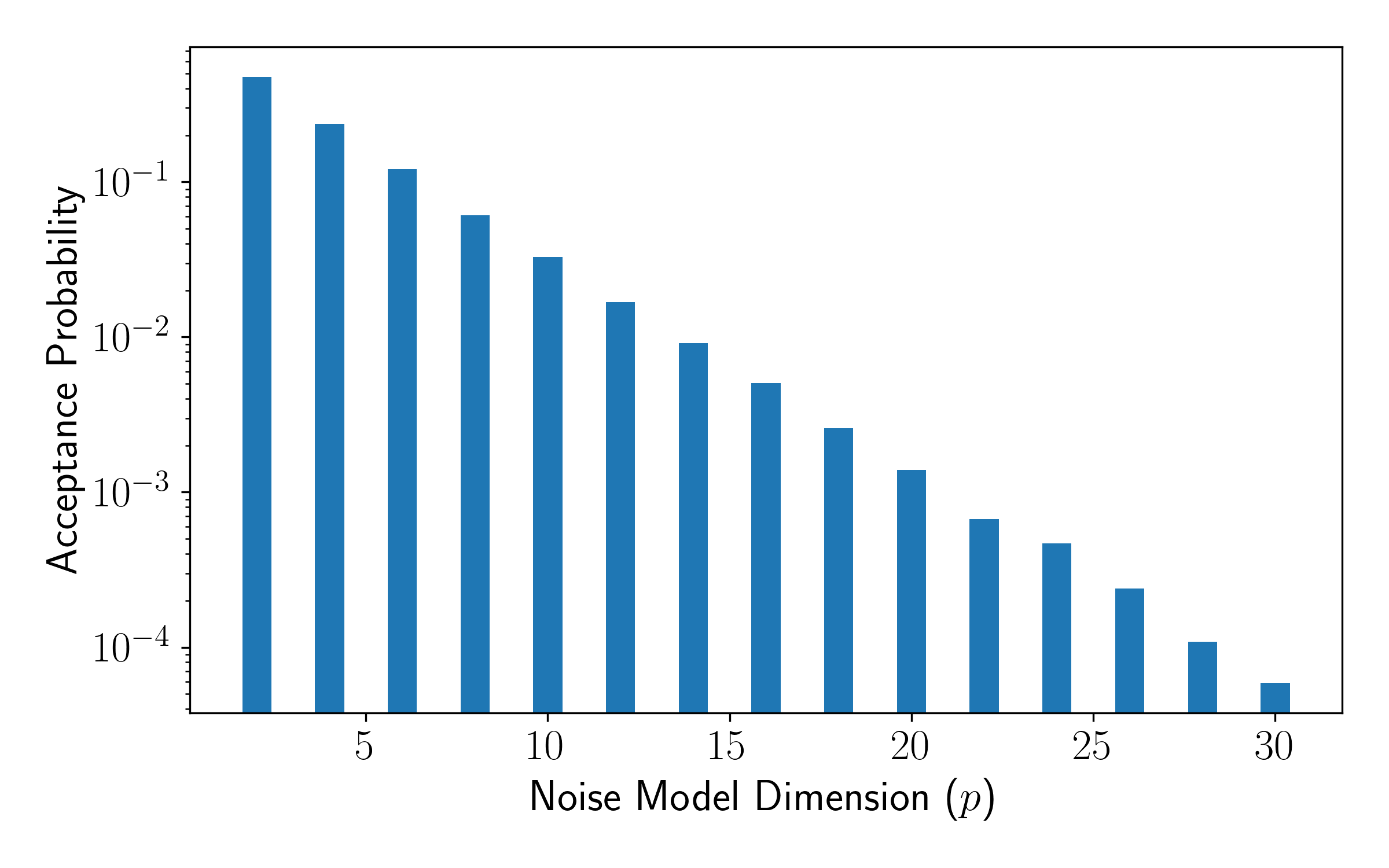}
  \includegraphics[width=0.49\textwidth]{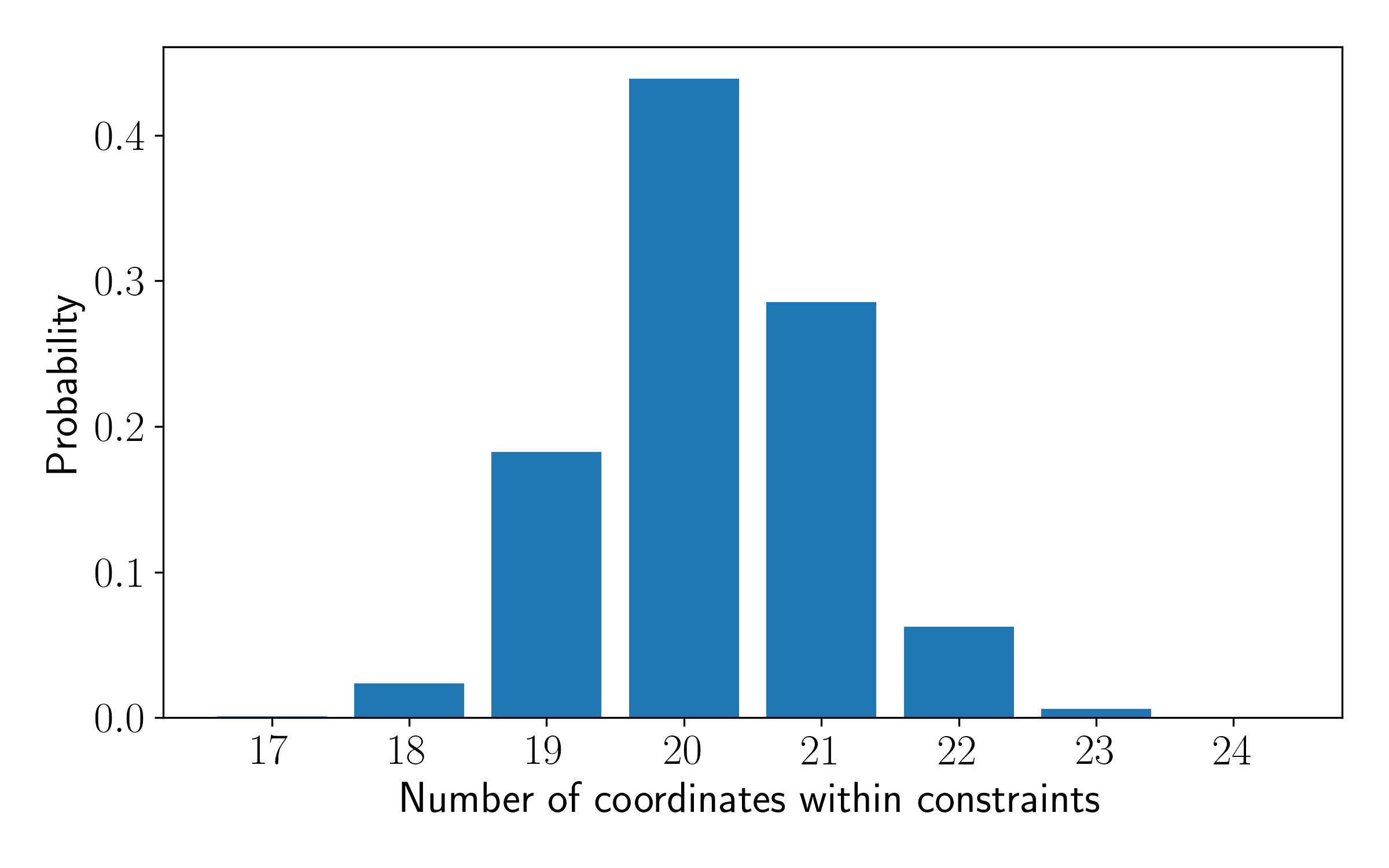}
  \caption[VGS Sampler infeasibility in high-dimensional regimes]{Numerical illustrations of the VGS Sampler's infeasibility in high-dimensional regimes. The \textbf{left} panel shows the computed acceptance probability of a point drawn by the VGS Sampler with data generated from a non-negatively constrained Gaussian noise model. Crucially, at only $30$ dimensions, the acceptance probability is already less than $10^{-4}$ for this particular setup. The \textbf{right} panel shows the computed probability mass function for the number of non-negative constraint complying coordinates of a VGS sample with data generated from a non-negatively constrained linear Gaussian model in $40$ dimensions with a non-identity forward model. Since this is an example using a forward model with a large condition number ($\approx 1.6 \times 10^4$), we critically note that there is empirically zero probability of generating a sample within the non-negativity constraints.}
  \label{fig:vgs_sampler_failure}
\end{figure*}

\subsection{Polytope sampler for general settings}
\label{subsec:mcmc_algo}

In settings where the forward model is not linear and full column rank, \Cref{alg:ellipsoid_sample_full_rank} fails.
The ineffectiveness of this algorithm increases if the condition number of the linear forward model is large, such that most of the pre-image ellipsoid defining the \bergerboos set lies outside of the constraint set.
Although these scenarios induce particular geometric challenges, in the linear-Gaussian case with convex $\mathcal{X}$, we are still fundamentally sampling a convex set for which there is a vast literature.
For example, there is a vast sampling literature for computing Bayesian posteriors in high dimensions via nested sampling \cite{skilling_supp, nest_samp_phys_sci, buchner_supp}.
In particular, nested sampling has been successfully applied in high-dimensional cosmology settings using sophisticated approaches to strategically sample the parameter space by restricting prior sampling in various ways \cite{buchner_supp, montel_nested_samp}.
Although these approaches provide tools addressing a sampling setting similar to ours (i.e., the \bergerboos set can be viewed as a portion of the parameter space defined by a cutoff on the likelihood) they are ultimately aimed at sampling from a particular distribution (i.e., the posterior), which is a stronger criterion than required here.
For sampling general convex sets, simple algorithms like Hit-and-Run are available \cite{hit_and_run, lovasz_hnr_mix_fast, lovasz_hnr_from_corner}.
However, in the particular case considered here, more sophisticated and efficient algorithms can be devised.
In particular, there exists a deep literature on random walks over polytopes such that the asymptotic stationary distribution of the walk is a uniform distribution over the polytope of interest \cite{kannan, narayanan, mcmc_polytope_supp}.
As such, we propose to first construct a bounding polytope, $\mathcal{P}^d$ composed of $d$ half-spaces around $\mathcal{B}_\eta$, sample $C$ random walks within the \bergerboos set using the Vaidya walk as described in \cite{mcmc_polytope_supp} each starting at a point from a collection of strategically chosen locations, and then combine the parallel chains to create the final sample set.
The detailed algorithm can be seen in \Cref{alg:polytope_sampler}.

Although MCMC algorithms are typically evaluated using trace plots on individual dimensions of the parameter space, given the high-dimensionality of the problem and the final step combining several MCMC chains in the parameter space started from different positions, it is more meaningful to evaluate this algorithm's ability to sample fully from the functional space.
Namely, we can solve for the largest and smallest values the functional can take within the \bergerboos set as follows:
\begin{equation} \label{eq:ssb_bounds}
    I_{\BB}(\by) := \bigg[\mu^{l}_{\BB}, \mu^{u}_{\BB} \bigg] = \bigg[\min_{\bx \in \mathcal{B}_\eta} \varphi(\bx), \max_{\bx \in \mathcal{B}_\eta} \varphi(\bx) \bigg].
\end{equation}
When $\varphi(\bx) = \bh^\top \bx$ for some $\bh \in \mathbb{R}^p$, $I_{\BB}(\by)$ corresponds to the SSB interval in \cite{stanley_unfolding_supp}.
Computing the sets $\localcs$ and $\globalcs$ well is then contingent upon sampling functional values within $I_{\BB}(\by)$ well, since a functional value can only be included in the inverted set if the sampler has a non-zero probability of sampling arbitrarily close to it.
By ``well'', we informally mean that the sampled functional values range at least between the endpoints of $I_{\BB}(\by)$ and that if we partition $I_{\BB}(\by)$ into $n$ sub-intervals, that most if not all of the sub-intervals contain at least one sample.
In practice, it will often be the case that the sampled functional values can lie outside the endpoints of $I_{\BB}(\by)$ since the MCMC chains are sampling a bounding polytope of the \bergerboos set.

\myparagraph{Choosing the bounding polytope.}
Any bounded polytope of $\mathcal{B}_\eta$ is the intersection of a finite number of half-spaces defined in $\mathbb{R}^p$.
The task of constructing a bounding polytope is then equivalent to choosing a collection of $d$ hyperplanes in $\mathbb{R}^p$ to construct a set $\mathcal{P}^d := \{\bx \in \mathbb{R}^p : \bA \bx \leq \bb \}$ such that $\mathcal{B}_\eta \subseteq \mathcal{P}^d$, where $\bA \in \mathbb{R}^{d \times p}$ and $\bb \in \mathbb{R}^d$.
Let $\ba_i^\top \in \mathbb{R}^p$ denote the $i$-th row vector of $\bA$.
We compute $b_i$ (the $i$-th element of $\bb$) as
\begin{equation}
    \label{eq:hyperplane_max}
    b_i = \max_{\bx \in \mathcal{B}_\eta} \ba_i^\top \bx.
\end{equation}
This construction ensures the necessary inclusion.
We consider three approaches to pick the vectors $\ba_i$: (i) using the constraints $\mathcal{X}$, (ii) using the known eigenvectors defining the bounded directions of the pre-image ellipsoid, and (iii) randomly.
In practice, we combine these approaches to ensure that we consider only parameter settings in agreement with our physical constraints, and to tighten the bounding \bergerboos set polytope as much as possible.
There is a tradeoff with respect to the latter consideration since the mixing time and computational cost for the Vaidya walk increase with the number of hyperplanes \cite{mcmc_polytope_supp}.

To incorporate the known parameter constraints, consider the non-negativity constraint used in \Cref{sec:numerical_exp}, i.e., $\bx \in \mathbb{R}^p_+$.
To enforce non-negativity, we set $\ba_i = -\be_i$ for $i = 1, \dots, p$, where $\be_i$ is defined by its $i$-th element set to one and the rest of its elements set to zero, and $b_i = 0$.
These choices produce $p$ rows in $\bA$ corresponding to the desired lower bounds (i.e., $x_i \geq 0$ for all $i$), but we can compute an additional $p$ constraints using \eqref{eq:hyperplane_max} with $\ba_i := \be_i$ for $i = p + 1, \dots 2p$.
These $2p$ constraints define a hyperrectangle enclosing the \bergerboos set in the parameter space.
To incorporate polytope constraints based upon the forward model, we use the ellipsoidal definition of the pre-image as shown in \eqref{eq:sim_conf_set_around_ls} and eigendecomposition of $\bK^\top \bSigma^{-1} \bK = \bP \bOmega \bP^\top$ shown in \eqref{eq:change_of_basis_justification}.
Note that this ellipsoid form is valid under the linear-Gaussian noise assumption and for all $\bK$.
In the event that $\bK$ is not full column rank, the ellipsoid defined by \Cref{eq:sim_conf_set_around_ls} is still defined via the pseudo-inverse of $\bK^\top \bSigma \bK$, but where the ellipsoid is unbounded in some directions.
The column vectors of $\bP$ corresponding to the non-zero entries on the diagonal of $\bOmega$ are the eigenvectors corresponding to the bounded principal axes.
As such, both $\bp_i$ and $-\bp_i$ for $i = 1, \dots, p$ (the column vectors of $\bP$) can be used as rows of $\bA$ with their corresponding bounds defined by \eqref{eq:hyperplane_max}.
Finally, to further tighten the polytope around the \bergerboos set, we sample a multivariate Gaussian, i.e., $\ba_i \sim \mathcal{N}(\bm{0}, \bI_p)$ to include random hyperplanes.
Note, it is possible that the unbounded directions of the ellipsoid defined by \Cref{eq:sim_conf_set_around_ls} are not bounded when intersecting with the parameter constraint set, and hence there is no bounding polytope.
In this case, our interval construction should be unbounded since there is not enough information to produce a bounded confidence set for the quantity of interest.

Once the components $(\bA, \bb)$ have been defined using some or all of the above hyperplane generation strategies, the Vaidya walk can immediately be employed to perform a random walk around the polytope.
The primary intuitive requirement we wish to satisfy with any sampling scheme in this context is that every region of the \bergerboos set has a non-zero probability of being sampled.
Asymptotically, the Vaidya walk samples the desired polytope uniformly at random, which satisfies a stronger requirement.
In practice, the need to sample non-uniformly can arise if there are particularly meaningful parameter settings in regions that are difficult for the random walk to reach.
We explore such a case in \Cref{sec:3d_constrained_gaussian}.
Additionally, although the asymptotic distribution of the Vaidya sampler is theoretically sufficient, in practice, it often has some difficulty reaching the corners of the generated polytope.
Although MCMC chain mixing is typically evaluated by looking at trace plots, this diagnostic is insufficient here because the dimension is high, and the defined polytope can make different dimensions difficult to compare.
Instead, we consider how well the sampler samples the functional values $\varphi(\bx)$.
This motivates taking a collection of starting points between the SSB endpoints and the Chebyshev center of the polytope, as described in the following section.

\myparagraph{Constructing the parallel chain starting points.} 
Although the random walks defined in \cite{mcmc_polytope_supp} asymptotically sample uniformly over $\mathcal{P}^d$, given the long and thin shape of the pre-image, running the Vaidya walk from even a ``good'' starting position does not consistently sample the functional space well.
Instead, we use the following heuristic to construct several parallel chains that constitute a complete sample when combined.
We define an even number of starting points, $C \in 2\mathbb{N}$, to roughly span the \bergerboos set, run the Vaidya walk for $M_{p}$ steps from each starting point to collect parameter settings $\{\tilde{\bx}_1, \dots, \tilde{\bx}_{M_p} \}$, and combine the samples from each walk, resulting in $C \times M_p$ total samples.
In practice, $C$ and $M_p$ are chosen to yield a total of $M$ samples as desired for \Cref{alg:max_q_rs} or \Cref{alg:max_q_qr}.
Since we have designed this process to avoid the neccessity of any individual chain reaching its asymptotic distribution, we do not need a burn-in period for any chain as long as we are confident that the union of the sampled points sufficiently covers Interval~\eqref{eq:ssb_bounds} so that we may accept or reject any quantity of interest value in that interval.
We define the collection of starting points using the line segments defined by the parameters generating the endpoints of $I_{\BB}(\by)$ and the Chebyshev center of $\mathcal{P}^d$ as defined in \cite{convex_opt}, and denoted by $\bx_c$.
This point is the center of the largest ball contained within $\mathcal{P}^d$ and therefore acts as a way to characterize the center of $\mathcal{P}^d$.
We denote the parameter settings generating the endpoints of $I_{\BB}(\by)$ by $\hat{\bx}^{l}$ for the lower endpoint and $\hat{\bx}^{u}$ for the upper endpoint.
We then define a uniform grid of values $\{ \tau_l \}_{l = 1}^{C / 2}$ such that $\tau_l \in (0, 1)$ and $\tau_l < \tau_{l + 1}$ for all $l$.
For each $k \in [C]$, define $k' := \lceil k / 2 \rceil$ and set $\bx^{start}_k := \tau_{k'} \hat{\bx}^{l} + (1 - \tau_{k'}) \bx_c$ if $k$ is odd and $\bx^{start}_k := \tau_{k'} \hat{\bx}^{u} + (1 - \tau_{k'}) \bx_c$ if $k$ is even.
Creating starting positions along the lines connecting these endpoints and the Chebyshev center accommodates the chosen polytope while helping ensure that samples are chosen spanning the range of possible functional values over the \bergerboos set.

The empirical performance of the Polytope sampler can be seen in \Cref{fig:sampler_figure}, showing (left) a histogram of the functional values sampled and (right) a trace plot for the sampled Vaidya walks starting from points constructed as described above.
These plots are generated for one observation from the $80$-dimensional ill-posed inverse problem in the coverage study performed in \Cref{sec:wide_bin_deconvolution}.
Critically, the histogram shows that the Polytope sampler samples the functional space well, and the trace plot shows that our starting point construction heuristic performs well in practice.

\begin{figure*}[!ht]
  \includegraphics[width=0.49\textwidth]{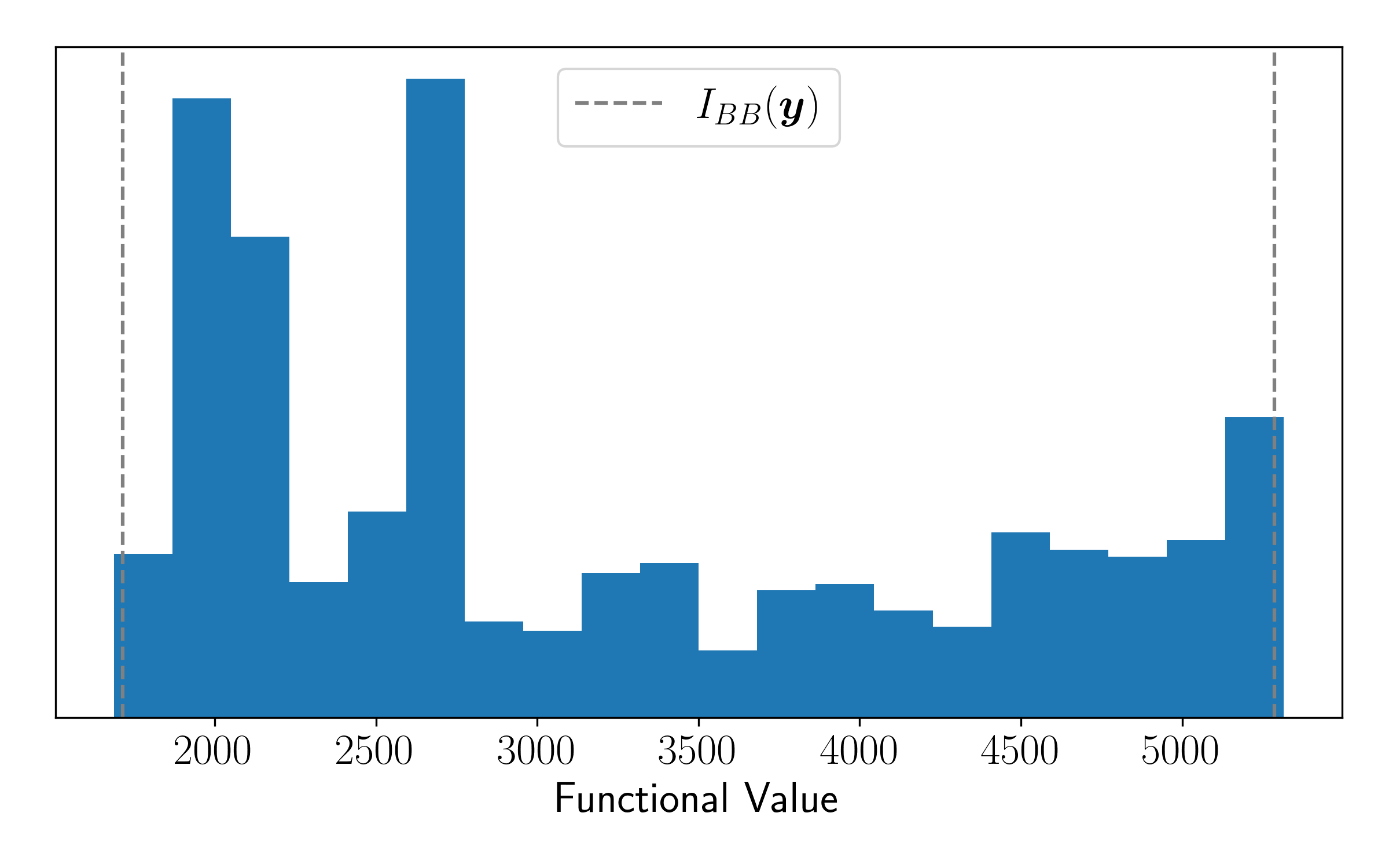}
  \includegraphics[width=0.49\textwidth]{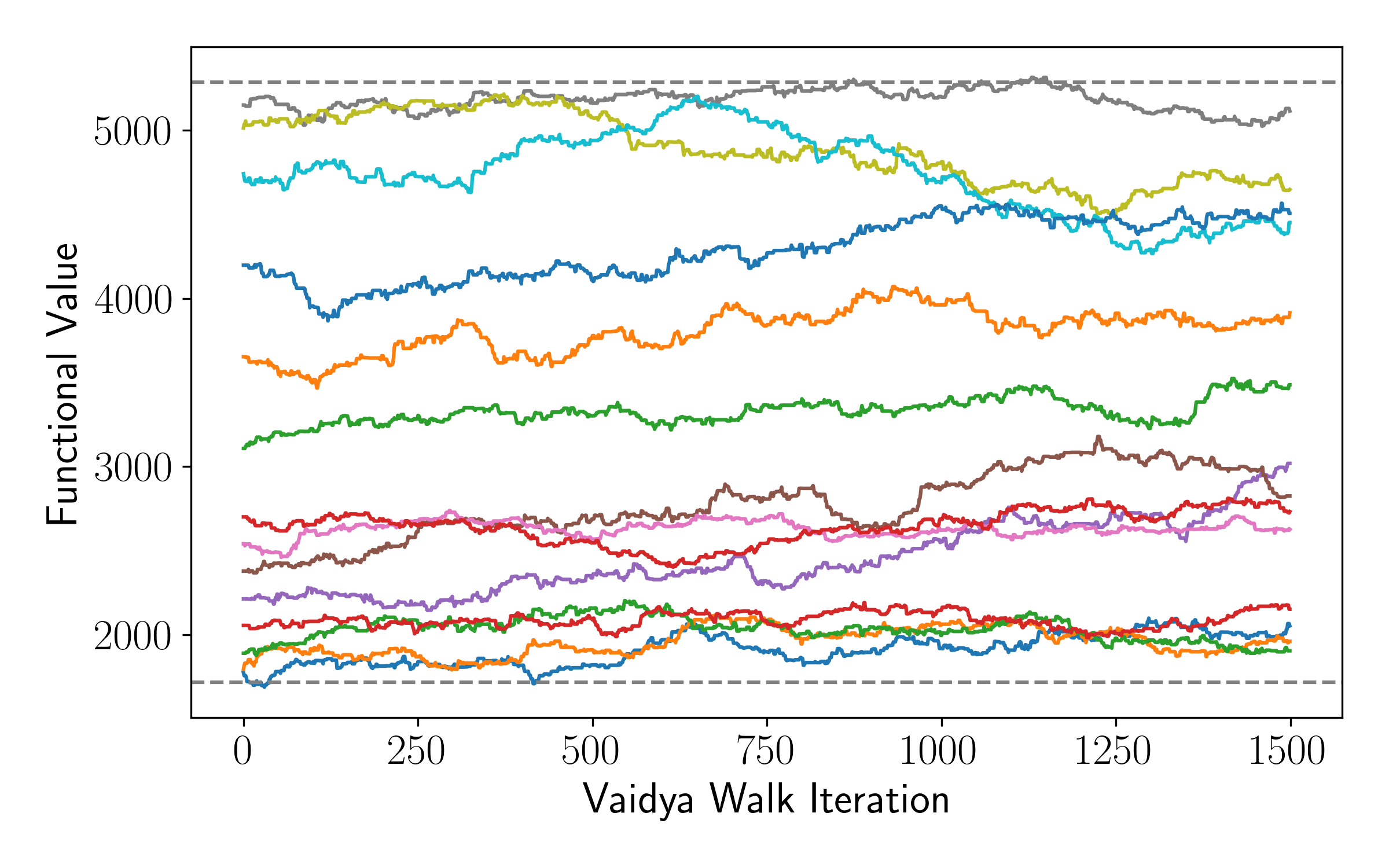}
  \caption[Polytope sampler output for the wide-bin deconvolution problem]{Polytope sampler output for a realization of the $80$-dimensional ill-posed inverse problem studied in \Cref{sec:wide_bin_deconvolution}. The \textbf{left} panel contains a histogram of sampled functional values which both span and cover well the range of $I_{\BB}(\by)$ (shown by the dashed gray lines in both plots). The \textbf{right} panel contains trace plots of the $14$ Vaidya walks (each indicated by a different color) which together constitute the full sample. Our heuristic for choosing starting points along the lines connecting the parameter settings generating the endpoints of $I_{\BB}(\by)$ and the Chebyshev center of the polytope provides a good initial spread of starting functional values.}
  \label{fig:sampler_figure}
\end{figure*}

\myparagraph{Sampling from the \bergerboos set.}
With $\mathcal{P}^d$ and the starting positions defined, we construct a sample $\mathcal{S} := \{\tilde{\bx}_1, \dots,  \tilde{\bx}_{C \times M_p}\}$.
We leave the details of the Vaidya walk to the original paper \cite{mcmc_polytope_supp}, but note an essential radius tuning parameter of the algorithm that must be chosen.
This radius impacts the spread of a Gaussian proposal distribution for the walk and thus affects a new proposed point's acceptance probability.
Choosing a radius that is too large results in a low acceptance rate of proposed steps, thus creating a walk that does not mix well.
In contrast, choosing a radius that is too small results in a high acceptance rate with relatively small step sizes.
In practice, we find a radius setting of $0.5$ works well as it produces an acceptance probability of $\approx33.3\%$, producing a reasonable tradeoff between taking meaningful steps and not rejecting too many steps.

\begin{algorithm}[!ht]
\caption{Polytope sampler}
\label{alg:polytope_sampler}
\begin{algorithmic}[1]
\REQUIRE $M_p \in \mathbb{N}$ and $C \in 2\mathbb{N}$. $\bA \in \mathbb{R}^{d \times p}$, $\bb \in \mathbb{R}^{d}$. $\{\tau_l \}_{l = 1}^{C / 2}$, where $\tau_l \in (0, 1)$ and $\tau_l < \tau_{l + 1}$ for all $l$.
\vspace{0.35em}
\STATE Let $\mathcal{S} := \{ \}$ be the set in which we store all sampled parameter settings.
\STATE \textbf{Construct Chebyshev center:} Solve for $\bx_c$ using the following optimization.
\begin{align}
    \underset{\bx_c, r}{\text{maximize}} \quad& r \nonumber \\
    \text{subject to} \quad& \ba_i^\top \bx_c + r \lVert \ba_i \rVert_2 \leq b_i, \; \; i = 1, \dots, d. \nonumber
\end{align}
\STATE \textbf{Compute $\varphi$ extremes of \bergerboos set:} Extreme points are computed with respect to the functional of interest:
\begin{align}
    \hat{\bx}^{l} &:= \text{argmin} \; \; \varphi(\bx) \quad \text{subject to} \; \; \bx \in \mathcal{B_\eta}, \nonumber \\
    \hat{\bx}^{u} &:= \text{argmax} \; \; \varphi(\bx) \quad \text{subject to} \; \; \bx \in \mathcal{B_\eta}.
\end{align}
\FOR{$k = 1, 2, \dots, C$}
    \STATE \textbf{Construct starting point:} Define $k' := \lceil k / 2 \rceil$. Define $\bx^{start}_k = \tau_{k'} \hat{\bx}^{l} + (1 - \tau_{k'}) \bx_c$ if $k$ is odd and $\bx^{start}_k = \tau_{k'} \hat{\bx}^{u} + (1 - \tau_{k'}) \bx_c$ if $k$ is even.
    \STATE \textbf{Run the Vaidya walk for $M_p$ steps:} Collect samples $\{\tilde{\bx}_1, \dots, \tilde{\bx}_{M_p} \}$ and add them to $\mathcal{S}$.
\ENDFOR
\vspace{0.15em}
\ENSURE $\mathcal{S}$ containing sampled points over $\mathcal{B}_\eta$.
\end{algorithmic}
\end{algorithm}

\section{Quantile regression overview}
\label{app:sec:quantile_regression_overview}

% \Cref{alg:max_q_qr} explained in \Cref{sec:interval_constructions} involves using quantile regression to learn a quantile surface from a collection of pairs of design points and samples from the LLR test statistic.
% As previously mentioned, similar approaches have been taken in \cite{dalmasso2020confidence, lf2i, waldo_supp, masserano_berger_boos}, and since quantile regression is a technique facilitating our interval constructions, we will only give a brief overview of quantile regression and some different ways to implement it.
Fundamentally, given a one-dimension random variable $z \sim P_{\bx}$ that depends on the parameter $\bx \in \mathbb{R}^p$, we are interested in the upper $\gamma$-quantile at every parameter setting, i.e.,
\begin{equation}
    \mathbb{P}_{\bx}(z > Q_{\bx}(1 - \gamma)) = \gamma.
\end{equation}
We note that the quantile surface itself is not random, so we can use draws from the distribution $P_{\bx}$ to estimate $Q_{\bx}$ at a given parameter setting $\bx$.
However, as noted in \Cref{sec:interval_constructions}, performing such an estimate is not always computationally feasible, and intuitively, we might expect quantiles to vary smoothly over the parameter space, which would imply that information about a quantile at $\bx_1$ should be related to a quantile at $\bx_2$ if these points are close.
In statistics literature, estimating $Q_{\bx}$ is framed as estimating a quantile function conditional on known covariates and is often thought of as a generalization of estimating the conditional median \cite{qr_econ_perspectives}.
Just as conditional mean and conditional median estimation can be accomplished by using an appropriate loss function (sum of squares and absolute differences, respectively), estimating conditional quantiles can be accomplished by minimizing the pinball loss defined as follows:
\begin{equation} \label{eq:pinball_loss}
    L_\gamma(z, q) := \begin{cases}
        (1 - \gamma) (q - z), & z < q, \\
        -\gamma (z - q), & z \geq q.
    \end{cases}
\end{equation}
\cite{koenker2005quantile, pinball_loss}.
As such, estimating the quantile surface can be framed as a risk-minimization problem, leaving only standard modeling choices to fill in for $q$ in \eqref{eq:pinball_loss}.
Although initial efforts were focused on linear parametric quantile regressors \cite{original_qr_paper}, in recent years, modeling efforts have focused on nonparametric varieties.
\cite{qr_forests} adapted random forests to quantile regression.
\cite{nonpara_qr} leveraged Reproducing Kernel Hilbert Spaces to construct smooth quantile regressors.
Closer to our application, \cite{waldo_supp} used neural networks to optimize the pinball loss to learn the quantile surface for their application.

\section{Additional details and illustrations in \Cref{sec:numerical_exp}}
\label{sec:additional-details-sec:numerical_exp}

\subsection{Importance-like sampler for the three-dimensional example in \Cref{sec:3d_constrained_gaussian}}
\label{app:sec:importance_like_sampler}

Since the parameter settings with quantiles meaningfully larger than $\chi^2_{1, \alpha}$ are located close to the constraint boundary, a sampling challenge is presented.
Using the samplers as described in \Cref{subsec:preimage_sampling} results in under-sampling of this large-quantile region since both samplers provide uniform random samples over the \bergerboos set.
\Cref{alg:3d_sampler} presents a modified version of \Cref{alg:max_q_rs}, an importance-like sampler to increase the probability mass of samples close to the constraint boundary.
Note, we say ``importance-like'' because we do not provide any theoretical guarantee regarding this sampler's ability to produce draws from a particular target distribution.
We tailored \Cref{alg:3d_sampler} to settings with a non-negativity constraint and hand-tuned the length scale parameter to the particular three-dimensional example in \Cref{sec:3d_constrained_gaussian}.

\begin{algorithm}[!ht]
\caption{Importance-like sampler for three-dimensional constrained Gaussian}
\label{alg:3d_sampler}
\begin{algorithmic}[1]
\REQUIRE Number of samples: $M \in \mathbb{N}$, inverse length scale: $\gamma_p$, and order of norm: $q \in (0, 1)$.
\vspace{0.35em}
\STATE Instantiate a list $\mathcal{S}$ of length $M$ to store sampled points.
\WHILE{$\lvert \mathcal{S} \rvert < M$}
    \STATE Draw $M - \lvert \mathcal{S} \rvert$ realizations from \Cref{alg:polytope_sampler}: $\tilde{\bx}_1, \dots, \tilde{\bx}_{M - \lvert \mathcal{S} \rvert}$
    \FOR{$k = 1, \dots, M - \lvert \mathcal{S} \rvert$:}
    \STATE Compute the probability of accepting the $k$-th draw: $p_k := \exp\left(-\gamma_p \lVert \tilde{\bx}_k \rVert_q \right)$.
    \STATE Draw $z_k \sim \text{Bernoulli}(p_k)$.
    \IF{$z_k = 1$}
        \STATE $\mathcal{S}[k] \leftarrow \tilde{\bx}_k$ 
    \ENDIF
    \ENDFOR
\ENDWHILE
\vspace{0.15em}
\ENSURE Sampled parameters in \bergerboos set $\mathcal{S}$.
\end{algorithmic}
\end{algorithm}

The key to \Cref{alg:3d_sampler} is the additional accept/reject step where the $k$-th sample is accepted with probability $p_k$.
Additionally, by setting $q \in (0, 1)$, we prioritize retaining samples closer to the non-negativity boundary.
This effect can be seen in \Cref{fig:3d_importance_sampler}, where the vast majority of samples are found closer to the constraint boundary.
The ability of \Cref{alg:3d_sampler} to better sample the high-quantile regions of the \bergerboos set can be seen in \Cref{fig:importance_vs_polytope_sampler}.
In particular, for the functional values $\mu \in (-4, -2)$, the importance-like sampler is substantially more effective than the Polytope sampler at finding parameter settings with $\gamma$-quantiles greater than $1.1$.

\begin{figure}[p]
  \includegraphics[width=0.95\textwidth]{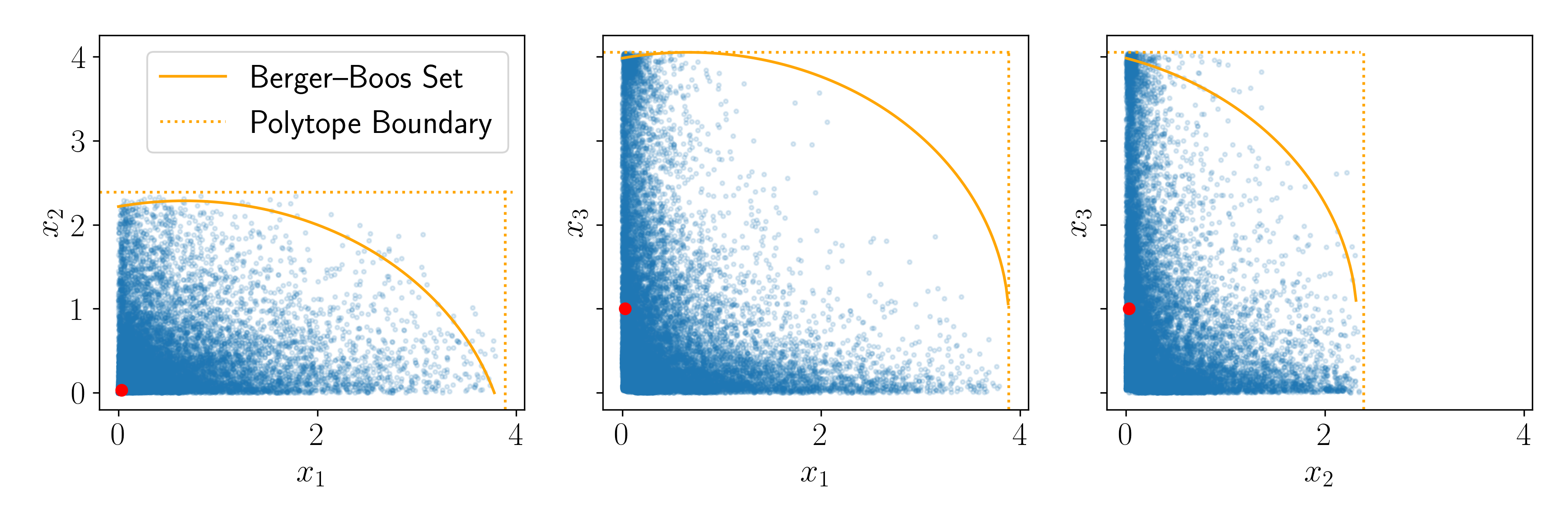}
  \caption[Importance-like sampler points near the constraint boundary]{When sampling points using \Cref{alg:3d_sampler}, the vast majority of samples are found closer to the non-negativity constraint boundary. The true parameter setting is shown by the red point, while the parameter settings sampled by \Cref{alg:3d_sampler} are shown by the blue points. This sampling prioritization helps adequately sample the regions of the \bergerboos set where the quantile surface is larger than $\chi^2_{1, \alpha}$. Furthermore, the vast majority of sampled points lie within the \bergerboos set, with some lying outside within the bounding polytope.}
  \label{fig:3d_importance_sampler}
\end{figure}

\begin{figure}[p]
  \includegraphics[width=0.95\textwidth]{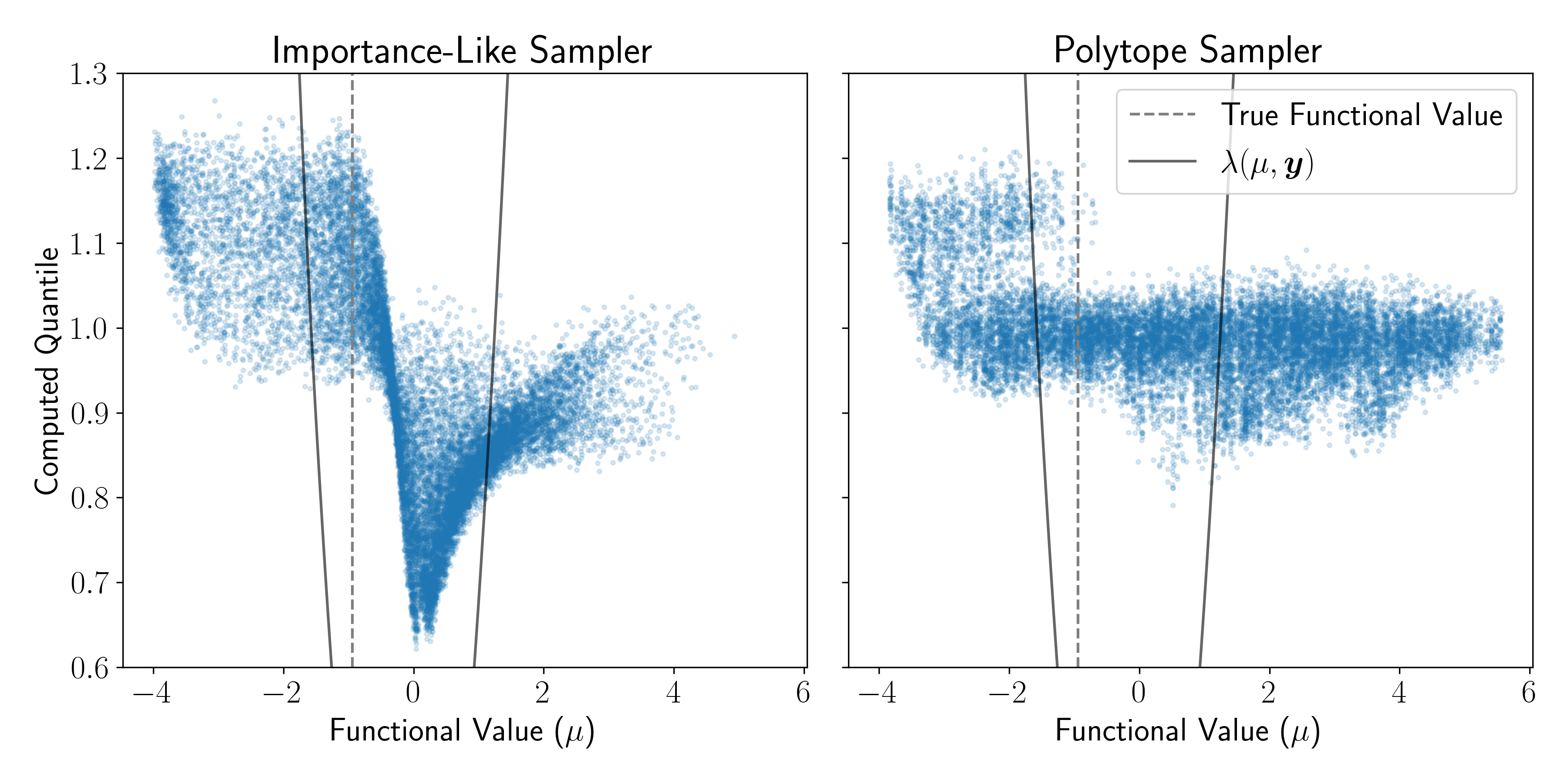}
  \caption[Importance-like sampler versus Polytope sampler]{The importance-like sampler described by \Cref{alg:3d_sampler} is more effective than the Polytope sampler described by \Cref{alg:polytope_sampler} at sampling parameter settings with $\gamma$-quantile greater than $1.1$. Each parameter setting sampled by \Cref{alg:3d_sampler} is shown by a blue point. This improved ability helps ensure the coverage guarantee shown in the left panel of \Cref{fig:3d_coverage_length}.}
  \label{fig:importance_vs_polytope_sampler}
\end{figure}

\end{document}

%% file: shared.tex
\input{macros}
\usepackage[left=1in,top=1in,right=1in,bottom=1in,head=.1in]{geometry}
\usepackage{amsfonts}
\usepackage{graphicx}
\usepackage{epstopdf}
\usepackage{algorithmic}
\ifpdf
  \DeclareGraphicsExtensions{.eps,.pdf,.png,.jpg}
\else
  \DeclareGraphicsExtensions{.eps}
\fi

\usepackage{enumitem}
\setlist[enumerate]{leftmargin=.5in}
\setlist[itemize]{leftmargin=.5in}

\newsiamremark{remark}{Remark}
\newsiamremark{hypothesis}{Hypothesis}
\crefname{hypothesis}{Hypothesis}{Hypotheses}
\newsiamthm{claim}{Claim}
\newsiamremark{fact}{Fact}
\crefname{fact}{Fact}{Facts}

\headers{Confidence intervals for functionals in constrained inverse problems}{M. Stanley et~al.}

\title{Confidence intervals for functionals in constrained inverse problems via data-adaptive sampling-based calibration}

\newcommand{\titletext}{Confidence intervals for functionals in constrained inverse problems via data-adaptive sampling-based calibration}

\author{Michael Stanley\thanks{Analytical Mechanics Associates, USA}
\and Pau Batlle%\thanks{Department of Computing and Mathematical Sciences, California Institute of Technology, USA}
\and Pratik Patil\thanks{Department of Statistics and Data Sciences, University of Texas at Austin, USA}
\and Houman Owhadi\thanks{Department of Computing and Mathematical Sciences, California Institute of Technology, USA}%\footnotemark[2]
\and Mikael Kuusela\thanks{Department of Statistics and Data Science, Carnegie Mellon University, USA}
}

\usepackage{amsopn}

%% file: macros.tex
\usepackage{bm}

\usepackage{cleveref}

\newcommand\myparagraph[1]{\textbf{#1}}

\usepackage{algorithm}
\usepackage{algorithmic}

\usepackage{booktabs}

\usepackage{float}
\usepackage{ragged2e}
\usepackage{tabularx}

\newcolumntype{Y}{>{\centering\arraybackslash}X}
\newcolumntype{P}{>{\raggedleft\arraybackslash}X}
\usepackage{multirow}
\newcolumntype{C}[1]{>{\centering\arraybackslash}m{#1}}
\newcolumntype{L}[1]{>{\raggedright\arraybackslash}m{#1}}

\renewcommand{\tilde}{\widetilde}
\renewcommand{\hat}{\widehat}

\usepackage{colortbl}

\usepackage[font=small,skip=5pt]{caption}

\allowdisplaybreaks

\newcommand{\RR}{\mathbb{R}}

\newcommand{\cC}{\mathcal{C}}

\newcommand{\cI}{\mathcal{I}}

\newcommand{\cN}{\mathcal{N}}
\newcommand{\cO}{\mathcal{O}}

\newcommand{\cX}{\mathcal{X}}

\newcommand{\ba}{{\bm{a}}}
\newcommand{\be}{{\bm{e}}}
\newcommand{\bb}{{\bm{b}}}
\newcommand{\bp}{{\bm{p}}}

\newcommand{\bw}{{\bm{w}}}
\newcommand{\by}{{\bm{y}}}
\newcommand{\bx}{{\bm{x}}}
\newcommand{\bz}{{\bm{z}}}

\newcommand{\bA}{{\bm{A}}}

\newcommand{\bK}{{\bm{K}}}
\newcommand{\bh}{{\bm{h}}}
\newcommand{\bI}{{\bm{I}}}

\newcommand{\bP}{{\bm{P}}}

\newcommand{\bSigma}{\bm{\Sigma}}
\newcommand{\bOmega}{\bm{\Omega}}

\newcommand{\bepsilon}{\bm{\varepsilon}}

\newcommand{\pto}{\xrightarrow{\textup{p}}}

\definecolor{revisionburntorange}{RGB}{191,87,0}
\def\BB{\mathsf{BB}}

\newcommand{\opt}{\mathrm{opt}}
\newcommand{\inv}{\mathrm{inv}}
\newcommand{\tr}{\mathrm{tr}}
\newcommand{\gl}{\mathrm{gl}}
\renewcommand{\sl}{\mathrm{sl}}

\newcommand{\VGS}{\mathrm{VGS}}

\let\bar\overbar
\let\hat\widehat
\let\tilde\widetilde

\def\aosbQ{\bar{q}_{\gamma, \eta}}
\def\aosbQlocal{\bar{q}_{\gamma, \eta}(\mu)}
\def\bbset{f^{-1}(\Gamma_\eta(\by)) \cap \mathcal{X}}
\def\maxqRS{\bar{q}_{\gamma, \eta}^{\mathrm{de}}}

\def\localQ{Q_{\mu, 1 - \alpha}^{\mathrm{max}}}
\def\globalQ{Q_{1 - \alpha}^{\mathrm{max}}}
\def\localcs{C_\alpha^{\sl}(\by; \mathcal{B}_\eta)}
\def\globalcs{C_\alpha^{\gl}(\by; \mathcal{B}_\eta)}

\newcommand{\convprob}[2]{#1 \pto #2}

\usepackage{xspace}

\usepackage[xspace]{ellipsis}
\usepackage{sidecap}

\newcommand{\bergerboos}{Berger--Boos\xspace}